\documentclass[10pt]{article}

\usepackage{xcolor} 
\newcommand{\Cr}[1]{\textbf{\textcolor{red}{{#1}}}}
\newcommand{\Cb}[1]{\textbf{\textcolor{blue}{{#1}}}}
\newcommand{\Cbr}[1]{\textbf{\textcolor{brown}{{#1}}}}

\newcommand{\Cp}[1]{\textcolor{violet}{{#1}}}

\usepackage{enumerate}

\usepackage{hyperref}

\usepackage{amsmath}
\usepackage{amsthm}
\usepackage{amssymb}
\usepackage{amsfonts}
\usepackage{bbm}
\usepackage{nicefrac}
\usepackage{mathtools}
\usepackage{soul}

\usepackage[title, toc]{appendix}

\usepackage{tikz}
\usepackage{color}
\usepackage{graphicx}
\usepackage{fullpage}
\usepackage{float}
\usepackage{wrapfig}
\usepackage{caption}
\usepackage{subcaption}
\usepackage{listings}
\usepackage{verbatim}
\numberwithin{equation}{section}

\let\Horig\H

\DeclareMathOperator{\E}{\mathbb{E}}
\DeclareMathOperator{\Var}{Var}
\DeclareMathOperator*{\diag}{diag}

\DeclareMathOperator*{\im}{Im}

\newcommand{\beq}{ \begin{equation} }
\newcommand{\eeq}{ \end{equation} }
\newcommand{\beqq}{ \begin{equation*} }
\newcommand{\eeqq}{ \end{equation*} }

\newcommand{\ii}{\mathrm{i}}

\newcommand{\bhp}[1]{\mathcal O\left(#1\right)}

\newcommand{\RN}[1]{%
  \textup{\uppercase\expandafter{\romannumeral#1}}%
}

\def \eqids {\overset{\mathcal{D}}{=}}
\def \eqidsgibbs {\overset{\mathfrak{D}}{=}}

\def \simeqids {\stackrel{\mathcal{D}}\simeq}
\def \simeqidsgibbs {\stackrel{\mathfrak{D}}\simeq}

\def \P {\mathbb{P}}
\def \R {\mathbb{R}}

\def \dd {\mathrm{d}}

\def \NN {\mathcal{N}}

\def \g {\gamma}
\def \e {\varepsilon}
\def \d {\delta}

\def \s {\sigma}

\newcommand{\rma}[1]{\mathcal{#1}}
\newcommand{\gib}[1]{\mathfrak{#1}}

\def \ham {\rma H}
\def \tmp {T}

\def \sGOE {M}
\def \eg {\lambda}

\def \pat {\rma Z}
\def \fe {\rma F}
\def \felim {F}
\def \ef {h}
\def \efv {\mathbf{g}}
\def \efve {g}
\def \efres {H}

\def \eigvm {O}
\def \cp {\gamma}
\def \scl{\dd \sigma_{scl}(x)}
\def \crv{\rma E_N}
\def \crvlim{\rma E}

\def \airy {\alpha}
\def \slim {\varsigma}
\def \twgoe  {\TW_{GOE}}
\def \egres {a}
\def \G {\rma G}
\def \ffl {\widetilde{\rma F}}
\def \mn {\alpha}
\newcommand{\go}{\rma S}
\newcommand{\lsl}{\rma L}
\def \mage {\mathcal M}
\def \mgn {\gib M}
\def \scp {\rma X}
\def \Gmgn {\rma G_{\mgn}}
\def \cpmgn {\cp_{\mgn}}
\def \vu {\mathbf{u}} 
\def \OM {\gib O}
\def \OG {\gib G}
\def \OP {\Omega}
\def \ve {\mathbf{e}}
\def \Gom {\rma G_\OM}
\def \cpm {\cp_\OM}
\def \evg {\nu}
\def \scpz {\scp^0}
\def \scpo {\scp^1}

\def \cO {\mathcal{O}}

\def \PP {\mathbb{P}}

\newtheorem{theorem}{Theorem}[section]

\newtheorem{lemma}[theorem]{Lemma}
\newtheorem{conjecture}[theorem]{Conjecture}

\newtheorem{definition}[theorem]{Definition}

\newtheorem{result}[theorem]{Result}

\theoremstyle{remark}
\newtheorem{remark}[theorem]{Remark}

\DeclareMathOperator{\sphve}{\sigma}

\DeclareMathOperator{\TW}{TW}

\def \ovl {\gib R}
\def \Govl {\rma G_{\ovl}}
\def \cpovl{\cp_{\ovl}}
\def \sovl{q_{\ovl}}
\DeclareMathOperator{\Bern}{\gib B} 

\def \bpm {b}
\def \OMz {\OM^0}

\def \ovlz {\ovl^0}

\def \sigmaovl {\sigma_{\ovl}}

\def \mgnz {\mgn^0}
\def \vecv {\mathbf{v}}

\def \stild {t} 
\def \sumws {\Upsilon}
\def \so {p} 
\def \sgr {r}

\def \Fos {F_0}
\def \chid {\scp_d}
\def \sva {\sigma_s}
\def \chimicro {\scp_{\text{micro}}}
\def \chidmic {\scp_{d,\text{micro}}}
\def \chidma {\scp_{d}^0}

\def \meso {\text{meso}}
\def \micro {\text{micro}}

\def \sigmamgn {\sigma_{\mgn}}
\def \sigmaom {\sigma_{\OM}}
\def \hefv {\hat{\efv}}

\def \mv {\mathbf{v}}
\def \spin {\boldsymbol{\sigma}}
\def \sphv {\spin}
\def \hsigma {\hat{\spin}}
\def \hmv {\hat{\mv}}

\begin{document}

\title{Spherical spin glass model with external field}
\author{Jinho Baik\footnote{Department of Mathematics, University of Michigan,
Ann Arbor, MI, 48109, USA, \texttt{baik@umich.edu}} \and Elizabeth Collins-Woodfin\footnote{Department of Mathematics, University of Michigan,
Ann Arbor, MI, 48109, USA, \texttt{elicolli@umich.edu}} \and 
Pierre Le Doussal\footnote{Laboratoire de Physique de l'Ecole Normale Sup\'erieure, ENS, Universit\'e PSL, CNRS, Sorbonne Universit\'e, Universit\'e de Paris, 75005 Paris, France,\texttt{ledou@lpt.ens.fr}} 
\and 
Hao Wu\footnote{\texttt{lingluan@umich.edu}}}

\maketitle

\begin{abstract}
We analyze the free energy and the overlaps in the 2-spin spherical Sherrington Kirkpatrick spin glass model with an external field for the purpose of understanding the transition between this model and the one without an external field.  We compute the limiting values and fluctuations of the free energy as well as three types of overlaps in the setting where the strength of the external field goes to zero as the 
dimension of the spin variable grows. 
In particular, we consider overlaps with the external field, the ground state, and a replica.  
Our methods involve a contour integral representation of the partition function along with random matrix techniques.  We also provide computations for the matching between different scaling regimes.  Finally, we discuss the implications of our results for susceptibility and for the geometry of the Gibbs measure. 
Some of the findings of this paper are confirmed rigorously by Landon and Sosoe in their recent paper which came out independently and simultaneously. 
\end{abstract}

\setcounter{tocdepth}{1}
\tableofcontents

\section{Introduction}

\subsection{The model, definitions, and notation}


Spin glasses are disordered magnetic alloys \cite{BinderK_YoungAP_1986,Young1998} that provide a physical context for the development of various mathematical models with wide ranging applications, not only in physics, but also in computer science and other areas \cite{MR1026102}.  One of the most studied spin glass models is the Sherrington-Kirkpatrick (SK) model \cite{SherringtonKirkpatrick1975, Parisi1980, Talagrand2000, GuerraToninelli}, in which the spin variable $\spin$ is a random vector from the $N$-dimensional hypercube $\{-1,+1\}^N$.  In this paper, we focus on the continuous analog of SK - the spherical Sherrington-Kirkpatrick (SSK) model. 
The SSK model shares many properties with the SK model but is usually easier to analyze and thus allows us to obtain results that remain out of reach for the SK model.

Particular quantities of interest in the study of the SSK model are the free energy and overlaps in the presence of an external field.
In the absence of a field, 
the model exhibits, at large $N$, a transition to a spin glass phase at 
low temperature. In the presence of a field, this phase transition disappears.
However, there are interesting regimes when the field is scaled as a power of the dimension $N$.  Those transitional regimes (with respect to the external field) 
will be the focus of this paper. We compute the free energy as well as three types of overlaps, up to fluctuations, when $h$, the strength of the external field, converges to zero as the dimension, $N$, of the system grows.

For the SSK model, the spin variable $\spin=(\sigma_1, \cdots, \sigma_N)$ is in $S_{N-1}$, the sphere of radius $\sqrt{N}$ in $\R^N$:
\beqq
	S_{N-1}=\{ \spin \in \R^N : \|\spin\|=\sqrt{N}\}.
\eeqq
The 2-spin spherical Sherrington-Kirkpatrick (SSK) model with external field is defined by the Hamiltonian 
\beq
\label{eq:def_Hamilt}
\ham(\sphv) =  - \frac12 \sum_{i,j = 1}^{N} \sGOE_{ij} \sphve_i \sphve_j - \ef \sum_{i = 1}^N \efve_i \sphve_i 
	= - \frac 12\sphv \cdot \sGOE \sphv  - h  \, \efv\cdot \sphv 
\eeq 
for $\sphv\in S_{N-1}$, where $M$ and $\efv$ are respectively a random matrix and a random vector, specified below. 
The associated Gibbs measure is 
\beq\label{eq:Gibbsmeasure}
	p(\sphv)= \frac1{\pat_N} e^{-\beta \ham(\spin)} \quad \text{for $\sphv\in S_{N-1}$}
\eeq
where 
\beq
	\beta=1/T
\eeq
 denotes the inverse temperature. 
The partition function and the free energy per spin component are 
\beq
\label{eq:def_fe}
	\pat_N = \int_{S_{N - 1}} e^{-\beta \ham(\spin)} \dd \omega_N(\sphv) \quad\text{and} \quad \fe_N = \fe_N(T,h)= \frac{1}{N\beta} \log \pat_N,
\eeq 
where $\omega_N$ is the normalized uniform measure on $S_{N - 1}$. 
Since the disorder variables $\sGOE$ and $\efv$ are random, the Gibbs measure is a random measure, which we also call a thermal measure, 
and the free energy $\fe_N$ is a random variable.  
We are interested in the fluctuations of the free energy when $\ef\to 0$ as $N\to \infty$. 

We also consider the behavior of the spin variables taken from the Gibbs measure.
We focus on the following three particular overlaps.  
\begin{itemize}
\item (overlap with the external field) Define
\beq \label{eq:magntde}
	\mgn =\frac{\efv \cdot \sphv}{N}  .
\eeq
\item (overlap with the ground state) 
Let $\vu_1$ be a unit eigenvector corresponding to the largest eigenvalue of the disorder matrix $M$. The vectors $\pm\vu_1$ are the ground state in the absence of an external field, and we simply call them the ground states. Define 
\beq
	\OG  = \frac{|\vu_1\cdot \sphv|}{\sqrt{N}} \quad\text{and}\quad \OM=\OG^2
\eeq
\item (overlap with a replica) Let $\spin^{(1)}$ and $\spin^{(2)}$ be two independent spin variables from the Gibbs measure for the same sample (i.e. disorder variables $M_{ij}$ and $g_i$); $\spin^{(2)}$ is a replica of $\spin^{(1)}$. Define 
\beq
\mathfrak{R}=\frac{\spin^{(1)} \cdot\spin^{(2)}}N .
\eeq
\end{itemize}

The factors $N$ and $\sqrt{N}$ are included since $||\mathbf{u}_1||=1$, $\|\spin\|=\sqrt{N}$, and the expected value of $\|\efv\|^2=g_1^2+\cdots + g_N^2$ is $N$ (see below). 

The overlaps depend on the spin variable and also the disorder sample. Hence, there are two different expectations to consider. 
We consider the thermal (Gibbs) fluctuations of the overlaps for a given disorder sample. For some quantities, we also consider the sample-to-sample fluctuations of the thermal average.  
We denote the thermal (Gibbs) average for a given disorder sample by the bracket $\langle \cdot \rangle$. On the other hand, the sample-to-sample average of an observable $O$ is denoted by $\bar{O}$ or $\E_s[O]$.
For example, the thermal averages 
\beqq
	\mage =\langle \mgn \rangle \quad \text{and}\quad \scp =\frac1{\ef} \langle \mgn \rangle 
\eeqq
are called magnetization and susceptibility, respectively. Many of the results of this paper are about the thermal fluctuations of  overlaps 
for a given disorder sample, i.e. for a given quenched disorder. 

\subsection{Assumptions on disorder samples}

The disorder parameters in the Hamiltonian \eqref{eq:def_Hamilt} are chosen as follows.
We define 
\beqq
	\sGOE=(\sGOE_{ij})_{1\le i,j\le N}
\eeqq 
to be a disorder matrix given by a random symmetric matrix from the Gaussian orthogonal ensemble (GOE), which is a matrix ensemble whose probability is rotationally invariant.
For $i\le j$, the variables $\sGOE_{ij}$ are independent centered Gaussian random variables with variance $\frac1N(1 + \delta_{ij})$. 
By the symmetry matrix condition,  $\sGOE_{ij}= \sGOE_{ji}$ for $i>j$. 
We denote by 
\beq
	\lambda_1\ge\cdots\ge \lambda_N \quad \text{and} \quad \vu_1, \cdots, \vu_N
\eeq
the eigenvalues of the disorder matrix $\sGOE$ and corresponding unit eigenvectors. The GOE assumption implies that the eigenvalues and eigenvectors are independent of each other. 
The external field is given by the vector 
\beq
	\efv= (\efve_1, \efve_2, \cdots, \efve_N)^T, 
\eeq 
which we assume to be a standard Gaussian vector. 
The strength of the external field is denoted by a non-negative scalar $\ef$.
We also define
\beq
	n_i=\vu_i\cdot \efv,
\eeq
the overlap of the eigenvector and the external field. The external field and eigenvectors appear in the results and analysis of this paper only as this combination. 
The variables $\eg_i$ and $n_i$ are collectively called disorder variables. 
We call the joint realization of $\eg_i$ and $n_i$ a disorder sample throughout the paper.

Note that $(n_1, \cdots, n_N)$ is a standard Gaussian vector, whose entries are independent of the eigenvalues $\eg_1, \cdots, \eg_N$.
The analysis of this paper also applies, after some changes of formulas, to the case when $\efv=(1, \cdots, 1)^T$. However, we restrict to the Gaussian external field since the Gaussian assumption makes calculations simpler. 

\subsection{Summary of prior research}

The purpose of this paper is to study the case $\ef\to 0$ systematically including up to the fluctuation term for the free energy and the three overlaps.  Here we provide a survey of some of the existing research as it connects to our study.

The free energy for the Hamiltonian \eqref{eq:def_Hamilt} above when $\ef=0$ converges to a deterministic value which was computed by Kosterlitz, Thouless and Jones in \cite{kosterlitz1976spherical}. 
The Hamiltonian \eqref{eq:def_Hamilt} 
is the 2-spin case of the more general $p$-spherical spin glass model which includes interactions between multiple spin coordinates. 
The limit of the free energy for the general spherical spin glass models which also includes the external field is given by the Crisanti-Sommers formula \cite{crisanti1992sphericalp}.
This formula is the spherical version of the Parisi formula \cite{parisi1980sequence} for the spins in hypercubes. 
The Parisi formula and Crisanti-Sommers formula are proved rigorously by Talagrand in \cite{TalagrandSK, TalagrandSSK}. 
The result of Kosterlitz, Thouless and Jones shows that when $\ef=0$, there are two phases: the spin glass phase when $T<1$ and the paramagnetic phase when $T>1$. On the other hand, they argued that when $\ef>0$, assuming that the external field is uniform, there is no phase transition.

The subleading (in $N$) term of the free energy depends on the disorder and hence it describes the fluctuations of the free energy. 
For $\ef=0$ and $T>1$, the fluctuation term is of order $N^{-1}$ and has the Gaussian distribution. 
This is proved for both the hypercube case \cite{AizenmanLebowitzRuelle, FrohlochZegarlinski, CometsNeveu} and the spherical case \cite{BaikLeeSSK}.
For $\ef=0$ and $T<1$, for the Hamiltonian above, the fluctuation term is of order $N^{-2/3}$ and has the GOE Tracy-Widom distribution \cite{BaikLeeSSK}. 
 Chen, Dey, and Panchenko performed a similar calculation for the case with Ising spins where $h>0$ is of order 1 and $g$ is the vector of all $1$s.  In this case, they find \cite{ChenDeyPanchenko} that 
the fluctuation term is of order $N^{-1/2}$ and has the Gaussian distribution for all temperature.
They claim that similar results hold for the spherical case and our results confirm this claim using a different method.  We note that their result also holds 
for mixed $p$-spin with even degree terms.
Chen and Sen \cite{ChenSen} computed the ground state energy for spherical mixed $p$-spin models (of which SSK is a specific case) and found that the fluctuations of the ground state energy are Gaussian in the presence of an external field.
 
 In \cite{FyodorovleDoussal}, the large deviations of the free energy distribution
was obtained at $T=0$ from a non-rigorous saddle point calculation of the 
moments of ${\cal Z}_N$ in the large $N$ limit (see also 
\cite{ DemboZeitouni} for a rigorous version). From this calculation
a transitional regime $h \sim N^{-1/6}$ for the fluctuations
of the free energy was 
conjectured. A proof of the existence of this regime was 
obtained in \cite{kivimae2019critical}.
In the current paper, we obtain explicitly the
fluctuations of the free energy in the regime $h \sim N^{-1/6}$ for any $T<1$
and in the regime $h \sim N^{-1/4}$ for $T>1$. 
As we show, our results match in the tail of the distribution with those of \cite{FyodorovleDoussal}.
Note also the recent physics work \cite{FyodorovRashel} where a different spherical model of random optimization was considered, which exhibits a similar phenomenology.

\medskip

The overlap with the external field has been studied in the context of magnetism and susceptibility.  
Kosterlitz, Thouless, and Jones \cite{kosterlitz1976spherical} computed the susceptibility as $h$ tends to zero and observed a transition at the temperature $T=1$.  
Cugliandolo, Dean, and Yoshino \cite{cugliandolo2007nonlinear} computed two different versions of this limit of the susceptibility, in the first case taking $\lim_{h\to0}\lim_{N\to\infty}$ and in the second case taking $\lim_{N\to\infty}\lim_{h\to0}$.  In the first of these cases, they get the same result as \cite{kosterlitz1976spherical} with a transition at $T=1$, but in the second case they do not observe a transition.  Furthermore, they find that the two types of limits agree for $T>1$ but not for $T<1$.  
They also extend that results to a more general class of models (beyond Gaussian) and to non-linear susceptibility.  
We focus on the linear susceptibility and differential susceptibility in the Gaussian case, and obtain a more detailed picture. 
By considering the three regimes $h=O(1)$, $h\sim N^{-1/6}$, and $h\sim N^{-1/2}$, we see that the first limit considered by Cugliandolo et al agrees with our result for the $h\to0$ limit of the $h=O(1)$ case.  The second limit that they consider is analogous to our result for the $H\to0$ limit of the $h\sim N^{-1/2}$ case where we define $H=hN^{1/2}$.  However, we find in this case that the susceptibility depends on the sample and is a function of $\mathbf{g}\cdot \mathbf{u}_1$, the inner product of the external field and the ground state.  
This dependence was not apparent in \cite{cugliandolo2007nonlinear}, since their set-up fixes $\mathbf{g}\cdot\mathbf{u}_1=1$.  When $\mathbf{g}\cdot\mathbf{u}_1=1$ we find, as they do, that there is no transition in the susceptibility between high and low temperature.  However, a transition does exist for all other values of $\mathbf{g}\cdot \mathbf{u}_1$.

\medskip

The overlap with the ground state is relevant to understanding the geometry of the Gibbs measure.  Subag \cite{Subag_2017} examines the geometry of the Gibbs measure for general $p$-spin spherical models and finds that the Gibbs measure concentrates in spherical bands around the critical points of the Hamiltonian.  These bands are of the form $\text{Band}(\spin_0,q,q')=\{\spin\in S_{N-1}:q\leq R(\spin,\spin_o)\leq q'\}$ where $\spin_0$ is a critical point of $\ham$ and $R(\spin,\spin_0)$ is the overlap of $\spin$ and $\spin_0$.  We focus specifically on the overlap with the ground state (where $\spin_0$ is the critical point corresponding to the largest eigenvalue).  In the $h=0$ regime, as expected, we see the Gibbs measure concentrates in a band and we examine how this geometry changes for the case of positive constant $h$ as well as the cases of $h\sim N^{-1/6}$ and $h\sim N^{-1/3}$.

\medskip

The overlap with a replica has been  studied extensively, both for the Ising spin models and the spherical spin models with general $p$-spin interaction.
For $p=2$ the non-rigorous replica method 
used in \cite{kosterlitz1976spherical, crisanti1992sphericalp, FyodorovleDoussal} obtains a replica
symmetric saddle point leading to a prediction for the overlap $q$ as a function of $h$. In particular, at $h=0$, the prediction is that $q=1-T$ for $T<1$ and $q=0$ for $T>1$.
These calculations were confirmed rigorously in \cite{panchenko2007overlap}.
Recently, Landon, Nguyen, and Sosoe extended the results further  to examine the fluctuations of the overlap at high temperature \cite{nguyen2018central} and low temperature \cite{LandonSosoe}.  They find, in particular, that the overlap has Gaussian fluctuations in the high temperature regime, whereas, in the low temperature regime, the fluctuations 
are of order $N^{-1/3}$ and converge to a random variable that has an explicit formula in terms of the GOE Airy point process (see subsection \ref{sec:edgebehavior} for a description of this).  
In this paper, we obtain similar results for $h\sim N^{-1/6}$ and $h\sim N^{-1/2}$. 


\subsection{Method of analysis}

Our computations are based on contour integral representations which we present in Section \ref{sec:integralreps}.
The free energy and the moment generating functions of two of the overlaps can be expressed in terms of a single integral, whereas the moment generating function in the case of the overlap with a replica can be written as a double integral.  The integrand for each of these integral representations contains disorder variables and hence we have random integral formulas. 
The single integral formula for the free energy was first observed by Kosterlitz, Thouless and Jones \cite{kosterlitz1976spherical} and the authors use the method of the steepest descent to evaluate the limiting free energy.  
For the case of $\ef=0$, this calculation was extended in \cite{BaikLeeSSK} to find the fluctuation terms using the recent advancements in random matrix theory, in particular the rigidity results on the eigenvalues \cite{MR3098073} and the linear statistics \cite{JohanssonCLT, Bai2005, Lytova2009}. 
Similar ideas were also used in \cite{BaikLeeFerromagnetic,BaikLeeBipartite,BaikLeeWu}, including the case for the overlap with a replica in \cite{LandonSosoe}. 
This paper extends the integral formula approach to the case when $h=O(1)$ and $h\to 0$ in the transitional regimes. 
When there is an external field, the analysis becomes more involved. 
In this case, the dot products of the eigenvectors and the external field play an important role in the analysis. 

The steepest descent analysis of this paper can be made mathematically rigorous after some efforts using probability theory and random matrix theory. However, this paper will focus on computations and interpretations assuming that various estimates in the steepest descent analysis can be obtained.  We use the label ``Result" for findings in which we do not provide rigorous proofs and the label ``Theorem" for findings that we cite from prior papers that include rigorous proof.  We use the label ``Lemma" for short findings that we prove in full detail.

In a recent preprint \cite{landon2020fluctuations} which was obtained independently and simultaneously with this paper, Landon and Sosoe consider a similar SSK model in which the external field is a fixed vector and the disorder matrix has zero diagonal entries. 
Their work is mathematically rigorous and contains proofs of some of the results obtained in this paper, namely for the free energy and some aspects of the overlaps with the external field and with a replica in Subsections \ref{sec:freeenergypof}, \ref{sec:fe1/4analysis}, \ref{sec:fetranltmp}, \ref{sec:mgnefpos}, \ref{sec:ext1/6}, \ref{sec:ovlpos}, and \ref{sec:ovl1/6}.
After the completion of this paper, one of us, Collins-Woodfin, also proved the results in Subsection \ref{sec:mgnlowtmn12} on the overlap with the microscopic external field rigorously in \cite{CollinsWoodfin}.

\subsection{Organization of the paper}

The results of the calculations are scattered throughout the paper. In Section \ref{sec:results}, we present some of the highlights of the results of this paper. 
The single and double integral representation of the free energy and the generating functions of the overlaps are given in Section \ref{sec:integralreps}. 
The next three sections of the paper address the free energy.  Section \ref{sec:feh>0} summarizes known results for the $h=0$ case and explains our findings for the $h>0$ case.  Section \ref{sec:fe1/4} addresses the $h\to0$ case for $T>1$ and section \ref{sec:fe1/6} addresses $h\to0$ for $T<1$.  Sections \ref{sec:ext}, \ref{sec:ground}, and \ref{sec:2ol} provide our results for each of the three types of overlaps. 
Section 8 also provides our results for magnetization and susceptibility.  
Section \ref{sec:Geometry} describes the geometry of the spin vector configuration under the Gibbs measure. We include as appendices the proof of the contour integral formulas and also a perturbation lemma. 


\subsection*{Acknowledgements}

The authors would like to thank Benjamin Landon and Philippe Sosoe for sharing their recent work with us.  
The work of Baik was supported in part by the NSF grants DMS-1664692 and DMS-1954790.  The work of Collins-Woodfin was supported in part by the NSF grants DMS-1701577 and DMS-1954790.  The work of Le Doussal was supported in part by the ANR grant ANR-17- CE30-0027-01 RaMaTraF.
Le Doussal would like to thank the Department of Mathematics of the University of Michigan for hospitality; this joint project started during his visit to the department.

\section{Highlights of the results}\label{sec:results}

\subsection{Results for the free energy}
We examine the behavior of the free energy, including its leading order and the sample-to sample fluctuation term, as $N\to\infty$ when $h=O(1)$ and when $h\to0$. 
We find that, in each case,
\beq
\fe_N (\tmp, h) \simeqids \felim(\tmp, h) + \text{ sample fluctuations}
\eeq
where $\simeqids$ denotes an asymptotic expansion in distribution with respect to the disorder variables. 
The limiting free energy $\felim(T,h)$ includes all deterministic (depending only on $h$ and $T$) terms whose order exceeds that of the sample fluctuations.  The ``sample fluctuations" refers to the largest order term that depends on the disorder sample.  Our findings in each case are summarized in Table \ref{tab:fesummary}.  Upon computing the leading term and sample fluctuations for $\fe_N(T,h)$ with $h=O(1)$, we made two key observations.  Firstly, the free energy for $h=O(1)$ does not exhibit a transition as we see in the $h=0$ case; this observation is consistent with the result of \cite{ChenDeyPanchenko} for Ising spins.  Secondly, while the limiting free energy is continuous in $T$ and $h$,  the sample fluctuations in the $h=O(1)$ case do not agree with those in the $h=0$ case (neither for $T>1$ nor for $T<1$).  This suggests the existence of transitional regimes.  We found that, for $T>1$, the transition occurs at $h\sim N^{-1/4}$ while, for $T<1$, the transition occurs at $h\sim N^{-1/6}$.  We computed the asymptotic expansion of $\fe_N(T,h)$ in these transitional regimes. 

\begin{table}[H]
\renewcommand{\arraystretch}{1.5}
\centering
\begin{tabular}{l|l|l|c}
Case & Limiting free energy $\felim(T,h)$ & Sample fluctuations&Result\\
\hline
$h=0,\;T>1$ & $\frac{1}{4T}$ &  $N^{-1}$ Gaussian distribution&  \ref{thm:freehigh} \\ 

$h=0,\, T<1$ & $1 - \frac{3\tmp}{4} + \frac {\tmp\log \tmp} 2$ &  $N^{-\frac 23}$ $\twgoe$ distribution&  \ref{thm:freelow} \\ 

$h=O(1)$ & $\frac{\cp_0}{2} - \frac{\tmp s_0(\cp_0)}{2} - \frac{\tmp -\tmp\log \tmp}{2}  + \frac{\ef^2 s_1(\cp_0)}{2}$ &  $N^{-\frac 12}$ Gaussian distribution& \ref{thm:feh>0}\\

$h\sim N^{-\frac 14},\;T>1$ & $\frac{1}{4T}+\frac{h^2}{2T}$ &  $N^{-1}$ Gaussian distribution& \ref{thm:fe1/4T>1}\\

\vspace{-0.1in}
$h\sim N^{-\frac 16},\, T<1$ & $1 - \frac{3\tmp}{4} + \frac {\tmp\log \tmp} 2+\frac{h^2}{2}$ &  $N^{-\frac 23}$ function of the GOE Airy& \ref{thm:fe1/6T<1}\\
&& point process and Gaussian r.v.'s&\\ 
\end{tabular}
\caption{This table summarizes our findings for the leading term and fluctuations of $\fe_N(T,h)$ in the various cases we considered.  The $h=0$ cases were already known \cite{BaikLeeSSK} but are included here for completeness.  In the limiting free energy for the $h=O(1)$ case, the quantity $\gamma_0$ is deterministic and depends only on $T$ and $h$.  The functions $s_0$ and $s_1$ are defined in section \ref{sec:RMT}.  For more details on the notation, derivation, and precise formulas for the fluctuation terms, see the corresponding result.}
\label{tab:fesummary}
\end{table}

When comparing the fluctuations in each regime, we observe that the order of the fluctuations are largest in the $h=O(1)$ case, where they have order $N^{-1/2}$ and Gaussian distribution.  This holds for all temperatures.  When $T>1$ but $h=0$ or $h\to0$, the fluctuations remain Gaussian, but their order shrinks to $N^{-1}$.  When $T<1$ and $h=0$ or $h\to0$, the fluctuations have order $N^{-2/3}$.  In the case of $h=0$ they have GOE Tracy-Widom distribution while, in the case of $h\sim N^{-1/6}$, their distribution is a function of the GOE Airy point process and of a sequence of i.i.d. standard Gaussian random variables.  See Table \ref{tab:fesummary} for the equations corresponding to each of these results.

\subsection{Results for the overlaps}

In the next three tables we state our findings for the overlap with the external field, with the ground state and with a replica. 
In each case the thermal average and thermal fluctuations are presented in interesting regimes of $h$ and $T$. 
The thermal average and fluctuations in most cases depend on the disorder sample. 
Our findings also have implications for magnetization and susceptibility, which will be described in more detail in section \ref{sec:ext}.

\begin{table}[H]
\renewcommand{\arraystretch}{1.7}
\centering
\begin{tabular}{l|l|l|c}
Case & Thermal average $\langle\mgn\rangle$ & Thermal fluctuations of $\mgn$ &Result\\
\hline
\vspace{-0.15in}
$h=O(1)$ for all $T$ & $hs_1(\gamma_0)+O(N^{-\frac12})$ & $N^{-\frac12}$ Gaussian & \ref{thm:exth>0}\\
(and $h=0,\, T>1$) &&& \ref{result:extfld0} \\
$h\sim N^{-\frac16},\, T<1$ & $h+O(N^{-\frac12})$ & $N^{-\frac12}$ Gaussian & \ref{thm:ext1/6}\\
\vspace{-0.15in}
$h\sim N^{-\frac12},\, T<1$ & $h +\frac{|n_1|\sqrt{1-T}}{\sqrt{N}}\tanh \big(\frac{|n_1|h\sqrt{N(1-T)}}{T} \big)  $ & $N^{-\frac12}$ [Gaussian + Bernoulli
] & \ref{thm:ext1/2}\\
(and $h=0,\, T<1$) & &&  \ref{result:extfld0} \\
\end{tabular}
\caption{This table summarizes our finding for $\mgn$, the overlap with the external field. Here, $\cp_0=\cp_0(h,T)$ in the first row is  deterministic. The variable $n_1$ in the third row is $n_1=\vu_1\cdot \efv$.  
For the top two rows, the leading term in $\langle\mgn\rangle$ and the thermal fluctuations of $\mgn$ do not depend on the disorder sample.
However, the $O(N^{-\frac12})$ subleading terms in $\langle\mgn\rangle$ for the top two cases and both the leading term in $\langle\mgn\rangle$ and the thermal fluctuations of $\mgn$ of the last row do depend on the disorder sample. }
\label{tab:extsummary}
\end{table}

\begin{table}[H]
\renewcommand{\arraystretch}{1.7}
\centering
\begin{tabular}{l|l|l|c}
Case & Thermal average $\langle\mathfrak{G}^2 \rangle$ & Thermal fluctuations of $\mathfrak{G}^2$ &Result\\
\hline
\vspace{-0.15in}
$h=O(1)$ for all $T$ & $\frac1{N}\left(\frac{h^2n_1^2}{(\gamma_0-2)^2}+\frac{T}{\gamma_0-2}\right)$ & $N^{-1}$ $\chi$-squared (non-centered) &  \ref{thm:groundh>0}\\
(and $h=0,\, T>1$) &&&  \ref{result:grounds0} \\
$h\sim N^{-\frac16},\, T<1$ & $1-T-  \sum_{i=2}^N\frac{ n_i^2 h^2N^{1/3}}{(\stild+a_1-a_i)^2}  $ & $N^{-\frac16}$ Gaussian & \ref{thm:ground1/6}\\
\vspace{-0.15in}
$h\sim N^{-\frac13},\, T<1$ & $1-T + O(N^{-\frac13})$ & $N^{-\frac13}$ r.v. that depends on disorder & \ref{thm:ground1/3} \\
(and $h=0,\, T<1$) & & &  \ref{result:grounds0} \\
\end{tabular}
\caption{This table summarizes our finding for $\mathfrak{G}^2=\OM$, the squared overlap with the ground state.  Here $n_i=\mathbf{u}_i\cdot\mathbf{g}$ and $a_i=N^{2/3}(\lambda_i-2)$. The quantity $\cp_0$ in the top row is the same term from Table \ref{tab:extsummary}. 
In the second row, the variable $\stild$ and the total sum, which is $O(1)$, depends  
on the disorder sample. 
All leading and subleadings terms of $\langle \OG^2\rangle$, and the thermal fluctuations of $\OG^2$, except the leading term, $1-T$, of $\langle \OG^2\rangle$ in the last row, depend on the disorder sample.}
\label{tab:groundsummary}
\end{table}

\begin{table}[H]
\renewcommand{\arraystretch}{1.7}
\centering
\begin{tabular}{l|l|l|c}
Case & Thermal average $\langle\mathfrak{R}\rangle$ & Thermal fluctuations of $\mathfrak{R}$ &Result\\
\hline
\vspace{-0.15in}
$h=O(1)$ for all $T$ & $h^2s_2(\gamma_0)+O(N^{-\frac12})$ & $N^{-\frac12}$ Gaussian &  \ref{result:replica1} \\
(and $h=0,\, T>1$) & & &   \ref{result:replica0}  \\
$h\sim N^{-\frac16},\, T<1$ & $1-T + O(N^{-\frac13})$ & $N^{-\frac13}$ r.v. that depends on disorder &  \ref{result:replica16}\\
\vspace{-0.15in}
$h\sim N^{-\frac12},\, T<1$ & $ (1-T)\tanh^2 \big(\frac{|n_1|h\sqrt{N(1-T)}}{T} \big) $ & $O(1)$ Bernoulli &  \ref{result:ovl12mainr} \\
(and $h=0,\, T<1$) & & &  \ref{result:replica0} \\
\end{tabular}
\caption{This table summarizes our finding for $\mathfrak{R}$, the overlap between two independent spins.  
The quantity $\gamma_0$ is the same term from Table \ref{tab:extsummary} and $n_1=\vu_1\cdot \efv$. 
The subleading terms of $\langle \mathfrak{R} \rangle$ in the top two rows and the leading term of  $\langle \mathfrak{R} \rangle$  in the third row depend on the disorder sample. The thermal fluctuations of $\mathfrak{R}$ also depend on the disorder sample for the bottom two rows.}
\label{tab:2olsummary}
\end{table}

\subsection{Geometry of the Gibbs measure}\label{sec:gibbsgeometry}

The results for the overlaps give us information on the geometry of the spin configuration under the Gibbs measure, some of which we summarize here. 
Recall that the spin configuration is parameterized by the vector $\spin=(\sigma_1, \cdots, \sigma_N)$ which belongs to the $N-1$ dimensional sphere of radius $\sqrt{N}$ and we consider the limit of large $N$.
At high temperature, $T>1$, the spin vector $\spin$ is nearly orthogonal to the ground state $\pm \vu_1$ when $h=0$.
For $h=O(1)$, the spin vector concentrates on the intersection of the sphere and the single cone around the vector $\efv$. 
The leading term of the cosine of the angle between the spin and the external field $\efv$ depends on the temperature and the field but not on the disorder sample, and, as one can expect, is an increasing function of the field. See Figure \ref{fig:hs1cp0} (a). This implies that as the field becomes stronger, the cone becomes narrower. There are no transitions between $h=0$ and $h=O(1)$.

\begin{figure}[H]
\centering
\begin{subfigure}{0.45\textwidth}
\includegraphics[width = 0.9\textwidth]{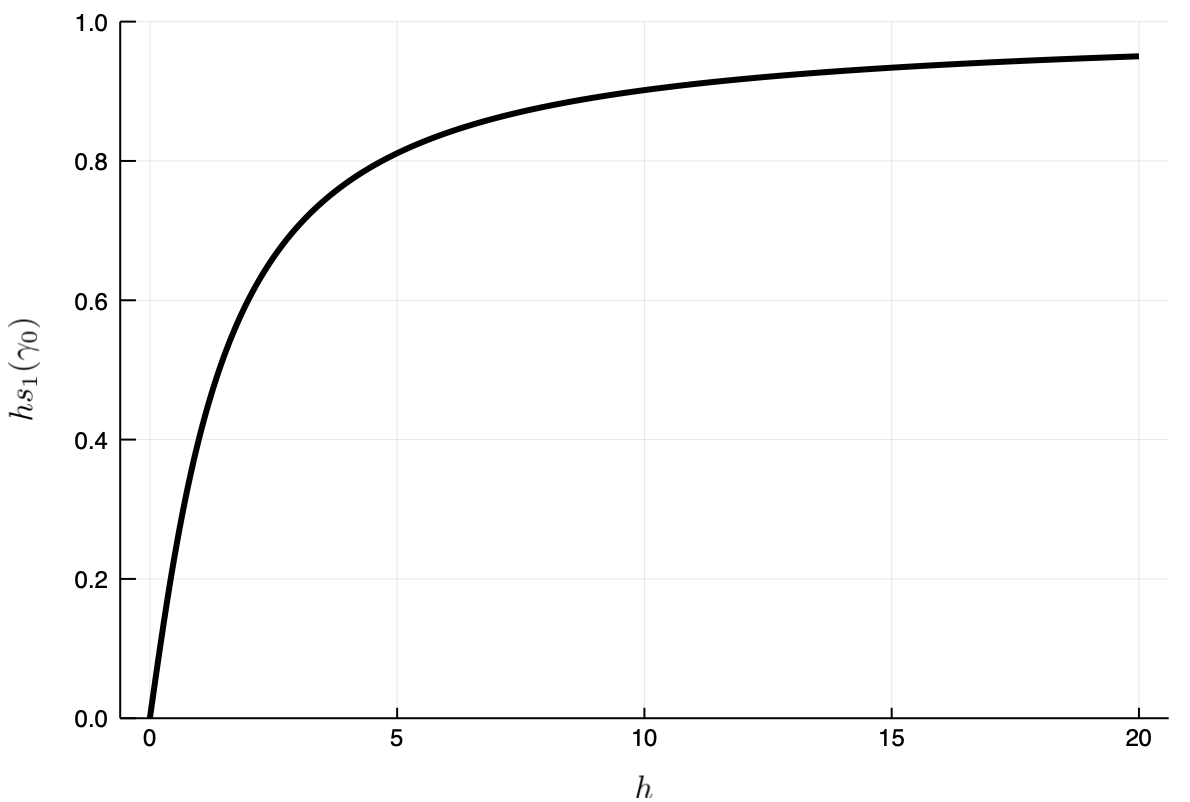}
\caption{$T  = 2$}
\end{subfigure}
\begin{subfigure}{0.45\textwidth}
\includegraphics[width = 0.9\textwidth]{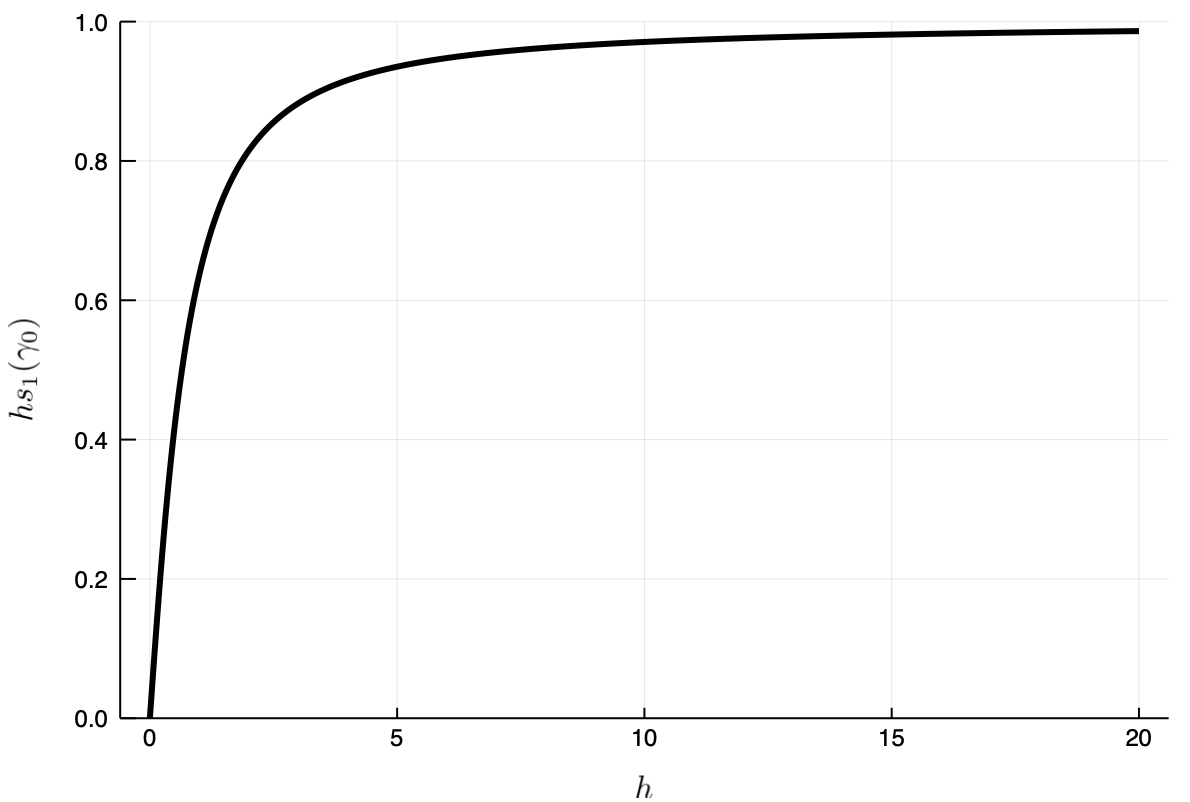}
\caption{$T =  0.5$}
\end{subfigure}
\caption{These are plots of the leading term of the angle between the spin and $\efv$. The formula is given by $\mgn^0$ in Subsection \ref{sec:mgnhplt}. The function depends only on $T$ and $h$. (a) $\tmp = 2$, (b) $\tmp = 0.5$.}
\label{fig:hs1cp0}
\end{figure}

Now consider the low temperature regime $0<T<1$. When $h=0$, the spins are concentrated on the intersection of the sphere with the double cone around the ground state $\pm \vu_1$ such that the leading term of the cosine of the angle is $\sqrt{1-T}$. 
This angle was found in \cite{kosterlitz1976spherical, crisanti1992sphericalp, FyodorovleDoussal} and in particular, \cite{panchenko2007overlap} showed that spins are distributed uniformly on the intersection of this double cone with the sphere. 
Consider increasing the external field strength $h$. When $h=O(1)$, the spin vector concentrates on the intersection of the sphere and the single cone around the vector $\efv$ just like the high temperature case. See Figure \ref{fig:hs1cp0} (b), which is qualitatively same as Figure (a). 
However now between $h=0$ and $h=O(1)$, there are two interesting transitional regimes, $h\sim N^{-1/2}$, which we call the microscopic regime, and $h\sim N^{-1/6}$, the mesoscopic regime.

In the microscopic regime, $h\sim N^{-1/2}$, at low temperature $0<T<1$, 
the results of this paper lead us to the Conjecture \ref{conj:signedoverlap}, which implies that the double cone 
becomes polarized into a single cone. 
The spin vector prefers the cone which is closer to $\efv$ to the other cone by the 
\beqq
	\text{$e^{\frac{2h\sqrt{N} |n_1|\sqrt{1-T}}{T}}$ to $1$ probability ratio.}
\eeqq
The spin vector is more or less uniformly distributed on the cones.
In this regime, the response of the spin to the field is the sum of (i) a linear response in the direction
transverse to $\pm \vu_1$ (i.e. along the cones) and, (ii)
the response of an effective 2-level system, which may be modeled
as a single one-component effective Ising spin $\frac{\spin}{\sqrt N} = \pm S \vu_1$
of size $S=\frac{|n_1| \sqrt{1-T}}{\sqrt{N}}$ with energy scale $E=N h S = \sqrt{N} h |n_1| \sqrt{1-T}$
(leading to a mean magnetization $S \tanh (E/T)$). Note that both $S$ and $E$ are sample dependent,
but depend only on $|n_1|$, the overlap of the ground state and the field.

%
%
%

For $h \sim N^{-1/6}$, progressively all eigenvectors and eigenvalues become important. 
In this regime, the spins are concentrated on the intersection of the sphere and a single cone around the ground state, but the cone depends on the disorder sample. 
The cosine of the angle between the spin and $\vu_1$ changes from $\sqrt{1-T}$ to a function which depends on all eigenvalues $\eg_i$ and the overlaps $n_i=\vu_i\cdot \efv$ of the eigenvectors and the external field. 
Furthermore, the spins are no longer uniformly distributed on the cone. They are pulled into the direction of $\efv$. 
This regime can be called ``mesoscopic" as sample to sample fluctuations are strong and non trivial.
Note that in the present model the magnetic response to the field, although non-trivial and sample dependent, does not exhibit jumps (so-called static avalanches or shocks) at very low temperature, as were observed and studied in other mean-field models such as the SK model; see \cite{YoungBrayMoore, YoshinoRizzo, LeDoussalMullerWiese08, LeDoussalMullerWiese10, LeDoussalMullerWiese12}. 

For more details on the geometry of the Gibbs measure see Section \ref{sec:Geometry}, 
in particular, the Table \ref{table:geometrysummary} and the summary in Subsection \ref{sec:GeometrySummary}.

\subsection{Magnetization and susceptibility}

We also evaluate the magnetization, susceptibility, and differential susceptibility. 
One of the results is that at low temperature $0<T<1$, for $h\sim N^{-1/2}$, the linear susceptibility defined by $\scp= \frac1{h}\langle \mgn\rangle$ satisfies 
\beq
	\scp \simeq  1 + \frac{|n_1| \sqrt{1 - \tmp}}{hN^{1/2}} \tanh \left(\frac{hN^{1/2}|n_1|\sqrt{1 - T}}{T}\right) 
\eeq
for asymptotically almost every disorder sample. This formula is a consequence of the spin geometry which has an interpretation as an effective 2 level system, as discussed above.
 
The above formula implies the zero external field limit of the susceptibility:
\beq \label{eq:sc2}
	\lim_{H\to 0} \lim_{\substack{N\to \infty \\ h=HN^{-1/2}}} \scp = 1 +  \frac{ n_1^2  (1 - \tmp)}{\tmp}
	\quad \text{and} \quad  
	\lim_{H\to 0} \lim_{\substack{N\to \infty \\ h=HN^{-1/2}}} \bar{ \scp}  = \frac1{T}. 
\eeq
The limit of $\scp$ depends on the disorder variable $n_1^2$. See Figure \ref{fig:susvarvsh} (b). 
This result shows that the Curie law holds for the sample-to-sample average, but not for a given disorder sample. 
If we take a different limit, namely if we let $N\to\infty$ with $h>0$ first and then let $h\to 0$, then the limit
of the susceptibility is deterministic and given by $\min \{ T^{-1}, 1\}$. See Figure \ref{fig:susvarvsh} (a). This formula was previously obtained in \cite{kosterlitz1976spherical}, and also in \cite{cugliandolo2007nonlinear}. 
See Subsections \ref{sec:suscep} and \ref{sec:diffsusc} for details. 

\begin{figure}[H]
\centering
\begin{subfigure}{0.45\textwidth}
\includegraphics[width = 0.9\textwidth]{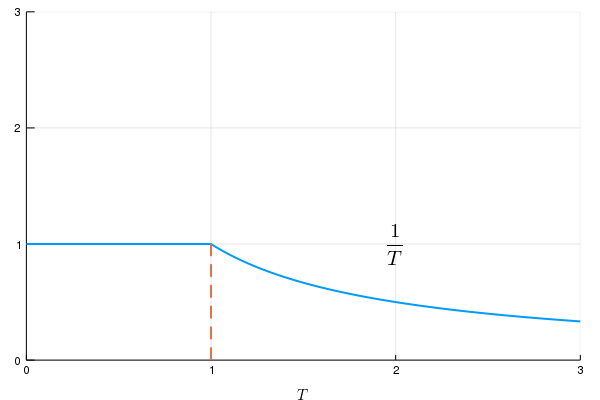}
\caption{$\displaystyle \lim_{\ef\to 0} \lim_{\substack{N\to \infty \\ \ef>0}} \scp$}
\label{fig:sus_vs_h}
\end{subfigure}
\begin{subfigure}{0.45\textwidth}
\includegraphics[width = 0.9\textwidth]{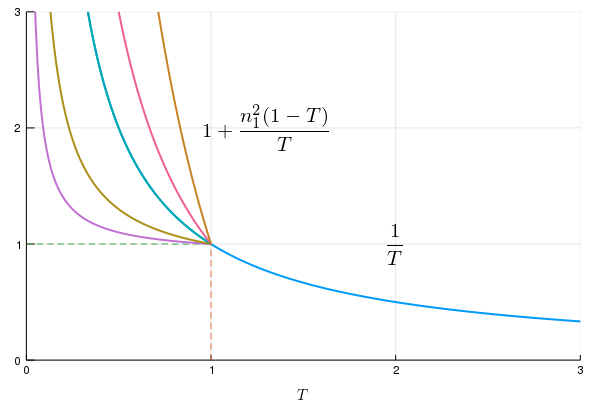}
\caption{$\displaystyle \lim_{H\to 0} \lim_{\substack{N\to \infty \\ h=HN^{-1/2}}}\scp$ for sevaral values of $n_1^2$}
\label{fig:susvar_vs_h}
\end{subfigure}
\caption{Graph of the zero external field limit of the susceptibility as a function of $T$.}
\label{fig:susvarvsh}
\end{figure}

\section{Contour integral representations}\label{sec:integralreps}

The partition function is an $N$-fold integral over a sphere. 
Using the Laplace transform and Gaussian integrations,  Kosterlitz, Thouless and Jones showed in  \cite{kosterlitz1976spherical} that this integral can be expressed as a single contour integral which involve the disorder sample.
We state this result and also include its derivation in Subsection \ref{sec:integralfreeeg}.
By the same method, the moment generating functions of the overlaps can also be written as a ratio of single or double contour integrals.  These results are presented in section \ref{sec:integraloverlaps} 

\subsection{Free energy} \label{sec:integralfreeeg}

The following result holds for any disorder sample. 

\begin{lemma}[\cite{kosterlitz1976spherical}]
Let $M$ be an arbitrary $N$ by $N$ symmetric matrix and let $\efv$ be an $N$ dimensional vector. 
Let $\eg_1\ge\cdots \ge\eg_N$ be the eigenvalues of the matrix $M$ and let $\vu_i$ be a corresponding unit eigenvector. 
Then, the partition function $\pat_N$ defined in \eqref{eq:def_fe} can be written as 
\beq
\label{eq:laplace_trans}
	\pat_N = C_N\int_{\cp - \ii \infty}^{\cp + \ii \infty} e^{\frac{N}{2}\G(z)}\dd z \quad \text{where} \quad C_N = \frac{\Gamma(N/2)}{2\pi \ii (N\beta/2)^{N/2 - 1}}
\eeq
and 
\beq
\label{eq:def_G}
	\G(z) = \beta z - \frac1N\sum_{i = 1}^N \log(z - \eg_i) + \frac{\ef^2\beta}{N}\sum_{i = 1}^N \frac{n_i^2}{z - \eg_i}
	\quad \text{with} \quad  n_i=  \vu_i\cdot\efv. 
\eeq
Here, the integration is over the vertical line $\cp+\ii\R$ where $\cp$ is an arbitrary constant satisfying $\cp> \eg_1$. 
\end{lemma}

\begin{proof}
Let $\Lambda = \diag(\eg_1, \eg_2, \cdots, \eg_N)$. Let $O=(\vu_1, \cdots, \vu_N)$ be an orthogonal matrix so that $M = O\Lambda O^T$. 
Let $S^{N-1}$ be the sphere of radius $1$ in $\R^N$ and let $\dd \Omega_{N-1}$ be the surface area element on $S^{N-1}$. 
Then, using the changes of variables $\frac1{\sqrt{N}} O^T\spin =x$, 
\beqq
	\pat_N = \frac{1}{|S^{N-1}|} I\left( \frac{\beta N}{2}, \ef \sqrt{2\beta} \right)  
	\quad \text{where} \quad I(t,s) = \int_{S^{N-1}} e^{t\sum_{i=1}^N \lambda_i x_i^2 + s\sqrt{t} \sum_{i=1}^N n_i x_i} \dd\Omega_{N-1}(x).
\eeqq
where $n_i= (\eigvm^T\efv)_i = \vu_i \cdot \efv$. 
We take the Laplace transform of $J(t)= t^{N/2-1} I(t,s)$. Making a simple change of variables $t=r^2$ and using Gaussian integrals, the Laplace transform is equal to 
\beqq
	L(z) = \int_0^\infty e^{-zt} J(t) \dd t = 2 \int_{\R^N} e^{-\sum_{i=1}^N (z-\lambda_i) y_i^2 + s\sum_{i=1}^N n_i y_i} \dd^N y
	= 2\prod_{i=1}^N e^{\frac{s^2n_i^2}{4(z-\eg_i)}} \sqrt{\frac{\pi}{z-\eg_i}}
\eeqq
for $z$ satisfying $z>\lambda_1$. 
We obtain a single integral formula of the partition function by taking the inverse Laplace transform. 
\end{proof} 

Note that the sign ambiguity of $\vu_i$ does not affect the result since the formula depends only on $n_i^2$.

\subsection{Overlaps}\label{sec:integraloverlaps}

In this section, we give the moment generating function of each of the overlaps, expressed as a ratio of contour integrals.  The proofs are similar to the computations for the free energy case and we give the proof in Appendix \ref{sec:integraloverlapproof}.

\begin{definition}. The following three functions are related to the function $\G$ and will be used to compute the three overlaps respectively.  We denote by $\eta\in\R$ the parameter that will be used for the moment generating function of each overlap.
\begin{itemize}
\item For the overlap with the external field, we use the function
\beq
\label{eq:G_mgn}
	\Gmgn(z) := \beta z - \frac1N\sum_{i = 1}^N \log(z - \eg_i) + \frac{(\ef + \frac{\eta}{N})^2\beta}{N} \sum_{i = 1}^N \frac{n_i^2}{z - \eg_i}. 
\eeq
Note that this is $\G(z)$ with $h$ replaced by $h+\eta N^{-1}$.

\item For the (square of the) overlap with the 
ground state, we use the function 
\beq
\begin{aligned}
\label{eq:G_om}
\Gom(z)  := & \beta z - \frac1N \log\left(z - \left(\eg_1 + \frac{2\eta}{N}\right)\right)- \frac{1}{N}\sum_{i = 2}^N \log(z - \eg_i) \\
		& \quad + \frac{\ef^2\beta}{N} \frac{n_1^2}{z - (\eg_1 + \frac{2\eta}{N})} + \frac{\ef^2\beta}{N} \sum_{i = 2}^N \frac{n_i^2}{z - \eg_i}.
\end{aligned}
\eeq
Note that this is $\G(z)$ with $\lambda_1$ replaced by $\lambda_1+\frac{\eta}{\beta N}$.

\item For the overlap with a replica, we use the function
\beq
	\Govl(z, w;a) := \beta (z+w) - \frac1{N} \sum_{i=1}^N \log \left( (z-\eg_i)(w-\eg_i) - a^2 \right) 
	+ \frac{\ef^2 \beta}{N} \sum_{i=1}^N \frac{n_i^2 (z+w-2\eg_i + 2a) }{(z-\eg_i)(w-\eg_i) - a^2}.
\eeq

\end{itemize}
\end{definition}

\begin{lemma} \label{lem:contour}
For real parameter $\eta$, the moment generation functions of the three overlaps are as follows:
\beq   
	\langle e^{\beta \eta \mgn}\rangle =  \frac{\int e^{\frac N2 \Gmgn (z)} \dd z}{\int e^{\frac N2 \G(z)} \dd z}, \qquad
	\langle e^{\beta \eta \OM} \rangle = \frac{\int e^{\frac N2 \Gom (z)} \dd z}{\int e^{\frac N2 \G(z)} \dd z}, \qquad
	\langle e^{\eta \ovl} \rangle = \frac{\iint e^{\frac{N}2 \Govl(z, w; \frac{\eta}{\beta N})} \dd z \dd w}{\iint e^{\frac{N}2 \Govl(z, w; 0)} \dd z \dd w}.	
\eeq
The contours are vertical lines in the complex plane such that all singularities lie on the left of the contour. 
See Appendix \ref{sec:integraloverlapproof} for the derivation.
\end{lemma}

\section{Results from random matrices}\label{sec:RMT}

Since the disorder matrix $M$ is a GOE matrix, the eigenvectors are uniformly distributed on the sphere. 
On the other hand, the eigenvalues statistics are well studied in random matrix theory. 
We summarize several definitions and properties of the eigenvalues and other related quantities that we use in this paper. 

\subsection{Probability notations}

There are two types of randomness, one from the disorder sample $M$ and $\efv$, and the other from the Gibbs (thermal) measure. 
We often need to distinguish them. 
We add the subscript $s$ to denote sample probability or sample expectation such as $\PP_s$ and $\E_s$. 
In addition, we use the following notations. 

\begin{definition} 
When describing the limiting distributions in our results, we consider two classes of random variables, which we refer to as sample random variables and thermal random variables.  To distinguish between these two classes, we denote them with the calligraphic font and the gothic font respectively.  For example a standard Gaussian sample random variable and a standard Gaussian thermal variable will be denoted below by
\beq
\mathcal{N} \quad\text{and}\quad\mathfrak{N}
\eeq
respectively. 
\end{definition}

\begin{definition} Asymptotic notations:
\begin{itemize}
\item If $\{E_N\}_{N=1}^\infty$ is a sequence of events, we say that $E_N$ holds asymptotically almost surely (or everywhere) if $\PP_s(E_N)\to1$ as $N\to\infty$.  This probability is with respect to the choice of disorder sample.
\item For two $N$-dependent random variables $A:= A_N$ and $B :=B_N$, the notation
\beq
		A = \bhp{B}
\eeq
means that, for any $\e>0$, the inequality $A\le BN^\e$ holds asymptotically almost surely.

\item The notation $\simeq$ means an asymptotic expansion up to the terms indicated on the right-hand side and the notation 
 $\asymp$ denotes two sides are of the same order.  When we say $A\asymp\cO(B)$ we mean that, for any $\e>0$, the inequality $BN^{-\e}<A<BN^\e$ holds asymptotically almost surely.
 \end{itemize}
 \end{definition}
 
 \begin{definition} Convergence notations:
 \begin{itemize}
\item The convergence in distribution of a sequence of random variables $X_N$ to a random variable $X$ with respect to the disorder variables is denoted by $X_N\Rightarrow X$. 
\item We use the notations $\eqids$ and $\simeqids$ to denote an equality and an asymptotic expansion in distribution with respect to the disorder sample, respectively. 
\item We use similar notations with a different font, $\eqidsgibbs$ and $\simeqidsgibbs$, to denote an equality and an asymptotic expansion in distribution with respect to the Gibbs (thermal) measure, respectively. 
\end{itemize}
\end{definition}
It is worth noting that many of our results actually hold with high probability (i.e., there exist some $D > 0,N_0>0$ such that, for all $N \geq N_0$,
$\P(E_N) >1- N^{-D}$).  While high probability is much stronger than asymptotically almost sure probability, it is much more delicate to prove and we do not discuss those proofs in the current paper.

\subsection{Semicircle law}

The empirical distribution of eigenvalues of $\sGOE$ converges to the semicircle law \cite{Mehta}: for every continuous bounded function $f(x)$,  
	\beq
	\label{eq:lln}
		\frac{1}{N} \sum_{i = 1}^N f(\eg_i) \rightarrow \int f(x)\scl
		\quad \text{where} \quad 
		\scl = \frac{\sqrt{4 - x^2}}{2\pi}\mathbbm{1}_{x \in [-2,2]}\dd x
	\eeq
with probability $1$ as $N\to \infty$. 

\begin{definition} We define the following functions for later use: 
\beq
\label{eq:def_sk}
	s_0(z) := \int \log(z - x)\scl \quad \text{and} \quad s_k(z) := \int \frac{\scl}{(z - x)^k} \quad\text{for $k=1,2,\cdots$,}
\eeq
\end{definition}
\hspace{-0.2in}\textbf{Properties:}
These functions can be evaluated explicitly as
\beq \label{eq:s_1}
\begin{split}
	s_0(z) =& \frac14 z(z - \sqrt{z^2 - 4}) + \log(z + \sqrt{z^2 - 4}) - \log 2 - \frac 12,\\
	s_1(z) =& \frac{z - \sqrt{z^2 - 4}}{2}, \quad 
	s_2(z) = \frac{z - \sqrt{z^2 - 4}}{2\sqrt{z^2 - 4}},
	\quad s_3(z) = \frac{1}{(z^2 - 4)^{3/2}}, 
	\quad s_4(z)= \frac{z}{(z^2-4)^{5/2}}
\end{split} \eeq
for $z$ not in the real interval $[-2,2]$. As $z\to 2$, we have 
\beq \label{eq:stjaspt}
\begin{split}
	s_1(z) \simeq 1-\sqrt{z-2}, \quad 
	s_2(z) \simeq \frac{1}{2\sqrt{z-2}} - \frac12,
	\quad s_3(z) \simeq \frac{1}{8 (z-2)^{3/2}}, 
	\quad s_4(z) \simeq \frac{1}{16(z-2)^{5/2}}. 
\end{split} \eeq

\subsection{Rigidity}\label{sec:rigidity}

\begin{definition}
For $i = 1,2,\cdots, N$, let  $\widehat{\eg}_i$ be the classical location defined by the quantile conditions 
\beq	\label{eq:def_classical_eg}
		\int_{\widehat{\eg}_i}^2 \scl  = \frac i N.
\eeq
We set $\widehat{\eg}_0 = 2$. We also set $\widehat{a}_i = (\widehat{\eg}_i-2) N^{2/3}$.
\end{definition}
\hspace{-0.2in}\textbf{Rigidity property:} 
The rigidity result  \cite{EYY, MR3098073} states that
\beq	\label{eq:rigidity}
		|\eg_i- \widehat \eg_i| \leq  (\min\{i, N + 1 - i\})^{-1/3}\bhp{N^{-2/3} }
\eeq
uniformly for $i = 1, 2,  \cdots, N$. 

The rigidity property allows us to apply the method of steepest descent to evaluate the integrals involving the eigenvalues since the eigenvalues are close enough to the classical location, and the fluctuations are small enough. 

\subsection{Edge behavior}\label{sec:edgebehavior}

\begin{definition}\mbox{}
\vspace*{-\parsep}
\begin{itemize}
\item Define the rescaled eigenvalues
\beq
\label{eq:scaledevii}
	a_i := N^{2/3}(\eg_i - 2). 
\eeq
\item Define $\{\airy_i\}_{i = 1}^\infty$ to be the GOE Airy point process process to which the rescaled eigenvalues converge in distribution as $N \rightarrow \infty$  \cite{TracyWidom94, soshnikov1999universality}: 
\beq
	\{a_i\}\Rightarrow \{\airy_i\}.
\eeq
\end{itemize}
\end{definition}
\hspace{-0.2in}\textbf{Properties:}\\
 The rightmost point $\airy_1$ of the GOE Airy point process has the $GOE$ Tracy-Widom distribution  
	\beq
	\label{eq:eg1dis}
		a_1 \Rightarrow \alpha_1 \eqids \twgoe.
	\eeq
 The GOE Airy point process satisfies the asymptotic property that 
\beq \label{eq:airypotamp}
		\airy_i \simeq  -\left(\frac{3\pi i}{2} \right)^{2/3} \quad \text{as $i\to \infty$.}
\eeq
This asymptotic is due to the fact that the semicircle law is asymptotic to $\frac{\sqrt{2-x}}{\pi} \dd x$ as $x\to 2$. 
 The above formula and the rigidity imply that, with high probability, 	
\beq \label{eq:egresk}
		\egres_i \asymp - i^{2/3} \quad \text{as $i, N\to \infty$ satisfying $i\le N$}
\eeq

\subsection{Central limit theorem of linear statistics}

For a function $f$ which is analytic in an open neighborhood of $[-2,2]$ in the complex plane, consider the sum of $f(\eg_i)$. 
The semicircle law \eqref{eq:lln} gives its leading behavior. If we subtract the leading term, the difference 
\beq \label{eq:linearstcs}
		\sum_{i = 1}^N f(\eg_i) - N\int f(x)\scl  
\eeq
converges to a Gaussian distribution with explicit mean and variance; see, for example, \cite{JohanssonCLT, Bai2005, Lytova2009}. 
Note that unlike the classical central limit theorem, we do not divide by $\sqrt{N}$. 

\begin{definition} 
Define
\beq \label{eq:defoflsl}
	\lsl_N(z):= \sum_{i=1}^N \log(z-\eg_i) - N s_0(z). 
\eeq
for $z>2$ where $s_0(z)$ is given by \eqref{eq:s_1}.
\end{definition}
\hspace{-0.2in}\textbf{Properties:}\\
The above-mentioned central limit theorem implies in this case that 
\beq \label{eq:clt} 
	\lsl_N(z) \Rightarrow \NN(M(z), V(z))
\eeq
where (see Lemma A.1 in \cite{BaikLeeSSK})
\beq \label{eq:cltmv}
	M(z) 
	= \frac12 \log \left( \frac{2\sqrt{z^2-4}}{z+\sqrt{z^2-4}} \right), \qquad
	V(z)= 2 \log \left( \frac{z+\sqrt{z^2-4}}{2\sqrt{z^2-4}} \right). 
\eeq
For later uses, we record that for $0<\beta<1$, 
\beq \label{eq:MandVforcp0}
	M(\beta+\beta^{-1})= \frac12\log(1-\beta^2), \qquad V(\beta+\beta^{-1})=-2 \log(1-\beta^2).
\eeq

\subsection{Special sums}
In this section we collect several important results about convergence of various types of sums that we will use in this paper.  Many of the results are motivated by the need to work with sums of the form
\beq \label{eq:rationalfunk}
	\frac1N \sum_{i = 2}^N \frac{1}{(\eg_1 - \eg_i)^k} , \qquad k = 1, 2, \cdots,
\eeq
or its variations. 
The above quantity looks superficially close to the linear statistics \eqref{eq:linearstcs} with $f(x)= \frac1{(\eg_1-x)^k}$ with one term removed but the function $f(x)$ is singular at $x=\eg_1$. 
We note that if we replace $f(x)$ by $\frac1{(2-x)^k}$ and use the semicircle law, we obtain $s_k(2)$ which diverges for $k\ge 2$. 
Hence, the result of the previous subsection does not apply. 
On the  hand, for $k=1$, $s_1(2)=1$. This fact indicates that the above sum still converges when $k=1$. 

We present several definitions, followed by their related convergence results and some brief explanation of why these results hold.
Recall the definition $a_i=(\eg_i-2)N^{2/3}$ and $n_i=\vu_i\cdot \mathbf{g}$. 

\begin{definition} We define the following random sums, which depend on the disorder sample:
\begin{itemize}
\item Define 
	\beq \label{eq:def_crvbp}
		\Xi_N := N^{1/3} \left(\frac1N \sum_{i = 2}^N \frac{1}{\eg_1 - \eg_i} - 1\right) = \sum_{i=2}^N \frac1{a_1-a_i} - N^{1/3}. 
	\eeq 
	
\item Define, for $w>0$, 
\beq \label{eq:crvdnf}
	\crv(w) := N^{1/3} \left[ \frac{1}{N} \sum_{i=1}^N \frac{n_i^2}{w N^{-2/3}+\eg_1-\lambda_i} -1 \right] 
	= \sum_{i=1}^N \frac{n_i^2}{w+\egres_1-\egres_i}-N^{1/3}. 
\eeq

\item 
Define, for $z>2$, 
\beq
\label{eq:go00}
	\go_N(z;k) := \frac{1}{\sqrt{N}} \sum_{i = 1}^N \frac{n_i^2 - 1}{(z - \widehat \eg_i)^k} \qquad\textit{for }k\ge 1.
\eeq

\end{itemize}
\end{definition}

\begin{definition} We define the following limits, which depend on the GOE Airy point process $\{ \alpha_i\}$:
\begin{itemize}
\item Define 
	\beq
	\label{eq:def_crvbplim}
		\Xi := \lim_{n \rightarrow \infty} \left(\sum_{i = 2}^n \frac{1}{ \airy_1 - \airy_i} - \frac1{\pi} \int_0^{\left(\frac{3\pi n}{2}\right)^{2/3}} \frac{\dd x}{\sqrt{x}}\right).
	\eeq
Landon and Sosoe showed that the limit exists almost surely \cite{LandonSosoe}. 

\item Define $\crvlim(s)$ as follows, where $\nu_i$ are i.i.d. Gaussian random variables with mean 0 and variance 1 independent of the GOE Airy point process $\alpha_i$:
\beq	
	\crvlim(s) := \lim_{n \rightarrow \infty} \left(\sum_{i = 1}^n \frac{\evg_i^2}{s +\airy_1 - \airy_i} - \frac1{\pi} \int_0^{\left(\frac{3\pi n}{2}\right)^{2/3}} \frac{\dd x}{ \sqrt{x}}\right). 
\eeq
This limit exists almost surely by a similar argument as in \cite{LandonSosoe} showing that $\Xi$ exists. 
\end{itemize}
\end{definition}

\begin{result} Using the notations above, we have the following convergence results.
\begin{itemize}

\item Landon and Sosoe proved in \cite{LandonSosoe} that
\beq\label{eq:def_crvbp}
	\Xi_N \Rightarrow \Xi .
\eeq
They use this result to describe the fluctuations of the overlap with a replica when $h=0$ and $T<1$.  

\item We also need another version of the result \eqref{eq:def_crvbp} where the constant numerators are replaced $n_i^2$:
\beq
\label{eq:weightedsum1}
	N^{1/3}\left(\frac1N \sum_{i = 2}^N \frac{n_i^2}{\eg_1 - \eg_i} - 1\right) \Rightarrow  \lim_{n \rightarrow \infty} \left(\sum_{i = 2}^n \frac{\nu_i^2}{\airy_1 - \airy_i} - \frac1{\pi} \int_0^{\left(\frac{3\pi n}{2}\right)^{2/3}} \frac{\dd x}{\sqrt{x}}\right)
\eeq
where $\nu_i$ are i.i.d standard Gaussians, independent of the GOE Airy point process $\alpha_i$.  
This follows from \eqref{eq:def_crvbp} and the fact that 
\beq
	\frac1{N^{2/3}} \sum_{i = 2}^N \frac{n_i^2-1}{\eg_1 - \eg_i} \Rightarrow \sum_{i = 2}^\infty \frac{\nu_i^2-1}{\airy_1 - \airy_i}
\eeq
which is a convergent series due to Kolmogorov's three series theorem and \eqref{eq:airypotamp}. 

\item By the same argument as for \ref{eq:weightedsum1}, 
\beq \label{eq:def_crv_lim}
	\crv(w)\Rightarrow\crvlim(w)\quad\textit{for }w>0.
\eeq
\item By the Lyapunov central limit theorem and the definition of $\widehat \eg_i$, we have 
\beq
\label{eq:go}
	\go_N(z;k) 
	\Rightarrow \NN(0, 2s_{2k}(z))
\eeq
as $N\to \infty$ for $z>2$.
(Note that the variance of $n_i^2-1$ is $2$.)
\end{itemize}
\end{result}

\begin{result} In addition to the convergence results listed above, we also need estimates that hold for asymptotically almost every disorder sample. 
\begin{itemize}
\item A consequence of \eqref{eq:def_crvbp} is that 
\beq
	\label{eq:res_edge}
		\frac1N \sum_{i = 2}^N \frac{1}{\eg_1 - \eg_i} = \frac1{N^{1/3}} \sum_{i=2}^N \frac1{a_1-a_i} = 1 + \bhp{N^{-1/3}}.
	\eeq
	 for asymptotically almost every disorder sample.
	 
\item We have 
\beq\label{lem:specialsum}
	\sum_{i=2}^N\frac{1}{(a_1-a_i)^k}=\bhp{1} \quad \text{and} \quad 
	\sum_{i=2}^N\frac{n_i^2}{(a_1-a_i)^k}=\bhp{1}  , \qquad\text{$k\geq2$,}
\eeq
 for asymptotically almost every disorder sample.
This follows from the fact that the $a_i\asymp -i^{2/3}$ and the difference $(a_1-a_2)^{-1}$ is of order 1 with vanishing probability.

\item We also need the result
\beq \label{eq:weightls}
	\frac1{N} \sum_{i = 1}^N \frac{n_i^2}{(z - \eg_i)^k}  = s_k(z) + \frac{\go_N(z;k)}{\sqrt{N}} + \bhp{N^{-1}}, 
	\qquad z>2, \quad k>1
\eeq  
for asymptotically almost every disorder sample.
To justify \eqref{eq:weightls}, we observe that 
\beq
	\frac1{N} \sum_{i = 1}^N \frac{n_i^2}{(z - \eg_i)^k} = \frac1{N} \sum_{i = 1}^N \frac{1}{(z - \eg_i)^k} + \frac1{N} \sum_{i = 1}^N \frac{n_i^2-1}{(z - \eg_i)^k}.
\eeq
We then use the central limit theorem \eqref{eq:linearstcs} for linear statistics for the first sum and replace $\eg_i$ by $\widehat \eg_i$ in the second sum using the rigidity \eqref{eq:rigidity}.  


\end{itemize}
\end{result}

\section{Fluctuations of the free energy}\label{sec:feh>0}

From the integral formula  \eqref{eq:laplace_trans},  using 
\beq	
	C_N = \frac{\sqrt{N} \beta}{2\ii \sqrt{\pi}(\beta e)^{N/2}} (1 + O(N^{-1})), 
\eeq
the free energy can be written as 
\beq\label{eq:fe_stedes}
	\fe_N = \frac{1}{2\beta}(\G(\cp) - 1 - \log\beta) + \frac1{N\beta} \log \left(  \frac{\sqrt{N} \beta}{2\ii \sqrt{\pi}} \int_{\cp - \ii \infty}^{\cp + \ii \infty}  e^{\frac N2 (\G(z) - \G(\cp))} \dd z\right) 
	+ O(N^{-2})
\eeq
where $O(N^{-2})$ is a constant that does not depend on the disorder sample $M$ and $\efv$. 
We evaluate the integral asymptotically using the method of steepest descent.  
The formula for $\G(z)$ is given in \eqref{eq:def_G} and  
\beq
\label{eq:dG}
	\G'(z) = \beta - \frac 1 N\sum_{i = 1}^N \frac{1}{z - \eg_i} - \frac{\ef^2\beta}{N}\sum_{i = 1}^N \frac{n_i^2}{(z - \eg_i)^2}
	\qquad \text{where} \quad  n_i= \vu_i\cdot \efv.
\eeq
For real $z$, $\G'(z)$ is an increasing function taking values from $-\infty$ to $\beta$ as $z$ moves from $\eg_1$ to $\infty$. 
Hence, there is a unique real critical point $\gamma$ satisfying 
\beqq
	\G'(\gamma)=0, \qquad \gamma>\eg_1. 
\eeqq
We set $\cp$ for the contour of \eqref{eq:fe_stedes} to be this critical point. 

In this section, we use the formula \eqref{eq:fe_stedes} to evaluate the fluctuations of the free energy when the external field strength $\ef$ is fixed. 
For the case  
$\ef=0$, this computation was done in \cite{kosterlitz1976spherical} for the leading deterministic term and in \cite{BaikLeeSSK} for the subleading term. 
For fixed $\ef>0$, the fluctuations for the SK model were computed in \cite{ChenDeyPanchenko} using a method different from the one of this paper. 
We first review the computation of \cite{BaikLeeSSK} for $\ef=0$ and then give a new computation for fixed $\ef>0$ using the above integral formula. 

\medskip

The following formula will be used in one of the subsections: Since $\G(z)-\G(\cp)= \G(z)-\G(\cp)-\G'(\cp)(z-\cp)$, 
we can write 
\beq	\label{eq:integconsfe}
	N(\G(z) - \G(\cp)) = -\sum_{i = 1}^N \left[\log(1 + \frac{z - \cp}{\cp - \eg_i}) - \frac{z - \cp}{\cp - \eg_i}\right] + \ef^2 \beta \sum_{i = 1}^N \frac{n_i^2(z - \cp)^2}{(z - \eg_i)(\cp - \eg_i)^2}.
\eeq

\subsection{No external field: $\ef = 0$}

\subsubsection{High temperature regime: $\tmp > 1$} \label{sec:seconhitzef}

When $\ef = 0$, we write, using the notation \eqref{eq:defoflsl}, 
\beq
	\G(z) = \beta z - \frac 1 N\sum_{i = 1}^N \log(z - \eg_i) = \beta z - s_0(z) - \frac{\lsl_N(z)}{N} , 
	\qquad s_0(z)= \int \log(z-x) \scl.
\eeq
From \eqref{eq:clt}, $\lsl_N(z)= \bhp{1}$ for fixed $z>2$. 
Thus, $\G_0(z):= \beta z - s_0(z)$ is an approximation of the function $\G(z)$ and we first find the critical point $\cp_0$ of $\G_0(z)$ satisfying $\cp_0>2$, where we recall that the largest eigenvalue $\eg_2\to 2$. 
Since $\G_0''(z)>0$, we find that $\min_{z\ge 2} \G'_0(z)= \G'_0(2)= \beta -1$ from the formula \eqref{eq:s_1} of $s_0'(z)=s_1(z)$. 
Thus, the critical point of $\G_0(z)$ exists only when $\frac1{\beta}=T>1$. From the formula, we find that for $T>1$, it is given by 
\beq
	\cp_0:= \beta + \beta^{-1}= T+T^{-1}. 
\eeq
In this case, a simple perturbation argument (see Appendix \ref{sec:pert}) implies that $\cp = \cp_0 + \bhp{N^{-1}}$ and 
\beq
\label{eq:Gcp_h0}
	\G(\cp) =\G(\cp_0) - \frac{\lsl_N(\cp_0)}{N} + \bhp{N^{-2}}=   \frac{\beta^2}{2} + 1 + \log \beta - \frac{\lsl_N(\cp_0)}{N} + \bhp{N^{-2}}.
\eeq

Even though the integral in \eqref{eq:fe_stedes} involves the disorder sample, the rigidity of the eigenvalues from Section \ref{sec:rigidity} implies that, with high probability, the eigenvalues are close to the non-random classical locations (i.e. the quantiles of the semicircle law). 
Thus, we can still apply the method of steepest descent when the disorder sample is in an event of the high probability. 
Using 
\beqq
	\G''(\gamma)\simeq \G_0''(\gamma_0)=s_2(\cp_0) = \frac{\beta^2}{1 - \beta^2}
\eeqq
and $\G^{(k)}(\cp)=\bhp{1}$ for all $k\ge 2$, the Gaussian approximation of the integral is valid and we find that  
\beq \label{eq:ich0highT}
	\int_{\cp - \ii \infty}^{\cp + \ii \infty}  e^{\frac N2 (\G(z) - \G(\cp))} \dd z 
	\simeq  \frac{\ii \sqrt{4\pi}}{\sqrt{N s_2(\cp_0)}}  = \frac{\ii \sqrt{4\pi (1-\beta^2)}}{\sqrt{N \beta^2}} .
\eeq

Inserting everything into \eqref{eq:fe_stedes} and using the fact that $\lsl_N(\cp_0)$ converges to a Gaussian distribution with mean and variance given by \eqref{eq:MandVforcp0}, we obtain the following result. This result was proved rigorously in \cite{BaikLeeSSK}. 

\begin{theorem}[\cite{BaikLeeSSK}] \label{thm:freehigh}
For $h=0$ and $T>1$, 
\beq
\label{eq:feh0htmp}
	\fe_N(\tmp, 0) = \frac{1}{4T} +\frac{T}{2 N} \left[ \log(1-T^{-2}) - \lsl_N(\cp_0) \right]	+ \bhp{N^{-3/2}}
\eeq
as $N\to \infty$ with high probability, where $\cp_0=T+T^{-1}$ and $\lsl_N(z)$ is defined in \eqref{eq:defoflsl}. As a consequence, 
\beq \label{eq:higtemzeef}
	\fe_N (\tmp, 0) \simeqids  \frac1{4\tmp} + \frac{\tmp}{2N} \NN(-\mn, 4\mn)  \qquad \mn := -\frac12 \log(1-\tmp^{-2}),
\eeq
where $\NN(a, b)$ is a (sample) Gaussian distribution of mean $a$ and variance $b$. 
\end{theorem}

\subsubsection{Low temperature regime: $\tmp < 1$}
\label{subsubsec:ltmph0}

In contrast to the previous section, the function $\G_0(z)= \beta z - s_0(z)$ is no longer a good approximation of $\G(z)$ for $0<\tmp<1$  when $h=0$.
Indeed, the function $\G_0(z)$ does not have a critical point satisfying $z>2$. Hence,  we need to find the critical point $\cp$ of $\G(z)$ directly. 
Since the critical point of $\G_0(z)$ when $T=1$ is given by $\cp_0=2$, it is reasonable to assume that  when $0<T<1$, $\cp$ is close to the large eigenvalue $\eg_1$. It turns out that $\cp=\eg_1+ \bhp{N^{-1}}$. 
We set $\cp = \eg_1+ sN^{-1}$ with $s=\bhp{1}$ and determine $s$. 
Separating out the term with $i = 1$ in the equation \eqref{eq:dG} and using \eqref{eq:res_edge}, 
\beq	
\label{eq:cph0ltmp}
	\G'(\cp) = \beta - \frac1{N(\cp-\eg_1)}- \frac1N \sum_{i=2}^N \frac1{\cp-\eg_i} = \beta  -  \frac{1}{s} - 1 + \bhp{N^{-1/3}} =0. 
\eeq
Thus $s = \frac{1}{\beta - 1} + \bhp{N^{-1/3}}$, which is consistent with our assumption that $s = \bhp{1}$. To evaluate 
\beqq
	\G(\cp) = \beta \cp - \frac1N\sum_{i = 1}^N \log( \cp - \eg_i) , 
\eeqq
we use \eqref{eq:defoflsl}-\eqref{eq:cltmv}. We need to evaluate $\sum_{i=1}^N \log(z-\lambda_i)$ for $z= 2+O(N^{-2/3})$.  Observe that 
\beqq
	M(z)=O(\log(z-2)) \quad \text{and} \quad V(z)=O(\log(z-2))\quad \text{as $z\to 2$.}
\eeqq
Hence, a formal application of \eqref{eq:clt} to this case using
	$s_0(z)
	= \frac12 + (z-2) + O((z-2)^{3/2})$
implies that for $z\to 2$ such that $|z-2|\ge N^{-d}$ for some $d>0$, 
\beq \label{eq:smoflz}
	\frac1{N} \sum_{i=1}^N \log\left(z-\lambda_i \right) = s_0\left( z \right) + \bhp{N^{-1}}
	= \frac12 + (z-2)+ \bhp {N^{-1}}+O((z-2)^{3/2}).
\eeq
This heuristic computation indicates that 
\beq
	\G(\cp) = \beta \cp - \frac 1 N\sum_{i = 1}^N \log( \cp - \eg_i)
	=2 \beta - \frac12 + (\beta - 1)(\eg_1 - 2) + \bhp{N^{-1}}.
\eeq

We now consider the integral in \eqref{eq:fe_stedes}. 
For $k\ge 2$, we have, using the notation \eqref{eq:scaledevii} for the scaled eigenvalues $\egres_i=N^{2/3}(\eg_i-2)$ and the estimate \eqref{lem:specialsum}, 
\beq \begin{split}
	\frac{\G^{(k)}(\cp)}{(-1)^k(k-1)!}= \frac1{N} \sum_{i=1}^N \frac1{(\cp-\eg_i )^k}  
	= \frac{N^{k-1}}{s^k} + N^{\frac23k-1} \sum_{i=2}^N \frac1{(\egres_1+sN^{-1/3}-\egres_i)^k}
	= \bhp{N^{k-1}}
\end{split} \eeq
with high probability. The estimate $\G''(\cp)=\bhp{N}$ indicates that the main contribution to the integral comes from a neighborhood of radius $N^{-1}$ of the critical point. 
However, all terms of the Taylor series
\beqq
	N\left( \G(\cp+u N^{-1}) - \G(\cp) \right) = \sum_{k=2}^N N^{1-k} \frac{\G^{(k)}(\cp)}{k!} u^k
\eeqq
are of the same order $\bhp{1}$ for finite $u$.  Hence, we cannot replace the integral with a Gaussian integral. 
Instead, we proceed as follows. Using the formula \eqref{eq:integconsfe}, separating out the $i=1$ term from the sum, using a Taylor approximation for the remaining sum, and using \eqref{lem:specialsum}, 
\beq
\label{eq:integparth0ltmp} \begin{split}
	N \left( \G(\cp+u N^{-1}) - \G(\cp) \right) 
	&=  - \log \left(1+ \frac{u}{s}\right) + \frac{u}{s} + O\left(  \sum_{i=2}^N   \frac{u^2 N^{-2/3}}{ (\egres_1-sN^{-1/3} - \egres_i)^2}   \right) \\	
	&= - \log \left(1 + \frac{u}{s}\right) + \frac{u}{s} + \bhp{N^{-2/3}}
\end{split} \eeq
with high probability for finite $u$. 
From this, 
\beq
	\int_{\cp - \ii \infty}^{\cp + \ii \infty}  e^{\frac N2 (\G(z) - \G(\cp))} \dd z 
	\simeq \frac1{N} \int_{-\ii\infty}^{\ii \infty} \frac{e^{\frac{u}{s}}}{1+ \frac{u}{s}} \dd u \asymp \bhp{N^{-1}}.
\eeq
We do not need the exact value of the integral, but only the estimate that its log is $\bhp{\log N}$. 

We thus obtain the following result, which was proved rigorously in \cite{BaikLeeSSK}.  

\begin{theorem}[\cite{BaikLeeSSK}] \label{thm:freelow}
For $h=0$ and $0\le T<1$, 
\beq  \label{eq:fe10ds}
	\fe_N(\tmp, 0) =  1- \frac{3T}{4} + \frac{ T \log T }{2} + \frac{1-T} {2N^{2/3}} a_1 + \bhp{N^{-1}}
\eeq
as $N\to\infty$
with high probability. 
As a consequence,  
\beq
\label{eq:fe_h0_ltemp}
	\fe_N(\tmp, 0) \simeqids  1 - \frac{3\tmp}{4} + \frac {\tmp\log \tmp} 2+ \frac{1-T}{2N^{2/3}} \twgoe.
\eeq
\end{theorem}

\begin{remark}
The zero temperature case $T=0$ of the theorem is the standard random matrix theory result that the largest eigenvalue of a GOE matrix converges to the Tracy-Widom distribution. We see that a formal  $T\to 0$ limit of the result implies this statement. 
Similarly, all results of this paper, other than those that have $T>1$ restrictions, have a convergent formal limit if we take $T\to 0$. 
Hence, even though we need a separate argument since there is no integral representation, we expect that all results are valid for the $T=0$ case as well. 
\end{remark}

\subsection{Positive external field: $\ef=O(1)$} \label{sec:freeenergypof}

Fix $\ef>0$. We use  \eqref{eq:clt} and \eqref{eq:weightls} to write 
\beqq
	\G(z) = \beta z - s_0(z) + \ef^2\beta \left[ s_1(z) + \frac{1}{\sqrt{N}}\go_N(z;1) \right]+ \bhp{N^{-1}} 
\eeqq
for $z>\eg_1$. 
The random variable $\go_N(z;k)$ is defined in \eqref{eq:go00} and it converges in distribution to $\NN(0, 2s_{2k}(z))$; see \eqref{eq:go}. 
This time, $\G(z)$ is approximated by the function $\G_0(z)=\beta z - s_0(z) + \ef^2\beta  s_1(z)$. 
Its derivative $\G_0'(z)= \beta - s_1(z)-\ef^2\beta s_2(z)$ is an increasing function for $z>2$ and $\G'(z)\to-\infty$ as $z\downarrow 2$ while $\G'(z)\to+\infty$ as $z\to +\infty$. 
Hence, unlike in the case of $h=0$, there is a point $\cp_0>2$ satisfying $\G_0'(\cp_0)=0$ for all $T>0$. It satisfies the equation
\beq
\label{eq:cp0h>0}
	\G_0'(\cp_0)= \beta - s_1(\cp_0) - \ef^2\beta  s_2(\cp_0) = 0 .
\eeq
A perturbation argument  (see Appendix \ref{sec:pert}) implies that the critical point $\cp$ of $\G(z)$ has the form
\beq
\label{eq:cp_h>0}
	\cp = \cp_0 + \cp_1 N^{-1/2} + \bhp{N^{-1}}.
\eeq
We do not need a formula for $\cp_1$ in this section, but we record it here since we use it in later sections;  
\beq \label{eq:cp1h>0}
	\cp_1 = \frac{\ef^2 \beta \go_N(\cp_0; 2)}{s_2(\cp_0) + 2\ef^2 \beta s_3(\cp_0)}
\eeq
where we used the fact that $\frac{\dd}{\dd z} \go_N(z; 1) = -\go_N(z; 2)$. 
The perturbation argument also implies that 
\beq \label{eq:Gcph>0}
	\G(\cp) = \beta\cp_0 - s_0(\cp_0) + \ef^2\beta s_1(\cp_0) + \frac{\ef^2\beta}{\sqrt{N}}\go_N(\cp_0; 1) + \bhp{N^{-1}}.
\eeq

The integral term in \eqref{eq:fe_stedes} can be evaluated using the steepest descent method as in the case of $h=0$ and $T>1$ since $\G^{(k)}(\cp)=\bhp{1}$ for all $k\ge 2$. From the Gaussian integral approximation, 
\beq\label{lem:freeh>0integral}
	\int_{\cp - \ii \infty}^{\cp + \ii \infty}  e^{\frac N2 (\G(z) - \G(\cp))} \dd z \simeq \frac{\ii \sqrt{4\pi}}{\sqrt{N\G''(\cp)}}\asymp \bhp{N^{-1/2}}.
\eeq

\begin{remark} We do not focus in this paper on justifying the use of steepest descent in this context, but instead provide the computations based on this method.  One can rigorously check that the steepest descent method works here, but it is also worth noting that all the contour integral computations needed in this paper can be achieved without the use of steepest descent.  In fact, for the contour integrals in sections \ref{sec:freeenergypof}, \ref{sec:fe1/4analysis}, \ref{sec:fetranltmp}, \ref{sec:mgnefpos}, \ref{sec:ext1/6}, \ref{sec:om_h>0}, \ref{sec:gs1/6}, \ref{sec:gs1/3}, \ref{sec:ovlpos}, and \ref{sec:ovl1/6} require no contour deformation at all.  Using the straight line contour and crude bounds on the order of the integrand, one can compute, up to leading order, the value of the integral in a neighborhood of $\gamma$ and then show that the tails are of smaller order.  These computations are fairly lengthy and will be omitted from this paper.  The integrals in sections \ref{sec:mgnlowtmn12} and \ref{sec:ovl12} can be treated by a similar method, but require a slight deformation of the original contour.  For ease of computation, we instead employ the steepest descent method here, but without providing rigorous justification.
\end{remark}

Combining the preceding information in this section, we obtain the following result.

\begin{result}\label{thm:feh>0} For fixed $\ef>0$ and $T>0$, 
\beq
\label{eq:feh>101}
	\fe_N(T, \ef) 
	= \felim(\tmp, h)
	+ \frac{\ef^2 \go_N(\cp_0; 1)}{2\sqrt{N}} + \bhp{N^{-1}}
\eeq
as $N\to \infty$ with high probability where $\go_N(z;k)$ is defined in \eqref{eq:go00} and 
\beq \label{eq:felim_h>0} 
\begin{split}
	\felim(\tmp, h) 
	&:= \frac{\cp_0}{2} - \frac{\tmp s_0(\cp_0)}{2} - \frac{\tmp -\tmp\log \tmp}{2}  + \frac{\ef^2 s_1(\cp_0)}{2} 
\end{split} \eeq
with  $\cp_0$ being the solution of the equation 
\beq \label{eq:apprxcpeqhp}
	1 - T s_1(\cp_0) - \ef^2 s_2(\cp_0) = 0, \qquad \cp_0>2. 
\eeq
\end{result}

Since $\go_N(\cp_0; 1)$ converges in distribution to $\NN(0, 2s_{2}(\cp_0))$ from \eqref{eq:go}, we conclude the following result. 

\begin{result}
For fixed $\ef>0$ and $T>0$, as $N\to \infty$, 
\beq \label{eq:fluct_h>0}
	\fe_N(\tmp, h) \simeqids  \felim(\tmp, h) + \frac{1}{\sqrt{N}}\NN \left(0, \frac{h^4 s_2(\cp_0)}{2} \right) . 
\eeq
\end{result}

This result shows that the order of the fluctuations of the free energy is $N^{-1/2}$ for all $T>0$, which is different from both $N^{-1}$ for  $h=0$, $T>1$ and $N^{-2/3}$ for $h=0$, $0<T<1$.  


\subsection{Comparison with the result of Chen, Dey, and Panchenko}

Chen, Dey, and Penchenko computed the fluctuations of the free energy of the SK model with $h>0$ in  \cite{ChenDeyPanchenko} when $\efv = \mathbf{1}$. We compare our result with theirs. 
The adaptation of the approach of \cite{ChenDeyPanchenko} to the SSK model with $\efv=\mathbf{1}$ implies that 
$\sqrt{N} \left ( \fe_N(\tmp, h) -  \E[ \felim(\tmp, h)] \right)$ 
converges in distribution as $N \to \infty$ to the centered Gaussian distribution with variance 
\beq	
\label{eq:fevarh>0parisi}
	\frac{\ef^4 (1 - q_0)^4}{2T^2(T^2-(1 - q_0))}
\eeq
where
\beq	
\label{eq:q0overlap}
	q_0 + \ef^2  = \frac{T^2 q_0}{(1 - q_0)^2}	.
\eeq
The quantity $q_0$ has the interpretation as the overlap of two independent spins from the Gibbs measure involving the same disorder sample, i.e. the overlap of a spin with a replica. 
The formula \eqref{eq:q0overlap} was predicted using the replica saddle point method in \cite{crisanti1992sphericalp} (equation (4.5)) and \cite{FyodorovleDoussal} (equation (29) with $n=0$).

Our result \eqref{eq:fluct_h>0} above is for the SSK model when $\efv$ is a Gaussian vector, but it extends to the case $\efv = \mathbf{1}$.  The only difference is that the variance of the limiting Gaussian distribution \eqref{eq:fluct_h>0} changes to 
\beq \label{eq:vofgiag}
	\frac{\ef^4}{2}(s_2(\cp_0) - (s_1(\cp_0))^2) .
\eeq
Using the fact that $s_2(z)= \frac{s_1(z)^2}{1-s_1(z)^2}$ for $z>2$, it is easy to check that \eqref{eq:fevarh>0parisi} and \eqref{eq:vofgiag} are same with $q_0$ and $\cp_0$ related by the equation 
\beq
	q_0=1- Ts_1(\cp_0). 
\eeq

\subsection{Matching between $\ef> 0$ and $\ef = 0$} \label{sec:freeenegcop}

We have considered three different regimes: (a) $\ef=0$ and $T<1$, (b)  $\ef=0$ and $T>1$, and (c) $\ef=O(1)$. 
The order of the fluctuations of the free energy in these regimes are $N^{-1}$, $N^{-2/3}$, and $N^{-1/2}$, respectively. 
In these cases, the fluctuations are governed by the disorder variables given by (a) all eigenvalues $\eg_1, \cdots, \eg_N$, (b) the top eigenvalue $\eg_1$, and (c) the combinations $n_i=\vu_i\cdot \efv$ of the eigenvectors and the external field. 
These differences indicate that there should be transitional regimes as $h\to 0$. 
We now study the limit $\ef\to 0$ of the result obtained for the case $h>0$ and determine the transitional scaling of $h$ 
heuristically by matching the order of the fluctuations. 
We need to consider the high temperature case and the low temperature case separately. 


\subsubsection{Asymptotic property of $\cp_0$}

Throughout this paper, we will make use of following property of the leading term $\cp_0$ of the critical point of $\G(z)$ when $h=O(1)$.

\begin{lemma} \label{lem:cp0behro}
Let $\cp_0>2$ be the solution of the equation \eqref{eq:apprxcpeqhp}, $1 - T s_1(\cp_0) - \ef^2 s_2(\cp_0) = 0$. Then,
as $h\to 0$, 
\beq \label{eq:cpohs} 
	\cp_0 = 
	\begin{dcases} \tmp +T^{-1}  + \frac{\ef^2}{\tmp} + O(\ef^4)  \qquad &\text{for $T>1$,}\\
	2+ \frac{\ef^4}{4(1-T)^2} - \frac{\ef^6}{4(1- \tmp)^4} + O(\ef^8) \quad &\text{for $0< T<1$.} 
	\end{dcases}
\eeq
On the other hand, as $h\to\infty$, 
\beq \label{eq:cpvalueforhplh}
	\cp_0 =  h + \frac{T}2 + O(h^{-1}) \qquad \text{for all $T>0$.} 
\eeq 

\end{lemma}

\begin{proof}
Consider the limit of $\cp_0$ as $h\to 0$. 
For $T>1$, the equation for $\cp_0$ becomes $1-Ts_1(\cp_0)=0$ when $h=0$, and its solution is $T+T^{-1}$. A simple perturbation argument applied to the equation for small $h$ implies the result. 
For $0<T<1$, we use the asymptotics
\beqq
	s_2(z) = \frac{1}{2\sqrt{z-2}} + O(1) \quad \text{and}\quad s_1(z)= 1 + O(\sqrt{z-2})  \quad \text{as $z\to 2$,}
\eeqq
which follow from the formulas in \eqref{eq:s_1}. Then, the equation for $\cp_0$ becomes 
\beq
	1 -T -\frac{\ef^2}{2\sqrt{\cp_0-2}} + O(\ef^2)+O(\sqrt{\cp_0-2})=0
\eeq
as $h\to 0$ and $\cp_0\to 2$. From this equation we find the result as $h\to 0$. The limit as $h\to \infty$ follows from $s_k(z)= z^{-k}+O(z^{-k-1})$ as $z\to \infty$. 
\end{proof}

\subsubsection{High temperature case, $T>1$}

From \eqref{eq:cpohs}, we find that for $T>1$, as $h\to 0$, 
\beqq \begin{split} 
	&s_0(\cp_0)= \frac1{2T^2} + \log T + \frac{\ef^2}{\tmp^2} + \bhp{\ef^4}, \quad s_1(\cp_0)= \frac1T - \frac{\ef^2}{T(T^2 - 1)} + \bhp{\ef^4}, \\
 	&s_2(\cp_0) =\frac{1}{\tmp^2 - 1} - \frac{2\tmp^2\ef^2}{(\tmp^2 - 1)^3} + \bhp{\ef^4}.
\end{split} \eeqq
Inserting the formulas into \eqref{eq:felim_h>0}, 
\beq \label{eq:feythl}
	\felim(\tmp, h) = \frac1{2\tmp} + \frac{\ef^2}{2\tmp} - \frac{h^4}{4\tmp(\tmp^2-1)} + O(h^6).
\eeq
Therefore, we find that if we first take $N\to\infty$ with fixed $h>0$ and then let $\ef\to 0$, then 
\beq \label{eq:freeegnhht}
	\fe_N(\tmp, h) \simeqids  \left[ \frac1{2\tmp} + \frac{\ef^2}{2\tmp} - \frac{h^4}{4\tmp(\tmp^2-1)} \right] + \frac{\ef^2}{\sqrt{2N (\tmp^2-1)}} \NN \left(0, 1 \right) 
\eeq
where the terms of orders $\ef^6$ and $\ef^4N^{-1/2}$ have been dropped. The fluctuations are of order $\frac{h^2}{\sqrt{N}}$. 
On the other hand, when $\ef=0$, the fluctuations are of order $N^{-1}$ (see \eqref{eq:higtemzeef}).
These two terms are of same order when $h\sim N^{-1/4}$.

\subsubsection{Low temperature case, $T<1$}

Using the $T<1$ case of \eqref{eq:cpohs}, the leading term \eqref{eq:felim_h>0} becomes 
\beq \label{eq:FtHhz} \begin{split}
	\felim(\tmp, h) 
	= 1 - \frac{3T}{4}  + \frac{ T\log T}{2} + \frac{\ef^2}{2}- \frac{\ef^4}{8(1 - T)} + O(h^6) 
\end{split} \eeq
and the variance of the Gaussian distribution in \eqref{eq:fluct_h>0} becomes $\frac{h^4 s_2(\cp_0)}{2} =\frac{h^2(T-1)}{2} + O(h^4)$. 
Thus, from \eqref{eq:fluct_h>0}, for $T<1$, we find that if we take $N\to \infty$ first and then take $h\to 0$, then 
\beq \label{eq:formallimiofhpaaaaa}
	\fe_N(\tmp, h) \simeqids  \left[ 1 - \frac{3T}{4}  + \frac{ T\log T}{2} + \frac{\ef^2}{2}- \frac{\ef^4}{8(1 - T)} \right] + \frac{1}{\sqrt{N}} \NN\left(0, \frac{\ef^2(1 - T)}{2}\right) 
\eeq 
where the terms of orders $h^6$ and $h^3N^{-1/2}$ have been dropped.
This implies that the fluctuations of the free energy are of order $\frac{h}{\sqrt{N}}$. 
On the other hand, when $\ef=0$, the fluctuations are of order $N^{-2/3}$ (see \eqref{eq:fe_h0_ltemp}).
These two terms are of same order when $h \sim N^{-1/6}$. 

\subsubsection{Summary}

In summary, a heuristic matching computation suggests that the transitional scaling is 
\beq \begin{split}
	h &= O(N^{-1/4}) \quad \text{for $T>1$,}\\
	h &= O(N^{-1/6}) \quad \text{for $T<1$.}
\end{split} \eeq
In next two sections, we compute the fluctuations of the free energy in the above transitional regimes.

\section{Free energy for $T>1$ and $\ef\sim N^{-1/4}$}\label{sec:fe1/4}

\subsection{Analysis}\label{sec:fe1/4analysis}

Assume that $T>1$ and set  
\beq
	\ef = \efres N^{-1/4}
\eeq 
for fixed $\efres > 0$.
In this case, using the notations \eqref{eq:defoflsl} and \eqref{eq:go00}, 
\beq
	\G(z)= \beta z - s_0(z) -  \frac{\lsl_N(z)}{N}  +  \frac{\efres^2\beta}{\sqrt{N}} \left[ s_1(z)+ \frac{\go_N(z; 1)}{\sqrt{N}}  \right] + \bhp{N^{-3/2}}
\eeq
where we recall that $\lsl_N(z)$ and $\go_N(z;1)$ are $\bhp{1}$ for $z>2$. 
We approximate the function by $\G_0(z)= \beta z- s_0(z)$ and, as we discussed in sub-subsection \ref{sec:seconhitzef}, this function has the critical point $\cp_0=\beta+\beta^{-1}$ for $T>1$. 
Applying a perturbation argument (see Appendix \ref{sec:pert}) and using the formulas of $s_0(z)$ and $s_1(z)$, 
the critical point of $\G(z)$ is given by 
\beq
	\cp = \cp_0 + \bhp{N^{-1/2}}  \qquad \text{with $\cp_0 = \beta + \beta^{-1}$.} 
\eeq
Furthermore, 
\beq \label{eq:GcphNhalf}\begin{split}
	\G(\cp) & = \frac{\beta^2}2 + 1 + \log \beta  + \frac{\efres^2 \beta^2}{\sqrt{N}} + \frac1N \left[ - \frac{\efres^4 \beta^4}{2(1 - \beta^2)} + \efres^2 \beta \go_N(\cp_0;1) -  \lsl_N(\cp_0) \right] + \bhp{N^{-3/2}} .
\end{split}
\eeq
Since 
\beq
	\G''(\cp)=\frac1N\sum_{i=1}^N\frac{1}{(\cp-\lambda_i)^2}+\frac{2H^2\beta}{N^{3/2}}\sum_{i=1}^N\frac{n_i^2}{(\cp-\lambda_i)^3}
	\simeq s_2(\gamma)+\frac{2H^2\beta}{N^{1/2}}s_3(\gamma)
	\simeq s_2(\gamma_0)
\eeq
and $\G^{(k)}(\gamma)=O(1)$ for all $k\geq2$, the method of steepest descent implies that 
\beq
	\int_{\cp - \ii \infty}^{\cp + \ii \infty}  e^{\frac N2 (\G(z) - \G(\cp))} \dd z
	\simeq \frac{\ii}{N^{1/2}}\sqrt{\frac{4\pi}{s_2(\cp_0)}} \asymp \bhp{N^{-1/2}}.
\eeq

\begin{result}
For $\ef=\efres N^{-1/4}$ with fixed $H>0$ and $T>1$, 
\beq \label{eq:freenghitr} \begin{split}
	\fe_N(T, \ef)  = &\frac1{4T} + \frac{\efres^2}{2T \sqrt{N}} + \frac{T}{2 N} \left[ \log(1 - T^{-2}) - \frac{\efres^4 }{2T^2(T^2-1)} + \frac{\efres^2}{T} \go_N(\cp_0;1) -  \lsl_N(\cp_0) \right] 
\end{split} \eeq
plus $\bhp{N^{-3/2}}$, 
as $N\to \infty$ with high probability where $\lsl_N(z)$ and $\go_N(z;1)$  are defined in \eqref{eq:defoflsl} and \eqref{eq:go00}, respectively, and $\cp_0= \cp_0(h=0)=T+T^{-1}$. 
\end{result}

The sample random variables $\go_N(\cp_0;1)$ and $\lsl_N(\cp_0)$ both converge to Gaussian distributions. 
Since $\go_N(\cp_0;1)$ depends only on $n_i$'s and $\lsl_N(\cp_0)$ depends only on $\eg_i$'s, these two random variables are independent. 
Therefore, we obtain the following result. 

\begin{result}\label{thm:fe1/4T>1} 
For $\ef=HN^{-1/4}$ and $T>1$, as $N\to \infty$, 
\beq
\label{eq:fluct_htemp_tran}
	\fe_N(T, \ef) \simeqids  \left[ \frac{1}{4T} + \frac{\efres^2}{2\tmp\sqrt{N}} \right] + \frac{T}{2N} \NN(-\alpha, 4\alpha), 
	\qquad \mn:= \frac{\efres^4}{2\tmp^2(\tmp^2-1)} - \frac12 \log(1-\tmp^{-2}) . 
\eeq
\end{result}

\subsection{Matching with $\ef = 0$ and $\ef =O(1)$ cases}

If we set $\efres=0$ in \eqref{eq:freenghitr}, we recover the result \eqref{eq:higtemzeef} for the case of $\ef=0$. 
We now consider the limit $\efres \rightarrow \infty$. 
If we formally set $\efres=hN^{1/4}$ in \eqref{eq:freenghitr} with $h$ small but fixed and $N$ large, then we have 
\beq \label{eq:htranlimHl}
	\fe_N(T, \ef) \simeq \frac1{4\tmp} + \frac{h^2}{2\tmp} - \frac{h^4}{4\tmp(\tmp^2-1)} + \frac{h^2}{2\sqrt{N}}  \go_N(\cp_0;1) 
\eeq
for asymptotically almost every disorder sample. 
This is the same as \eqref{eq:feh>101} when $h\to 0$ since  $\felim(\tmp, h)$ satisfies \eqref{eq:feythl} as $h\to 0$. 
Therefore, \eqref{eq:freenghitr} matches well with both regimes.

\section{Free energy for $T<1$ and $\ef\sim N^{-1/6}$}\label{sec:fe1/6}

\subsection{Analysis} \label{sec:fetranltmp}

Assume that $0<\tmp <1$ and we set 
\beq \label{eq:hNtl1a6q}
	h= \efres N^{-1/6}
\eeq
for fixed $\efres>0$.
We find the critical point $\cp>\eg_1$. Previously we had $\cp=\eg_1+\bhp{N^{-1}}$ when $h=0$ and $\cp=\eg_1+\bhp{1}$ when $h>0$. For $h\sim N^{-1/6}$, we make the ansatz  
\beq 
\label{eq:cpnots_N16}
	\cp = \eg_1 + sN^{-2/3}
\eeq
and find $s>0$ assuming that $s=\bhp{1}$. From the equation $\G'(\cp)=0$, see
\eqref{eq:dG}, the equation of $s$ is 
\beq \label{eq:seqgeneral}
	\beta-\frac1{N^{1/3}} \sum_{i=1}^N\frac{1}{s+a_1-a_i}-h^2\beta N^{1/3} \sum_{i=1}^N\frac{n_i^2}{(s+a_1 -a_i)^2} =0
\eeq
where we recall $a_i=N^{2/3}(\eg_i-2)$. Here, we did not change $h$ to $HN^{-1/6}$ since we will cite this equation in several places in the paper. 
From \eqref{lem:specialsum}, the second sum converges with high probability. 
The first sum is $1+\bhp{N^{-1/3}}$ from \eqref{eq:res_edge}. 
Thus, with $h=HN^{-1/6}$ the equation becomes,  under the assumption that $s=\bhp{1}$,
\beq \label{eq:bsao13}
	\beta-1 -H^2\beta\sum_{i=1}^N\frac{n_i^2}{(a_1+s -a_i)^2} + \bhp{N^{-1/3}}=0. 
\eeq
Let $\stild$ be the solution of the equation 
\beq\label{lem:fe1/6s0bound}
	\beta - 1 - \efres^2 \beta \sum_{i = 1}^N \frac{n_i^2}{(\stild + \egres_1 - \egres_i)^2} =0, \qquad \stild>0. 
\eeq
Using the rigidity, we can show that $\stild\asymp \bhp{1}$
with high probability. From this, comparing the equations for $s$ and $\stild$, we find that 
\beq
		s= \stild+ \bhp{N^{-1/3}}.
\eeq
which is consistent with the ansatz. 
The last equation can also be verified by checking the inequalities  
\beqq
	\G'(\lambda_1+\stild N^{-2/3}(1-N^{-\e}))<0<\G'(\lambda_1+\stild N^{-2/3}(1+N^{-\e}))
\eeqq
for any $0<\epsilon<1/3$. 

We now evaluate $G(\cp)$ which is given by 
\beq \label{eq:Gfortrnlowmt}
	\G(\cp)= \beta \cp - \frac1{N} \sum_{i=1}^N \log(\cp-\lambda_i) + \frac{\efres^2\beta}{N^{4/3}} \sum_{i=1}^N \frac{n_i^2}{\cp-\lambda_i}.
\eeq
Insert $\cp=\eg_1+ sN^{-2/3}= 2+ (a_1+s)N^{2/3}$. 
By \eqref{eq:smoflz}, the sum involving the log function becomes 
\beqq
	\frac1{N} \sum_{i=1}^N \log (\cp-\eg_i) = \frac12 + N^{-2/3}(\egres_1+s) + \bhp{N^{-1}} . 
\eeqq
The other sum is equal to 
\beqq
	\frac{\efres^2\beta}{N^{2/3}} \sum_{i=1}^N \frac{n_i^2}{a_1+s - a_i} = \frac{\efres^2\beta}{N^{2/3}} \left(N^{1/3}+ \crv(s) \right) 
\eeqq
using the random variable $\crv(w)$ defined by \eqref{eq:crvdnf}, which is  $\bhp{1}$ outside of a set whose probability shrinks to zero.
Thus, 
\beq
	\G(\cp) = 2 \beta - \frac12 + \frac{\efres^2 \beta}{N^{1/3}}   + \frac1{N^{2/3}} \left[  (\beta-1) (\egres_1+s) 
+ \efres^2 \beta  \crv(s) \right]  + \bhp{N^{-1}} . 
\eeq

To evaluate the integral in \eqref{eq:fe_stedes}, we observe that 
for $k\ge 2$,
\beqq \begin{split}
	\frac{\G^{(k)}(\cp)}{(-1)^k (k-1)!} = N^{\frac{2k}{3}-1}\sum_{i=1}^N \frac1{(s+\egres_1-\egres_i )^k}  
	+ k N^{\frac23k-\frac23} H^2\beta \sum_{i=1}^N  \frac{n_i^2}{(s+\egres_1-\egres_i )^{k+1}} 
	= \bhp{N^{\frac23k-\frac23}}. 
\end{split} \eeqq
For $k=2$, the leading term is 
\beq
	\G''(\cp) = 2 N^{2/3} H^2\beta \sum_{i=1}^N  \frac{n_i^2}{(s+\egres_1-\egres_i)^3} + \bhp{N^{1/3}} . 
\eeq
Since $\G''(\cp) \sim N^{2/3}$, the main contribution to the integral comes from a neighborhood of radius $N^{-5/6}$ 
near the critical point. By the Taylor series, for $u=O(1)$, 
\beq
	N\left( \G(\cp+ u N^{-5/6} ) - \G(\cp) \right) 
	= \sum_{k=2}^\infty  \frac{N^{1-\frac56k}}{k!} \G^{(k)}(\cp) u^k
	= H^2\beta \left( \sum_{i=1}^N  \frac{n_i^2}{(s+\egres_1-\egres_i)^3}\right) u^2 + \bhp{N^{-5/6}} 
\eeq
where all terms but $k=2$ are $\bhp{N^{-5/6}}$. 
Thus, from the Gaussian integral approximation, 
\beq\label{lem:fe1/6int}
	\int_{\cp - \ii\infty}^{\cp + \ii\infty} e^{\frac N2 (\G (z) - G(\cp))} \dd z
	\simeq  
	\frac{1}{N^{5/6}}\int_{-\ii\infty}^{\ii\infty} e^{H^2\beta \left( \sum_{i=1}^N  \frac{n_i^2}{(s+\egres_1-\egres_i)^3}\right) u^2 } \dd u \asymp \bhp{N^{-5/6}}.
\eeq

Combining all together in \eqref{eq:fe_stedes} and replacing $s$ by $\stild$, we obtain the following 

\begin{result}
For $h=H N^{-1/6}$ and $0<T<1$,  
\beq \label{eq:fenfortltm}
	\fe_N =   \Fos(T, h) 
	+ \frac{\ffl(T, H)}{N^{2/3}}   
	+ \bhp{N^{-1}}, 
	\qquad \Fos(\tmp, h)  := 1- \frac{3\tmp}{4} + \frac{\tmp \log \tmp}2 + \frac{\ef^2}{2}, 
\eeq
as $N\to \infty$ for asymptotically almost every disorder sample. Here, 
\beq \label{new} 
	\ffl(T, H) =\frac12 (1-T)(\stild+ \egres_1) + \frac12\efres^2 \crv(\stild)
\eeq
where $\crv(z)$ is defined in \eqref{eq:crvdnf} and $\stild$ is the solution of the equation \eqref{lem:fe1/6s0bound}, 
\beq \label{eq:stildeq}
	1-T= H^2 \sum_{i=1}^N \frac{n_i^2}{(\stild+a_1-a_i)^2}, \qquad \stild>0. 
\eeq
\end{result}

The function $\Fos(T,h)$ is equal to $F(T, h)$ of \eqref{eq:felim_h>0} if we set $\cp_0=2$. 
The order of fluctuations is $N^{-2/3}$ as in the $h=0$ case. 
But the fluctuations depend on all eigenvalues and $n_1, \cdots, n_N$. In contrast, when $h=0$ they depend only on the largest eigenvalue. Using \eqref{eq:def_crv_lim} for $\crv(\stild)$, we obtain the next distributional result. 

\begin{result}\label{thm:fe1/6T<1}
For $h=HN^{-1/6}$ and $0<T<1$, 
\beq \label{eq:fe_ltemp_htran} 
	\fe_N\simeqids \Fos(\tmp, h) + \frac{ (1-\tmp) (\varsigma + \airy_1) + \efres^2\crvlim(\slim) }{2N^{2/3}}
\eeq
as $N\to \infty$, where 
\beq  \label{eq:fluct_ltemp_htran}
		\crvlim(w) = \lim_{n \rightarrow \infty} \Big(\sum_{i = 1}^n \frac{\evg_i^2}{w +\airy_1 - \airy_i} - \frac1{\pi} \int_0^{\left(\frac{3\pi n}{2}\right)^{2/3}} \frac{\dd x}{ \sqrt{x}}\Big) 
\eeq
and $\slim$ is the solution of the equation 
\beq \label{eq:dG_t23_limit}
	1 - T = \efres^2 \sum_{i = 1}^\infty \frac{\evg_i^2}{(\varsigma + \airy_1 - \airy_i)^2} , \qquad \slim>0,
\eeq
where $\airy_i$ is the GOE Airy point process and $\evg_i$ are independent standard normal sample random variables. 
\end{result}

\subsection{Asymptotic behavior of the scaled limiting critical point $\stild$} 
\label{sec:s0smalllarge}

The solution $\stild$ of the equation \eqref{lem:fe1/6s0bound}, 
\beq \label{eq:stildeqabao}
	1-\tmp - \efres^2 \sum_{i=1}^N \frac{n_i^2}{(\stild+\egres_1-\egres_i)^2} =0, \qquad \stild >0,
\eeq
is the scaled limiting critical point that 
is used in the result \eqref{eq:fenfortltm} above. We now describe the behavior of $\stild$ as $H\to 0$ and $H\to\infty$. 
The following result is useful in the next two subsections and in two later sections.

\begin{result}
The solution $\stild$ of the equation \eqref{eq:stildeqabao} satisfies: 
\beq\label{eq:slimH0}
	\stild =  \frac{|n_1|}{\sqrt{1-T}} \efres + O(H^2) \qquad \text{as $H\to 0$}
\eeq
and 
\beq \label{eq:sqrs0whenehf} 
	\sqrt{\stild}  
	\simeq \frac{\efres^2}{2(1 - T)} \left[ 1+  \frac{\efres^2 \go_N\big( 2+\frac{\efres^4N^{-2/3}}{4(1 - \tmp)^2} ; 2 \big)}{(1 - \tmp)N^{5/6}} \right]  
	\qquad \text{as $H\to \infty$.}
\eeq
The second term inside the bracket of the equation \eqref{eq:sqrs0whenehf} 
is $\bhp{H^{-3}}$. 
\end{result}

For the $H\to 0$ limit, we see from the equation \eqref{eq:stildeqabao} that $\stild \to 0$ as $H\to 0$. If we set $\stild= y H$, then separating the term $i=1$, the equation becomes 
$1 - T = \frac{n_1^2}{y^2}  + O(H^2)$. Solving it, we obtain \eqref{eq:slimH0}.

We now consider the large $H$ behavior of $\stild$. 
We write the equation \eqref{eq:stildeqabao} as 
\beq \label{eq:1TH2aq}
	\frac{1-T}{H^2}=\sum_{i=1}^N \frac{n_i^2}{(\stild+\egres_1-\egres_i)^2} 
	= \frac1{N^{4/3}} \sum_{i=1}^N \frac{n_i^2}{(z-\eg_i)^2} , \qquad z= 2+(\stild+a_1) N^{-2/3}. 
\eeq
Note that $\stild\to \infty$ as $H\to \infty$. We evaluate the leading term of the right-hand of the above equation when $z\to 2$ such that $z-2\gg N^{-2/3}$. 
The equation \eqref{eq:weightls} when $k=2$ is 
\beqq
	\frac1{N} \sum_{i = 1}^N \frac{n_i^2}{(z - \eg_i)^2}  = s_2(z) + \frac{\go_N(z;2)}{\sqrt{N}} + \bhp{N^{-1}} 
\eeqq
for $z-2=O(1)$. We expect that this formula is still applicable to $z= 2+(\stild+a_1) N^{-2/3}$ since $\stild\to \infty$. 
Since $z\to 2$, we have $s_2(z) \simeq \frac{1}{2\sqrt{z-2}}$ from \eqref{eq:stjaspt}.  
The equation \eqref{eq:1TH2aq} becomes
\beq \label{eq:1ahqqnq}
	\frac{1-T}{H^2} 
	\simeq 
	\frac1{2N^{1/3}\sqrt{z-2}} + \frac{\go_N(z;2)}{N^{5/6}}. 
\eeq
The sample expectation of $\go_N(z;2)$ with respect to $n_i$s is $0$ and the variance is
\beqq
		\E_s[\go_N(z;2)^2] = \frac{2}{N} \sum_{i=1}^N \frac1{(z-\widehat \eg_i)^4} \simeq 2 s_4(z) \simeq \frac{1}{8(z-2)^{5/2}}
\eeqq
from \eqref{eq:stjaspt}. Thus, we expect that $\go_N(z;2)= \bhp{(z-2)^{-5/4}}$ as $z\to 0$ and \eqref{eq:1ahqqnq} becomes
\beqq
	\frac{1-T}{H^2} \simeq \frac1{2\sqrt{t}} +  \frac{\go_N(2+tN^{-2/3};2)}{N^{5/6}}
	\simeq \frac1{2\sqrt{\stild}} + \bhp{\stild^{-5/4}}. 
\eeqq
Solving it gives $\stild \simeq   \frac{\efres^4}{4(1 - \tmp)^2}$, the leading term of \eqref{eq:sqrs0whenehf}, 
as $H\to\infty$. 
Inserting it bask to the same equation, we obtain the next term and obtain \eqref{eq:sqrs0whenehf}. 
The last computation also shows that the second term in the bracket of \eqref{eq:sqrs0whenehf} 
is $\bhp{H^2t^{-5/4}}=\bhp{H^{-3}}$.

\subsection{Matching with $h=0$} \label{sec:comph0ltmp}

We show that a formal limit \eqref{eq:fenfortltm} as $H\to 0$ agrees with \eqref{eq:fe10ds} which is the result for $h=0$. 
The leading term satisfies
\beq
	\Fos(T, h) = 1 - \frac{3\tmp}{4}  + \frac{\tmp\log \tmp}{2} + O(H^2N^{-1/3}) .
\eeq
For the subleading term 
\eqref{new}, 
we use \eqref{eq:slimH0} for $\stild$ and find that
\beq
\label{eq:crvlimh0}
	\crv(\stild) 
	= \frac{n_1^2}{\stild} +  \sum_{i = 2}^N \frac{n_i^2}{\stild+a_1-a_i} -N^{1/3}
	= \frac{|n_1|\sqrt{1-T}}{H} + \bhp{1} 
\eeq
where the $\bhp{1}$ term follows from \eqref{eq:res_edge}. 
Therefore, if we set $h=HN^{-1/6}$ and take the limits $N\to \infty$ first and $H\to 0$ second, then
\beq \label{eq:ampqfe}
	\fe_N(\tmp, h) =  1 - \frac{3\tmp}{4}  + \frac{\tmp\log \tmp}{2}  + \frac{1-T}{2N^{2/3}} a_1 + O\left( \efres^2 N^{-1/3}\right) + \bhp{\efres N^{2/3}}
\eeq
for asymptotically almost every disorder sample. This agrees with result \eqref{eq:fe10ds} obtained when $h=0$. 

We remark that the two subleading terms in \eqref{eq:ampqfe} are comparable in size when $H= O(N^{-1/3})$, or equivalently when $h=O(N^{-1/2})$. 
This regime is not important for the computation of the free energy, but it will become important when we discuss the overlap of the spin variable with the external field in Subsection \ref{sec:mgnlowtmn12}.

\subsection{Matching with $h>0$} \label{sec:fenglmhp}

We show that the formal limit of \eqref{eq:fenfortltm} as $\efres \rightarrow \infty$ is consistent with the result \eqref{eq:felim_h>0} for $h>0$.

\subsubsection{Large $w$ limit of $\crv(w)$}

We first consider the behavior of $\crv(w)$, defined in \eqref{eq:crvdnf}, as $w\to \infty$ and then we insert $w=\stild$ which tends to $\infty$ from \eqref{eq:sqrs0whenehf}. 
This result is also used in other sections later. 

\begin{result}
As $w\to \infty$, 
\beq \label{eq:crvwlargew}
		\crv(w) \simeq -\sqrt{w} +\frac{\go_N(W; 1)}{N^{1/6}} + \bhp{w^{-1/2}}, \qquad W:=2+wN^{-2/3}. 
\eeq
where $\go_N(z;k)$ is defined in \eqref{eq:go00}. 
\end{result}

Let $\widehat \egres_i := N^{2/3}(\widehat \eg_i - 2)$ be the scaled classical location of the eigenvalues. Write 
\beq \label{eq:cvvrslt}
	\crv(w)= \sum_{i=1}^N \frac{n_i^2}{w- \widehat \egres_i}-N^{1/3} + \sum_{i=1}^N \frac{n_i^2 (a_i- \widehat \egres_i-a_1)}{(w+\egres_1-\egres_i)(w- \widehat \egres_i)} .
\eeq
Since $a_i\asymp -i^{2/3}$, we find that for any $\epsilon>0$, 
\beqq
	 \sum_{i=1}^N \frac{1}{(w-\egres_i)^2} \le \frac1{w^{1/2-\epsilon}} \sum_{i=1}^N \frac{1}{(w-\egres_i)^{3/2+\epsilon}}
	 =\bhp{w^{-1/2}}
\eeqq
as $w\to \infty$. Thus, considering in a similar way, the last sum in \eqref{eq:cvvrslt} is $\bhp{w^{-1/2}}$ 
since $\egres_1=\bhp{1}$, $\egres_i - \widehat \egres_i = \bhp{1}$, and $w\to \infty$. 
Setting $W = 2 + w N^{-2/3}$, \eqref{eq:cvvrslt} can be written as 
\beqq
	\crv(w) 
	= N^{1/3}\left[ \frac1{N} \sum_{i = 1}^N \frac{1}{W- \widehat \eg_i} - 1 \right] +\frac{\go_N(W; 1)}{N^{1/6}} + \bhp{w^{-1/2}}. 
\eeqq
From a formal application of the semicircle law, 
\beqq \label{eq:specsum1exp}
	\frac{1}{N}\sum_{i = 1}^N \frac{1}{W -\widehat \eg_i} \simeq s_1(W) 
	= 1 - \sqrt{W-2} +O(W-2) = 1 - \frac{\sqrt{w}}{N^{1/3}} + O(w N^{-2/3}). 
\eeqq
Thus, we obtain \eqref{eq:crvwlargew}.

\medskip

The equations \eqref{eq:sqrs0whenehf} 
and \eqref{eq:crvwlargew} imply the next result. 

\begin{result}
Let $\stild$ be the solution of \eqref{eq:stildeq}.
Then, as $H\to \infty$, 
\beq \label{eq:jsjsjsj}
	\crv(\stild) \simeq  - \frac{\efres^2}{2(1 - T)}-   \frac{\efres^4 \go_N(\Gamma_0; 2)}{2(1 - \tmp)^2 N^{5/6}}    +\frac{\go_N(\Gamma_0; 1)}{N^{1/6}},\qquad  \Gamma_0=2+\frac{\efres^4N^{-2/3}}{4(1 - \tmp)^2}.
\eeq
\end{result}

\subsubsection{Large $H$ limit}

From \eqref{eq:jsjsjsj}, we see that  the $N^{-2/3}$ term in \eqref{eq:fenfortltm} satisfies 
\beq \label{eq:FtilLargeH}
	\frac{\ffl(\tmp, H)}{N^{2/3}}
	\simeq  \frac{(1-T)a_1}{2N^{2/3}} - \frac{H^4}{8(1-T)N^{2/3}} + \frac{H^2 \go_N(\Gamma_0; 1)}{2N^{5/6}} 
	\simeq - \frac{h^4}{8(1-T)}+ \frac{h^2 \go_N(\Gamma_0; 1)}{2\sqrt{N}}
\eeq
writing in terms of $h=HN^{-1/6}$. 
Thus, we find that if we take $h=HN^{-1/6}$ and $N\to\infty$ and then take $H\to \infty$, then 
\beq \label{eq:feninh16m}
	\fe_N \simeq \left[ 1 - \frac{3T}{4}  + \frac{ T\log T}{2} + \frac{\ef^2}{2} - \frac{\ef^4}{8(1 - T)} \right]+
	 \frac{h^2 \go_N(\Gamma_0; 1)}{2\sqrt{N}},\qquad  \Gamma_0=2+\frac{\efres^4N^{-2/3}}{4(1 - \tmp)^2}
\eeq
for asymptotically almost every disorder sample. 
The point $\Gamma_0$ is approximately equal to $\cp_0$. 
The terms is the bracket are the same as the limit of $F(T,h)$ as $h\to 0$ in \eqref{eq:FtHhz}. 
The $O(N^{-1/2})$ term in \eqref{eq:FtHhz} agrees with the last term of \eqref{eq:feninh16m} since $\cp_0\simeq 2+ \frac{\ef^4}{4(1-T)^2}=\Gamma_0$ from \eqref{eq:cpohs}.
Hence, we find that the above formula is the same as the formal $h\to 0$ limit of the result \eqref{eq:felim_h>0}, which was obtained by taking $N\to \infty$ first with $h=O(1)$ fixed. 
Hence, the result matches with the $h=O(1)$ regime. 

\medskip

The last term of \eqref{eq:feninh16m} depends on the disorder sample. We consider its sample distribution and show that the sample distributions of the  $h=HN^{-1/6}$ regime and $h>0$ regime match for $0<T<1$.  
Using \eqref{eq:go}, we replace $\go_N(\Gamma_0; 1)$ by  $\NN(0; 2s_2(\Gamma_0))$. Using $s_2(z) \simeq \frac{1}{2\sqrt{z-2}}$ as $z\to 2$, we find that
\beq \label{eq:hsgobyGan}
 	\frac{h^2 \go_N(\Gamma_0; 1)}{2\sqrt{N}} \simeqids \frac{h \sqrt{1-T}}{\sqrt{2 N}} \NN\left(0,1\right) .
\eeq
The right-hand side is same as the fluctuation term in \eqref{eq:formallimiofhpaaaaa}, which shows the matching. 
This computation shows the matching of $h=HN^{-1/6}$ regime and $h>0$ regime for $0<T<1$ in terms of the sample distribution as well. 


\subsection{Comparison with the large deviation result of \cite{FyodorovleDoussal}}


We now compare our results with the large deviation result of \cite{FyodorovleDoussal}. To this aim we first
extend their calculation from $T=0$ to any $0<T<1$, which is straightforward. Denoting by $\E_s$ the sample expectation, we find that 
\beq \label{eq:largedecfort1}
	\E_s[\pat_N^n] = \E_s[ e^{\beta N n \fe_N} ] \simeq e^{\beta N n F^0} e^{N2^6\ef^6 G(\frac{\beta n}{8\ef^2})}
\eeq
where $F^0$ is the same as the terms in the bracket in \eqref{eq:feninh16m}, the sample-independent terms, and  
\beq
	G(x) = \frac{(1-T)^3}{3} x^3+ \frac{1-T}{4} x^2. 
\eeq
This formula is valid for fixed $T<1$, $n$, and $\ef$ to the leading order as $N\to \infty$ and in a second stage as $n, \ef\to 0$ so that $\frac{n}{\ef^2}$ is fixed. 
The full result for fixed $n$ and $\ef$ is in (94) and (95) of \cite{FyodorovleDoussal} and the above formula follows from it after changing $T\to 2T$, $\sigma\to 2\ef$, and $J_0=2$. 
Note that the term $e^{N2^6\ef^6 G(\frac{ n}{8T \ef^2})}$ is $O(1)$ when $\ef=O(N^{-1/6})$ and $n=O(\ef^2)= O(N^{-1/3})$. 
We have 
\beq \label{eq:pqpq}
	N2^6\ef^6 G\left( \frac{ n}{8T \ef^2} \right) 
	= \frac{N (1-T)^3 n^3}{24 T^3} + \frac{N\ef^2(1-T) n^2}{4T^2}. 
\eeq

We compare the above formula with the one obtained using the result \eqref{eq:fenfortltm}. 
From \eqref{eq:fenfortltm}, we find that 
\beq
	\E_s[\pat_N^n] = \E_s[ e^{ \frac{N n}{T} \fe_N} ] \simeq e^{\frac{N n }{T} \Fos(\tmp, h) } \E_s[ e^{ 
	\frac{N^{1/3} n}{T} \ffl(\tmp, H)} ] .
\eeq
Now we let $H\to \infty$. This term was computed in \eqref{eq:FtilLargeH} in which we neglected the contribution from $\egres_1$. 
Including this term, using \eqref{eq:hsgobyGan}, and also noting that $\go_N(z;1)$ and $a_1$ are independent, we obtain 
\beq 
	\E_s[ e^{ \frac{N^{1/3} n}{T} \ffl(\tmp, H)} ]  \simeq 
	e^{- \frac{N^{1/3} n H^4}{8T (1-T)}} e^{ \frac{N^{2/3}n^2H^4}{8T^2 \sqrt{\stild}} }
	\E_s\Big[e^{\frac{N^{1/3}n(1-T)}{2T} a_1 } \Big] .
\eeq
We can replace $\sqrt{\stild} \simeq \frac{H^2}{2(1-T)}$ from \eqref{eq:sqrs0whenehf} in the middle term. 
For the remaining expectation, using the right tail of the GOE Tracy Widom distribution 
$F_1(s) = \PP(\alpha_1<s) \sim \exp(- \frac23 s^{3/2})$, 
\beq
	\E[ e^{ \frac{N^{1/3} n (1-T) }{2T} a_1 } ]
	\simeq \int e^{\frac{N^{1/3} n (1-T) }{2T} a_1 - \frac23 \airy_1^{3/2}} \dd \airy_1
	\simeq \exp\left( \frac13 \Big( \frac{N^{1/3} n (1-T)}{2T} \Big)^3 \right).
\eeq
Combining the calculations together, we find that 
\beq \label{eq:ldlas}
	\E[\pat_N^n]   \simeq e^{\frac{N n }{T} \Fos(\tmp, h) }
	e^{- \frac{N^{1/3} n H^4}{8T (1-T)}} e^{ \frac{N^{2/3}n^2H^2(1-T)}{4T^2 } } 
	e^{\frac{N n^3 (1-T)^3}{24T^3} } .
\eeq
The exponents of the last two factors, upon writing $H=hN^{1/6}$ agree with \eqref{eq:pqpq}. 
Since $F^0= \Fos(\tmp, h)- \frac{\ef^4}{8(1 - T)}$, 
we find that \eqref{eq:ldlas} is the same as \eqref{eq:largedecfort1}. 
This shows that the tail of the typical fluctuations obtained here matches the large deviation tails at the exponential order. 


\section{Overlap with the external field}\label{sec:ext}

The overlap of a spin with the external field is
\beqq
	\mgn =\frac{\efv\cdot \sphv}{N}  .
\eeqq
We study the thermal fluctuation of the overlap for a given disorder sample in several regimes: $h=O(1)$, $h\sim N^{-1/6}$ and $h\sim N^{-1/2}$.
We also consider the magnetization, susceptibility,  and differential susceptibility,
\beqq
	\mage = \langle \mgn \rangle, \qquad \scp= \frac{\mage}{\ef}, \qquad \chid= \frac{\dd \mage}{\dd h}. 
\eeqq

\subsection{Thermal average from free energy} \label{sec:theavfrfe}

Before we discuss the thermal fluctuations of $\mgn$, we first derive the thermal average, the magnetization, from the results for the free energy in two regimes, $h=O(1)$ and $h\sim N^{-1/6}$, using 
\beq
	\mage =\langle \mgn\rangle = \frac{\dd \fe_N}{\dd h} . 
\eeq

\subsubsection*{Case $h=O(1)$:}

For $h>0$ and $T>0$, 
the result \eqref{eq:feh>101} for the free energy implies that 
\beq
	\langle \mgn\rangle    =\frac{\dd \fe_N}{\dd h} \simeq \frac{\dd F(T,h)}{\dd h} + \frac1{2\sqrt{N}} \frac{\dd }{\dd h} (h^2 \go_N(\cp_0;1))
\eeq
for asymptotically almost every disorder sample. 
Using $s_0'(z)=s_1(z)$ and $s_1'(z)=-s_2(z)$, 
\beq
	\frac{\dd F(T,h)}{\dd h} = hs_1(\cp_0) + \frac12 (1- Ts_1(\cp_0)-h^2 s_2(\cp_0)) \frac{d \cp_0}{\dd h} 
\eeq 
However, the equation for $\cp_0$ implies that the second term is zero. 
On the other hand, since $\go_N'(z;1)= - \go_N(z;2)$, 
\beq
	\frac{\dd }{\dd h} (h^2 \go_N(\cp_0;1)) = 2h \go_N(\cp_0;1) - h^2 \go_N(\cp_0;2) \frac{d \cp_0}{\dd h} .
\eeq
Using the equation for $\cp_0$ and $s_2'(z)=-2s_3(z)$, we find that
\beq
	\frac{d \cp_0}{\dd h} = \frac{2h s_2(\cp_0)}{Ts_2(\cp_0)+2h^2 s_3(\cp_0)}.
\eeq
Therefore, we conclude that, for fixed $h>0$ and $T>0$,
\beq \label{eq:thermavehpos}
	\langle \mgn\rangle \simeq  
	\ef s_1(\cp_0) + \frac1{\sqrt{N}} \left[ \ef \go_N(\cp_0;1)- \frac{\ef^3 s_2(\cp_0) \go_N(\cp_0; 2)}{T s_2(\cp_0) + 2\ef^2  s_3(\cp_0)} \right]
\eeq
for asymptotically almost every disorder sample. 

\subsubsection*{Case $h\sim N^{-1/6}$ and $T<1$:}

If we use the result \eqref{eq:fenfortltm} for the free energy when $h=H N^{-1/6}$ and $0<T<1$, we find that 
\beq
	\langle \mgn\rangle    = N^{1/6}  \frac{\dd \fe_N}{\dd H}    \simeq h + \frac{\efres \crv (\stild)}{\sqrt{N}}   
	 + \frac{\left( 1-T  + \efres^2 \crv' (\stild)\right)  }{2\sqrt{N}}  \frac{\dd \stild}{\dd H}
\eeq
for asymptotically almost every disorder sample. 
The formula for $\crv$ is given in \eqref{eq:crvdnf} and 
\beq 
	\crv'(w)= - \sum_{i=1}^N \frac{n_i^2}{(w+\egres_1-\egres_i)^2}.
\eeq
Since $\stild$ satisfies the equation \eqref{eq:stildeq}, we see that the term $1-T  + \efres^2 \crv' (\stild)=0$. Hence, for $h=H N^{-1/6}$ and $0<T<1$, 
\beq \label{eq:thermave16}
	\langle \mgn \rangle   \simeq h + \frac{\efres \crv (\stild)}{\sqrt{N}}  
\eeq
for asymptotically almost every disorder sample. 

In both of these regimes, it turns out that the thermal average is indeed the leading term. However, this calculation does not give us the thermal fluctuation term. 
To obtain that, we use the integral representation of the overlap in the following subsections. 
For the overlap and magnetization, it turns out that there is another interesting regime, $h\sim N^{-1/2}$, for $0<T<1$. 
This is the regime that occurs when the two terms in \eqref{eq:thermave16} have the same order; it was shown in \eqref{eq:crvlimh0} that $\efres \crv (\stild) \simeq \bhp{1}$ as $H\to 0$. See the following subsections for the details.

\subsection{Setup}

We obtain the thermal probability of the overlap by considering the moment generating function $\langle e^{\beta\eta \mgn}\rangle$ with respect to the Gibbs measure \eqref{eq:Gibbsmeasure}.  Here, $\eta$ is the variable for the generating function and we added $\beta$ for the convenience of subsequent formulas. 
It turns out that the thermal fluctuations of $\mgn$ are of order $N^{-1/2}$ in all regimes. Hence, we scale $\eta= \xi \sqrt{N}$ and use $\xi$ as the scaled variable for the moment generating function. 
From Lemma \ref{lem:contour}, we have the following formula: 
\beq \label{eq:ovwefdn}
	\langle e^{\beta \xi \sqrt{N} \mgn}\rangle = e^{\frac N2(\Gmgn(\cpmgn) - \G(\cp))} \frac{\int_{\cpmgn - \ii\infty}^{\cpmgn + \ii\infty} e^{\frac N2 (\Gmgn (z) - \Gmgn(\cpmgn))} \dd z}{\int_{\cp - \ii \infty}^{\cp + \ii \infty} e^{\frac N2 (\G(z) - \G(\cp))} \dd z}
\eeq
where 
\beq 
	\Gmgn(z) = \beta z - \frac1N\sum_{i = 1}^N \log(z - \eg_i) + \frac{(\ef + \frac{\xi}{\sqrt{N}})^2\beta}{N} \sum_{i = 1}^N \frac{n_i^2}{z - \eg_i}.
\eeq
Here, we take $\cpmgn>\eg_1$ to be the critical point of $\Gmgn(z)$ satisfying 
\beq
	\Gmgn'(\cpmgn)=0
\eeq
and we take $\cp>\eg_1$ to be the critical point of $\G(z)$.
The only difference between $\Gmgn$ and $\G$, which we studied extensively in the previous sections, is that $\ef$ is changed to $\ef+ \xi N^{-1/2}$. 

\medskip

We record two formulas that we use below. 
From the explicit formulas for $\Gmgn$ and $\G$, the equation $\Gmgn'(\cpmgn)- \G'(\cp)=0$ implies that 
\beq \label{eq:Gpcpcp} 
\begin{split}
	 (\cpmgn-\cp) & \left[ \frac{1}{N} \sum_{i=1}^N \frac{1}{(\cpmgn-\eg_i)(\cp-\eg_i)}
	+ \frac{h^2\beta }{N} \sum_{i=1}^N  \frac{n_i^2(\cp + \cpmgn-2\eg_i)}{(\cpmgn -\eg_i)^2(\cp-\eg_i)^2} \right] \\
	 = &\left( \frac{2\xi \ef}{N^{3/2}} +\frac{\xi^2}{N^2} \right)  \beta \sum_{i=1}^N \frac{n_i^2}{(\cpmgn-\eg_i)^2} . 
\end{split} \eeq
The other formula that we will need is 
\beq \label{eq:diffGpGoher} \begin{split}
	N( \Gmgn(\cpmgn) - \G(\cp) ) 
	& = -  \sum_{i=1}^N \left[ \log \left( 1+ \frac{\cpmgn-\cp}{\cp-\eg_i} \right) -  \frac{\cpmgn-\cp}{\cp-\eg_i}  \right] 
	+ \ef^2\beta\sum_{i=1}^N  \frac{n_i^2(\cpmgn-\cp)^2}{(\cpmgn-\eg_i)(\cp-\eg_i)^2} \\
&\qquad + \left( \frac{2\xi h}{\sqrt{N}} +\frac{\xi^2}{N} \right)  \beta \sum_{i=1}^N \frac{n_i^2}{\cpmgn-\eg_i}
	=: A_1+A_2+A_3,
\end{split} \eeq
which can be seen using $\Gmgn(\cpmgn) - \G(\cp)= \Gmgn(\cpmgn) - \G(\cp) - \G'(\cp)(\cpmgn-\cp)$.

\subsection{Positive external field: $\ef =O(1)$} \label{sec:mgnefpos}

\subsubsection{Analysis}

Fix $\ef>0$. The critical point $\cp$ of $\G(z)$ is evaluated in subsection \ref{sec:freeenergypof}.
It is shown in \eqref{eq:cp_h>0} that 
\beqq
	\cp = \cp_0 + \cp_1 N^{-1/2} + \bhp{N^{-1}}
\eeqq
where $\cp_0$ and $\cp_1$ are deterministic functions of $h$ and $T$. 
From the formulas for $\G$ and $\Gmgn$, we see that $\Gmgn'(z) = \G'(z) + \bhp{N^{-1/2}}$ for $z>\eg_1+O(1)$ (cf. \eqref{eq:weightls}). 
This implies that $\cpmgn- \cp=\bhp{N^{-1/2}}$.
We need to evaluate the difference precisely. 
From \eqref{eq:Gpcpcp}, we find, using the semicircle law, that 
\beqq \begin{split}
	(\cpmgn-\cp) \left( s_2(\cp) + 2h^2\beta s_3(\cp) + \bhp{N^{-1/2}} \right) = \frac{2\xi h\beta}{\sqrt{N}} s_2(\cp) + \bhp{N^{-1}}.
\end{split} \eeqq
Thus, 
\beq \label{eq:cpmgnandcpD}
	\cpmgn = \cp + \Delta N^{-1/2}, 
	\qquad \Delta = \frac{2 \ef \beta \xi s_2(\cp_0)}{s_2(\cp_0) + 2\ef^2 \beta s_3(\cp_0)} + \bhp{N^{-1/2}}.
\eeq

We now evaluate $N(\Gmgn(\cpmgn) - \G(\cp))$ for \eqref{eq:ovwefdn} via the equation \eqref{eq:diffGpGoher}. 
Using the Taylor expansion of the logarithm function, 
\beq
	A_1= \frac{\Delta^2}{2N} \sum_{i = 1}^N \frac{1}{(\cp - \eg_i)^2} + \bhp{\frac{1}{N^{3/2}}\sum_{i = 1}^N \frac{1}{(\cp - \eg_i)^3}}= \frac{\Delta^2 s_2(\cp) }{2} + \bhp{N^{-1/2}}. 
\eeq
Similarly, 
\beq
	 A_2
	 = \frac{\ef^2 \beta \Delta^2}{N} \sum_{i = 1}^N  \frac{n_i^2}{(\cp - \eg_i)^3} + \bhp{N^{-1/2}}
	  = \ef^2 \beta \Delta^2 s_3(\cp)  + \bhp{N^{-1/2}}.
\eeq
In these two equations, we replaced $\cpmgn$ by $\cp$ . For $A_3$, using \eqref{eq:cpmgnandcpD}  and the notation \eqref{eq:go00}, we have
\beq
\begin{aligned}
	&A_3 
	 = 2\xi\ef  \beta ( s_1(\cpmgn) \sqrt{N}  +  \go_N(\cpmgn;1) )  + \xi^2 \beta s_1(\cpmgn) + \bhp{N^{-1/2}}\\
	& \qquad = 2\xi \ef \beta s_1(\cp) \sqrt{N} + \left[ 2\xi\ef  (\go_N(\cp;1)- s_2(\cp) \Delta) + \xi^2 s_1(\cp) \right] \beta + \bhp{N^{-1/2}} . 
\end{aligned}
\eeq
Combining the three terms and inserting the formulas of $\cp$ and  $\Delta$, 
\beq \label{eq:NGSAAAA}
\begin{split}
	N( \Gmgn(\cpmgn) - \G(\cp) ) 
	=& 2\xi \ef \beta \left[ \sqrt{N} s_1(\cp_0)   -s_2(\cp_0) \cp_1 + \go_N(\cp_0;1) \right] \\
	&\qquad + \xi^2 \left[ \beta s_1(\cp_0)  - \frac{2\ef^2 \beta^2 s_2(\cp_0)^2} {s_2(\cp_0) +2\ef^2\beta s_3(\cp_0)  } \right] + \bhp{N^{-1/2}}.
\end{split}
\eeq

Now we consider the integrals in \eqref{eq:ovwefdn}.
Since $\G^{(k)}(\cp)= \bhp{1}$ for all $k\ge 2$, the method of steepest descent applies. 
It is also straightforward to check that 
\beqq
	\Gmgn''(\cpmgn) = \G''(\cpmgn) + \bhp{N^{-1/2}} = \G''(\cp) + \bhp{N^{-1/2}}.
\eeqq
Hence,  
\beqq
	 \frac{\int_{\cpmgn - \ii\infty}^{\cpmgn + \ii\infty} e^{\frac N2 (\Gmgn (z) - \Gmgn(\cpmgn))} \dd z}{\int_{\cp - \ii \infty}^{\cp + \ii \infty} e^{\frac N2 (\G(z) - \G(\cp))} \dd z}
	 \simeq \sqrt{\frac{\G''(\cp)}{\Gmgn''(\cpmgn)}} \simeq  1.
\eeqq

Inserting the above computations into \eqref{eq:ovwefdn}, moving the term involving $\sqrt{N}$ to the left, replacing $\beta \xi$ by $\xi$, using $\beta=1/T$, and inserting the formula \eqref{eq:cp1h>0} for $\cp_1$, we obtain the following.

\begin{result} For $h=O(1)$ and $T>0$, 
\beq \label{eq:olmg}
	\langle e^{ \xi \sqrt{N} \left( \mgn - \ef  s_1(\cp_0) \right) } \rangle 
	\simeq 
	e^{ \xi \ef \left[ \go_N(\cp_0; 1)  - \frac{\ef^3 s_2(\cp_0) \go_N(\cp_0; 2)}{T s_2(\cp_0) + 2\ef^2  s_3(\cp_0)} \right] + \frac{\xi^2}{2} \left[ T s_1(\cp_0)  - \frac{2T h^2 s_2(\cp_0)^2} {2T s_2(\cp_0) +h^2 s_3(\cp_0)  } \right] }
\eeq
as $N\to \infty$ for asymptotically almost every disorder sample, where $\cp_0>2$ is the solution of the equation \eqref{eq:cp0h>0}
and $\go_N(z;k)$ is defined in \eqref{eq:go00}.
\end{result}

Since the right-hand side is an exponential of a quadratic function of $\xi$, we obtain the following distributional result. 

\begin{result}\label{thm:exth>0}
For $h=O(1)$ and $T>0$, 
\beq \label{eq:mgnphc}
	\mgn 
	\simeqidsgibbs 
	\ef s_1(\cp_0) + \frac1{\sqrt{N}} \left[ \ef \go_N(\cp_0;1)- \frac{\ef^3 s_2(\cp_0) \go_N(\cp_0; 2)}{T s_2(\cp_0) + 2\ef^2  s_3(\cp_0)}
	+ \sigmamgn  \gib N \right]
\eeq
as $N\to \infty$ for asymptotically almost every disorder sample. The thermal random variable $\gib N$ is a standard normal random variable and the term $\sigmamgn>0$ is given by the formula  
\beq \label{eq:varmagh>0}
	\sigmamgn^2=  \tmp s_1(\cp_0)  - \frac{2\tmp \ef^2 s_2(\cp_0)^2} {\tmp s_2(\cp_0) + 2\ef^2 s_3(\cp_0)  } .
\eeq
\end{result}

The thermal average is given by the first three terms and they agree with the formula \eqref{eq:thermavehpos} obtained from the free energy. 

\subsubsection{Discussion on the leading term} \label{sec:mgnhplt}

The leading term 
\beq
	\mgn^0(h, T):=hs_1(\cp_0(h))
\eeq
in \eqref{eq:mgnphc} is deterministic. See Figure \ref{fig:mgnvarvsh} (a) for a graph as a functions of $h$. 
The function $\mgn^0$ satisfies the following properties:
\begin{itemize}
\item For every $T>0$, $\mgnz(h, T)$ is an increasing function of $h$. 
\item As $h\to\infty$, 
\beq \label{eq:qqqqqq1}
	\mgn^0(h, T)= 1- \frac{T}{2h} + O(h^{-2}) \quad\text{for all $T>0$. }
\eeq 
\item As $h\to 0$, 
\beq \label{eq:qqqqqq2}
	\mgn^0(h, T) \simeq 
	\begin{cases}
	\frac{\ef}{T} - \frac{\ef^3}{T(T^2-1)}   \quad &\text{for $T>1$,}\\ 
	\ef - \frac{\ef^3}{2(1 - \tmp)} \quad & \text{for $0<T<1$.} 
	\end{cases}
\eeq
\end{itemize}

The first property is consistent with the intuition that the overlap of the spin with the external field becomes larger as the external field becomes stronger. 
The proof follows from 
\beq
	\frac{\dd}{\dd h} \mgn^0= s_1(\cp_0)- hs_2(\cp_0) \cp_0' = \frac{Ts_1(\cp_0)s_2(\cp_0) + 2h^2(s_1(\cp_0)s_3(\cp_0)-s_2(\cp_0)^2)}{Ts_2(\cp_0)+2h^2 s_3(\cp_0)}
\eeq
and from checking that $s_1(z)s_3(z)- s_2(z)^2 >0$ for all $z>2$ using \eqref{eq:s_1}. 
The large $h$ and small $h$ limits follow from Lemma \ref{lem:cp0behro}.

\begin{figure}[H]
\centering
\begin{subfigure}{0.45\textwidth}
\includegraphics[width = 0.9\textwidth]{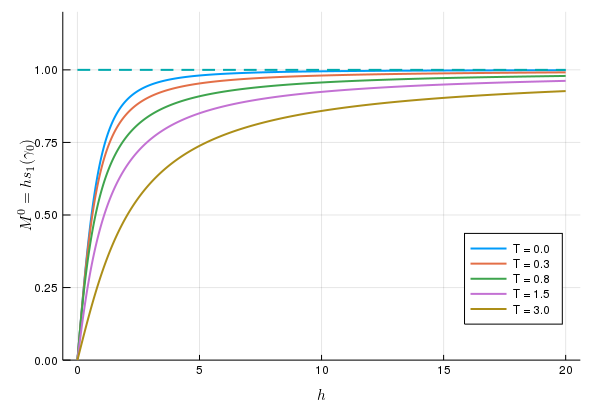}
\caption{$\mgn^0$}
\label{fig:m0_vs_h}
\end{subfigure}
\begin{subfigure}{0.45\textwidth}
\includegraphics[width = 0.9\textwidth]{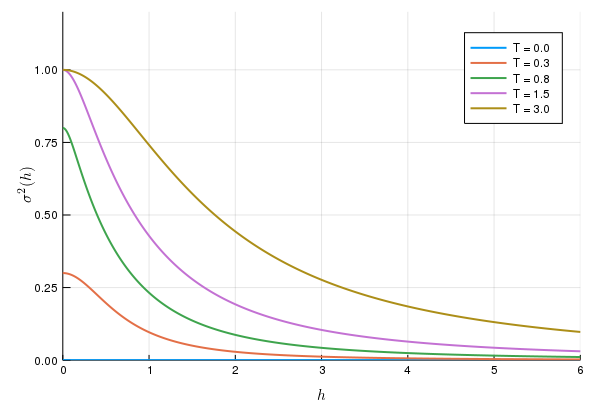}
\caption{$\sigmamgn^2$}
\label{fig:var_vs_h}
\end{subfigure}
\caption{Graph of $\mgn^0$ and $\sigmamgn^2$ as function of $\ef$ for various values of  $\tmp$}
\label{fig:mgnvarvsh}
\end{figure}

\subsubsection{Discussion on the variance}


The variance of the overlap satisfies
\beq
	\langle \mgn^2 \rangle - \langle \mgn\rangle^2 \simeq \frac{\sigmamgn^2}{N} .
\eeq
The term $\sigmamgn^2(h, T)=\sigmamgn^2$ is given in \eqref{eq:varmagh>0} and does not depend on the disorder sample. See Figure \ref{fig:mgnvarvsh} for the graph. Here are some properties of $\sigmamgn^2$. 

\begin{itemize}
\item For every $T$, $\sigmamgn^2(h,T)$ is a decreasing function. 
\item As $h\to\infty$, 
\beq \label{eq:qqqqqq3}
	\sigmamgn^2 = \frac{T}{h} + O(h^{-2}) \quad\text{for all $T>0$. }
\eeq 
\item As $h\to 0$, 
\beq \label{eq:qqqqqq4}
	\sigmamgn^2 \to  
	\begin{cases}
	1   \quad &\text{for $T>1$, }\\ 
	T \quad & \text{for $0<T<1$.} 
	\end{cases}
\eeq
\end{itemize}

The first property follows from \beq \label{eq:sigmaderz}
	\frac{\dd }{\dd \ef} \sigmamgn^2 = - \frac{T^2 s_2(\cp_0) \left[ \left( T s_2(\cp_0)^2- 12 \ef^4 s_3(\cp_0)^2 + 12 \ef^4 s_2(\cp_0)s_4(\cp_0) \right) \cp_0'(\ef)
	+ 4\ef T s_2(\cp_0)^2\right]  }{(\tmp s_2(\cp_0) + 2\ef^2 s_3(\cp_0))^2}
\eeq
by checking that $s_2(z) s_4(z)- s_3(z)^2>0$ for all $z>2$. 
The large and small $h$ limits follow from Lemma \ref{lem:cp0behro}. 

\subsubsection{Limit as $h\to \infty$} \label{sec:magnhla}


Consider the formal limit of the result \eqref{eq:mgnphc} as $\ef\to \infty$. 
Using \eqref{eq:cpvalueforhplh}, we have 
\beq \label{eq:isthiscorrect}
	h^k \go_N(\cp_0 ;k) = \frac{1}{\sqrt{N}} \sum_{i = 1}^N \frac{n_i^2 - 1}{(\frac{\cp_0}{h} - \frac{\widehat \eg_i}{h})^k}
	\simeq \frac{1}{\sqrt{N}} \sum_{i = 1}^N (n_i^2 - 1)
\eeq
and $s_k(\cp_0)\simeq h^{-k}$ as $h\to\infty$. 
Therefore, using \eqref{eq:qqqqqq1} and \eqref{eq:qqqqqq3},  we find that if we take $N\to \infty$ with $h>0$ and then take $h\to \infty$, we get
\beq \label{eq:malha}
	\mgn \simeqidsgibbs 1- \frac{T}{2h}+ \frac1{\sqrt{N}} \left[ \frac{\sum_{i = 1}^N (n_i^2-1) }{2\sqrt{N}} 
	+  \frac{\sqrt{T}}{\sqrt{\ef}} \gib N \right] . 
\eeq
The leading term $\mgn\simeq 1$ is trivial since the spin is likely to be pulled to the direction of the external field if $h$ is large.

\subsubsection{Limit as $\ef \to 0$ when $T>1$} \label{sec:limiofahtp}

Since $\cp_0 \to T+T^{-1}$ as $h\to 0$ for $T>1$ from \eqref{eq:cpohs}, the terms $\go_N(\cp_0 ;1)$ and $\go_N(\cp_0 ;2)$ remain $O(1)$. 
Hence the deterministic terms in the square brackets in \eqref{eq:mgnphc} converge to zero as $\ef\to 0$. 
We thus find, using \eqref{eq:qqqqqq2} and \eqref{eq:qqqqqq4} that, if we take the limit $N\to \infty$ with $\ef>0$ and then take $\ef \to 0$, the result for $T>1$ becomes
\beq \label{eq:724}
	\mgn \simeqidsgibbs \frac{\ef}{T} - \frac{\ef^3}{T(T^2-1)} + \frac1{\sqrt{N}}\left[  \gib N  + \ef \go_N(T+\frac1{T} ;1) \right].  
\eeq

\subsubsection{Limit as $\ef \to 0$ when $T<1$}\label{sec:limiofahtl}

The small $h$ limit \eqref{eq:cpohs}  of $\cp_0$ and the limit of $s_k(z)$ as $z\to 2$ obtained in \eqref{eq:stjaspt} imply that, as $h\to0$,
\beq \label{eq:skwcpa}
	s_2(\cp_0) \simeq \frac{(1 - \tmp)}{\ef^2} + \frac{\tmp}{2(1 - \tmp)}, \qquad s_3(\cp_0) \simeq \frac{(1 - \tmp)^3}{\ef^6}, 
	\qquad s_4(\cp_0) \simeq \frac{2(1 - \tmp)^5}{\ef^{10}} 
\eeq
when $0<T<1$. 
From these, we see that 
\beq \label{eq:mgnphc10101}
	 \ef \go_N(\cp_0;1)- \frac{\ef^3 s_2(\cp_0) \go_N(\cp_0; 2)}{T s_2(\cp_0) + 2\ef^2  s_3(\cp_0)}
	 \simeq  \ef \go_N(\cp_0;1)  - \frac{\ef^5  \go_N(\cp_0; 2)}{2(1 - T)^2}. 
\eeq
Thus,  by \eqref{eq:qqqqqq2} and \eqref{eq:qqqqqq4}, if take the limit $N\to \infty$ with $\ef>0$ and then take $\ef \to 0$, then 
\beq \label{eq:mhnatlh}
	\mgn \simeqidsgibbs \ef - \frac{\ef^3}{2(1 - \tmp)} + \frac1{\sqrt{N}} \left[ \ef \go_N(\cp_0;1)  - \frac{\ef^5  \go_N(\cp_0; 2)}{2(1 - T)^2}
	+  \sqrt{\tmp} \gib N \right] , \quad \cp_0\simeq 2+ \frac{\ef^4}{4(1-T)^2}.
\eeq 
Finally, we consider the terms $\ef \go_N(\cp_0;1)$ and $\ef^5  \go_N(\cp_0; 2)$. The sample-to-sample variance of $\go_N(\cp_0;k)$ is 
\beq \label{eq:goNh0ltmp}
	\frac{2}{N} \sum_{i=1}^N \frac{1}{(\cp_0 - \widehat \eg_i)^{2k}} \simeq 2 s_{2k}(\cp_0),
\eeq
which is expected to hold for $\cp_0-2\gg N^{-2/3}$, i.e., $h\gg N^{-1/6}$, Thus the sample-to-sample variance is $\bhp{\ef^{-2}}$ for $k=1$ and $\bhp{\ef^{-10}}$ for $k=2$ from \eqref{eq:skwcpa}. Hence, we expect that $\ef \go_N(\cp_0;1)$ and $\ef^5  \go_N(\cp_0; 2)$ are $\bhp{1}$ for $h\gg N^{-1/6}$. 

\subsection{No external field: $\ef=0$}

When $\ef=0$ and $T>1$, it is well-known in spin glass theory \cite{panchenko2007overlap, nguyen2018central} that the overlap of two independently chosen spins are asymptotically orthogonal, indicating that the spin variable becomes uniformly distributed on the sphere $\|\sphv\|=\sqrt{N}$ as $N\to\infty$. 
For $\ef=0$ the Gibbs measure is independent of $\efv$. 
Hence, the overlap $\mgn =\frac1{N} \efv\cdot \sphv$ of the spin with the random Gaussian vector $\sphv$ is the cosine of the angle of two independent vectors which are chosen more or less uniformly at random from the sphere.
Thus, we expect that $\mgn$ is approximately $\frac1{\sqrt{N}}$ times a standard normal distribution. 
The formal limit of \eqref{eq:724} as $\ef\to 0$ coincides with this result. 
Indeed when $T>1$, the analysis for $h>0$ with $h=O(1)$ extends to $h\ge 0$ and \eqref{eq:mgnphc} holds. 

When $\ef=0$ and $T<1$, it was argued in \cite{kosterlitz1976spherical} that $ \frac{ \langle | \vu_1\cdot \spin| \rangle }{\sqrt{N}} $ converges to $\sqrt{1-T}$. 
(In  \cite{kosterlitz1976spherical}, the authors claim that $\frac{ \langle  \vu_1\cdot\spin \rangle }{\sqrt{N}} \to \sqrt{1-T}$, but this seems to be a typographical error since $\langle \vu_1\cdot\spin \rangle=0$ due to the symmetry of the Gibbs measure under the transformation $\spin\mapsto -\spin$.) 
It was also proven in \cite{panchenko2007overlap} that the absolute value of the overlap of two independently chosen spins
converges to $1-T$. 
Hence, a spin variable may be written as $\frac{\sphv}{\sqrt{N}} =\pm \sqrt{1-T} \vu_1+ \sqrt{T} \mathbf{v}$, where the unit vector $\mathbf{v}$ is taken uniformly at random from the hyperplane perpendicular to $\vu_1$ and the signs $\pm$ are each taken with probability $1/2$: 
See more discussions on such decomposition of the spin variable in Section \ref{sec:Geometry}. 
Thus, using the notation $n_1=\vu_1\cdot \efv$, we expect that $\mgn\simeq  \frac{\pm n_1 \sqrt{1-T} + \sqrt{T} \gib N}{\sqrt{N}}.$ 
Recall that $\vu_1$ has the sign ambiguity and hence $n_1$ is defined up to its sign.
Thus, we find the following result for $h=0$. 
\begin{result} \label{result:extfld0}
For $h=0$, 
\beq \label{eq:h0folt}
	\mgn\simeq  \begin{dcases}
	\frac1{\sqrt{N}} \gib N \qquad &\text{for $T>1$,} \\
	\frac{ |n_1| \sqrt{1-T} \gib B+ \sqrt{T} \gib N}{\sqrt{N}} \qquad &\text{for $0<T<1$} 
\end{dcases} \eeq
as $N\to \infty$, for asymptotically almost every disorder sample, where $\gib N$ is a standard normal random variable, and $\gib B$ is independent of $\gib N$ and has the distribution $\P(\gib B=1)= \P(\gib B=-1)= \frac12$.
\end{result}

The right-hand side of  \eqref{eq:mhnatlh} involves the thermal random variable $\gib N$ but does not involve the other thermal random variable $\gib B$ in \eqref{eq:h0folt}.
Hence, the formal limit of \eqref{eq:mhnatlh} as $h\to 0$ is not equal to \eqref{eq:h0folt}  when $T<1$. 
This implies that there should be a transitional regime. 
It turns out that there are two transitional regimes, $h\sim N^{-1/6}$ and $h\sim N^{-1/2}$.
The first regime can be expected, since $\cp_0 = 2 + O(h^4)$ as $\ef\to 0$, and the subleading term $O(h^4)$ is of same order  as the fluctuations of the top eigenvalue $\eg_1$ when $h\sim N^{-1/6}$. This is the same transitional regime that was observed for the free energy. 
The second regime $h\sim N^{-1/2}$ arises because the ratio of the integrals in  \eqref{eq:ovwefdn}, which was approximately equal to $1$ when $\ef>0$ (and when $h\sim N^{-1/6}$ as well), is no longer close to $1$ when $h\sim N^{-1/2}$.
This will be responsible for the appearance of $\gib B$. 
We discuss these two transitional regimes in the next subsections. 
We will see in Subsection \ref{sec:mgnlowtmn12} that the result for $h=H N^{-1/2}$ actually holds even when $H=0$, implying that \eqref{eq:h0folt} indeed holds.

\subsection{Mesoscopic external field: $\ef \sim N^{-1/6}$ and $T < 1$}
\label{sec:ext1/6}

\subsubsection{Analysis}

We scale $\ef$ as
\beqq
	\ef= \efres N^{-1/6}
\eeqq
for fixed $H>0$. 
This scale is the same as the one considered in Subsection \ref{sec:fetranltmp}. 
We showed in that section that the critical point of $\G(z)$ is $\cp = \eg_1 + sN^{-2/3}$ where $s>0$ satisfies the equation \eqref{eq:bsao13}. To find the critical point of $\Gmgn(z)$, we make the ansatz that $\cpmgn \simeq \cp$, Then, the equation \eqref{eq:Gpcpcp} becomes 
\beqq
\begin{split}
	& (\cpmgn-\cp) \left[ N^{1/3} \sum_{i=1}^N \frac1{(s+\egres_1-\egres_i)^2}
	+ H^2\beta N^{2/3} \sum_{i=1}^N \frac{2n_i^2}{(s+\egres_1-\egres_i)^3}  \right] \\
	& \simeq \left( \frac{2\xi H}{N^{5/3}} +\frac{\xi^2}{N^2} \right)  \beta N^{4/3} \sum_{i=1}^N \frac{n_i^2}{(s+\egres_1-\egres_i)^2}, 
\end{split} \eeqq
implying that 
\beq
	\cpmgn-\cp = \bhp{N^{-1}},
\eeq
which is consistent with the ansatz. We do not need to determine the $\bhp{N^{-1}}$ term in this subsection.

We now evaluate $N(\Gmgn(\cpmgn) - \G(\cp))$ using  \eqref{eq:diffGpGoher}. From the Taylor series of the log function, 
\beqq
	A_1 \simeq \sum_{i=1}^N \frac{(\cpmgn-\cp)^2}{(\cp-\eg_i)^2} =  \sum_{i=1}^N \frac{(\cpmgn-\cp)^2 N^{4/3}}{(s+a_1-a_i)^2}  
	= \bhp{N^{-2/3}}. 
\eeqq
Inserting $h=HN^{-1/6}$, 
\beqq
	 A_2 
	 \simeq \frac{H^2\beta}{N^{1/3}} \left[ N^2 \sum_{i=1}^N \frac{n_i^2}{(s+a_1-a_i)^2} \right] (\cpmgn-\cp)^2
	 = \bhp{N^{-1/3}}.
\eeqq
The third term is 
\beqq \begin{split}
	A_3 
	&= \left(  2\xi H +\frac{\xi^2}{N^{1/3}} \right)  \beta \left[ \sum_{i=1}^N \frac{n_i^2}{s+\egres_1-\egres_i}
	+ \bhp{N^{-1/3}} \right] .
\end{split} \eeqq
Using the random variable $\crv(s)$ defined in \eqref{eq:crvdnf}, which is $\bhp{1}$, 
and combining all three terms, 
\beq
	N(\Gmgn(\cpmgn) - \G(\cp) ) 
	= 2\xi \efres \beta N^{1/3} +   2\beta\xi \efres  \crv(s)+  \beta\xi^2   + \bhp{N^{-1/3}}.
\eeq

Finally, consider the integrals in \eqref{eq:ovwefdn}. The denominator is computed in Section \ref{sec:fetranltmp}. 
The numerator can be computed in the same manner. Indeed, we can check, as with the denominator, that $\Gmgn^{(k)}(\cpmgn)  = \bhp{N^{\frac23k-\frac23}}$ for all $k\ge 2$ and 
\beq
	\Gmgn''(\cpmgn) = 2 N^{2/3} H^2\beta \sum_{i=1}^N  \frac{n_i^2}{(s+\egres_1-\egres_i)^3} + \bhp{N^{1/2}},
\eeq
which is the same as the denominator. Hence, the Gaussian integral approximations of the integrals imply that 
\beq
	 \frac{\int_{\cpmgn - \ii\infty}^{\cpmgn + \ii\infty} e^{\frac N2 (\Gmgn (z) - \Gmgn(\cpmgn))} \dd z}{\int_{\cp - \ii \infty}^{\cp + \ii \infty} e^{\frac N2 (\G(z) - \G(\cp))} \dd z}
	 \simeq \sqrt{\frac{\G''(\cp)}{\Gmgn''(\cpmgn)}} \simeq 1.
\eeq 

Combining the above computations into \eqref{eq:ovwefdn}, 
replacing $s$ by $\stild$ (the solution to  \eqref{lem:fe1/6s0bound}),  replacing $\beta \xi$ by $\xi$, and using $1/ \beta=T$, we obtain the following result. 

\begin{result}
For $h=HN^{-1/6}$ and $0< T<1$,
\beq \label{eq:mgfmgn}
	\langle e^{ \xi \sqrt{N} \left( \mgn - \ef \right) } \rangle \simeq e^{ \xi \efres \crv(\stild) + \frac{T \xi^2}{2}}, 
	\qquad \crv(\stild) :=  \sum_{i = 1}^N \frac{n_i^2}{\stild + a_1 - a_i} - N^{1/3}, 
\eeq
as $N\to \infty$ for asymptotically almost every disorder sample, where
$\stild>0$ is the solution of the equation  \eqref{eq:stildeq}. 
\end{result} 

Since the exponent of the right-hand side of \eqref{eq:mgfmgn} is a quadratic function of $\xi$, we obtain 

\begin{result}\label{thm:ext1/6} For $h=HN^{-1/6}$ and $0<T<1$,
\beq
\label{eq:mgn_H1/6}
		\mgn \simeqidsgibbs \ef + \frac1{\sqrt{N}} \left[ \efres\crv(\stild)   +  \sqrt{T} \gib N \right]
\eeq
as $N\to \infty$ for asymptotically almost every disorder sample, where the thermal random variable $\gib N$ has the standard Gaussian distribution.  
\end{result}

The thermal average is obtained from the first two terms. The average is the same as \eqref{eq:thermave16} that we obtained from the free energy. 

\subsubsection{Matching with $\ef=O(1)$}\label{sec:16maef0}

We take the formal limit $H\to \infty$ of \eqref{eq:mgn_H1/6}.
The  limit of $\crv(\stild)$ as $H\to \infty$ is obtained in \eqref{eq:jsjsjsj}. From this, we find that,
if we take $\ef = \efres N^{-1/6}$ and let $N\to \infty$ first and then take $H\to\infty$, then 
\beq	\label{eq:bbcc}
	\mgn \simeqids \ef - \frac{\ef^3}{2(1 - \tmp)} + \frac1{\sqrt{N}} \left[  \ef\go_N(\cp_0; 1)  - \frac{\ef^5 \go_N(\cp_0; 2)}{2(1 - T)^2}  +  \sqrt{\tmp} \gib N \right]
\eeq
as $\efres \to \infty$ where $\cp_0\simeq  2+ \frac{\ef^4}{4(1 - \tmp)^2}$. 
This result agrees with \eqref{eq:mhnatlh}, which is obtained by taking $h>0$ fixed and letting $N\to \infty$ first and then taking $h\to 0$.

\subsubsection{Formal limit as $\efres \to 0$}\label{sec:formlimh1/6mng}

We take the formal limit $H\to 0$ of  \eqref{eq:mgn_H1/6}. 
We obtained the limit of $\crv(\stild)$ as $H\to 0$ in \eqref{eq:crvlimh0}. Hence, we find that, if we take $N\to \infty$ with $\ef = \efres N^{-1/6}$ first and then take $\efres \to 0$, then 
\beq
\label{eq:mgnhN1/6H0}
	\mgn \simeqids \ef + \frac{1}{\sqrt{N}} \left[|n_1| \sqrt{1 - T} + \sqrt{\tmp} \gib N\right].
\eeq
This formula evaluated at $H=0$ is different from \eqref{eq:h0folt}.  In particular, the Bernoulli random variable $\gib B(1/2)$ is missing. 
In the next subsection, we consider a new regime $\ef=O( N^{-1/2})$ in which the two terms in \eqref{eq:mgnhN1/6H0} are of the same order.  We will show that this new regime interpolates between $\ef=O(N^{-1/6})$ and $\ef=0$.

\subsection{Microscopic external field: $\ef \sim N^{-1/2}$ and $\tmp < 1$} \label{sec:mgnlowtmn12}

\subsubsection{Analysis}

We set, for fixed $H>0$,  
\beq
	h=HN^{-1/2}.
\eeq
This is a new regime which did not appear in previous sections. The appearance of this scaling regime was first noticed in \cite{FyodorovleDoussal} for the zero temperature case.  

\subsubsection*{Critical points} 

We first compute the critical point $\cp$ of $\G(z)$. In previous sections, we had $\cp=\eg_1+\bhp{N^{-2/3}}$ for $h\sim N^{-1/6}$ and $\cp=\eg_1+ \bhp{N^{-1}}$ for $h=0$. For $h\sim N^{-1/2}$, it turns out that $\cp=\eg_1+ \bhp{N^{-1}}$. We make the ansatz that 
\beq \label{eq:sohca}
	\cp = \eg_1 + \so N^{-1}
\eeq 
with $\so=\bhp{1}$. Then, the critical point equation becomes 
\beq\label{eq:ext1/2sdef}
	\beta  - \frac1{N} \sum_{i=1}^N \frac{1}{\eg_1+\so N^{-1}-\eg_i} - \frac{H^2\beta}{N^2} \sum_{i=1}^N \frac{n_i^2}{(\eg_1+\so N^{-1}-\eg_i)^2} =0 .
\eeq
Separating out $i=1$ in both sums and using \eqref{eq:res_edge} and \eqref{lem:specialsum} for the remaining sums, 
the equation becomes
\beq\label{eq:shN1/2}
	\beta - 1 - \frac{1}{\so} - \frac{\efres^2 \beta n_1^2}{\so^2} + \bhp{N^{-1/3}} =0. 
\eeq
The solution is 
\beq
	\so = \frac{1 + \sqrt{1 + 4(\beta - 1)\efres^2 \beta n_1^2}}{2(\beta - 1)} + \bhp{N^{-1/3}}.
\eeq
Hence, $\so=\bhp{1}$, which is consistent with the ansatz. 

Now consider the critical point of $\Gmgn(z)$. 
Due to the scale $\ef=\efres N^{-1/2}$, the function $\Gmgn(z)$ is the same as $\G(z)$ with $H$ replaced by $H+\xi$. 
Thus, we find that 
\beq
	\cpmgn = \eg_1 + \so_\mgn N^{-1}
\eeq
where $\so_\mgn>0$ solves the equation 
\beq
\label{eq:smgnhN1/2}
	\beta - 1 - \frac{1}{\so_\mgn} - \frac{(\efres + \xi)^2 \beta n_1^2}{\so_\mgn^2} + \bhp{N^{-1/3}} = 0.
\eeq

\subsubsection*{Exponential terms} 

We evaluate $N(\Gmgn(\cpmgn)-\G(\cp))$ using \eqref{eq:diffGpGoher}.
For $A_1$, the sum with $i\ge 2$, using a Taylor approximation, is $\bhp{N^{-2/3}}$. Hence, 
\beqq
  A_1 
	= 	-\log(\frac{\so_\mgn}{\so}) + \frac{\so_\mgn}{\so} - 1 + \bhp{N^{-2/3}}.
\eeqq
The sum with $i\ge 2$ for $A_2$ is $\bhp{N^{-1}}$ and we obtain 
\beqq\begin{split}
	A_2= \frac{\efres^2 \beta n_1^2 (\so_\mgn - \so)^2}{\so_\mgn \so^2}  + \bhp{N^{-1}}.
\end{split} \eeqq
Finally, again separating the term with $i=1$ and using  \eqref{eq:weightedsum1} for the rest of the sum, 
\beq
	A_3 
	= (2\xi H+\xi^2)\beta \left(\frac{n_1^2}{\so_\mgn} +1 \right) +  \bhp{N^{-1/3}}.
\eeq
Therefore, 
\beq
\begin{aligned}\label{eq:ovl12NGm-G}
	&N(\Gmgn(\cpmgn)-\G(\cp)) \\
	&=   -\log(\frac{\so_\mgn}{\so}) + \frac{\so_\mgn}{\so} - 1 +\frac{\efres^2 \beta n_1^2 (\so_\mgn - \so)^2}{\so_\mgn \so^2} 
	+ (2\xi \efres  + \xi^2 )\beta \left(\frac{n_1^2}{\so_\mgn} + 1 \right) + \bhp{N^{-1/3}} . 
\end{aligned}
\eeq
Using the equations \eqref{eq:shN1/2} and \eqref{eq:smgnhN1/2} satisfied by $\so$ and $\so_\mgn$, the equation \eqref{eq:ovl12NGm-G} can be written as 
\beq
\begin{aligned}
	N(\Gmgn(\cpmgn)-\G(\cp))
	= & - \log(\frac{\so_\mgn}{\so}) + 2(\beta - 1) (\so_\mgn - \so) + (2\efres \xi + \xi^2) \beta + \bhp{N^{-1/3}}.
\end{aligned}
\eeq

\subsubsection*{Integrals}

We now consider the integrals in the formula \eqref{eq:ovwefdn}. The ratio of the integrals in this regime turns out to give a non-trivial contribution. 
We first show that we cannot use a Taylor series approximation. 
Consider the numerator; the denominator is the same as the numerator with $\xi=0$. For $k\ge 2$, we use the formula for $\Gmgn(z)$ to get 
\beqq \begin{split}
	\frac{\Gmgn^{(k)}(\cpmgn)}{(-1)^k (k-1)! }
	=&  \frac{1}{N} \sum_{i = 1}^N \frac{N^{\frac23 k}}{(\egres_1 +\so_\mgn N^{-1/3}- \egres_i)^k} + \frac{k(H+\xi)^2\beta }{N^2} \sum_{i = 1}^N \frac{n_i^2 N^{\frac23 (k+1)}}{(\egres_1 +\so_\mgn N^{-1/3}- \egres_i)^{k+1}} \\
	=& N^{\frac23 k-1} \left( \frac{N^{\frac13k}}{\so_\mgn^k} + \bhp{1} \right) 
	+ k (H+\xi)^2\beta N^{\frac23k-\frac43} \left( \frac{N^{\frac13(k+1)}}{\so_\mgn^{k+1}} + \bhp{1} \right) 
	= \bhp{N^{k-1}}.
\end{split} \eeqq
Since $\Gmgn^{(2)}= \bhp{N}$, the main contribution to the integral comes from a neighborhood of radius $N^{-1}$ around the critical point. 
If we use the new variable $z=\cpmgn + u N^{-1}$ and the Taylor series 
\beqq
	N \left( \Gmgn(\cpmgn+uN^{-1}) - \Gmgn(\cpmgn) \right) 
	= \sum_{k=2}^\infty  \frac{N^{-k+1}}{k!} \Gmgn^{(k)}(\cpmgn) u^k, 
\eeqq
we find that all terms in the series are $\bhp{1}$ for finite $u$. Since all terms in the Taylor series contribute to the integral, this method will not work and we instead proceed as follows. 
Using $\Gmgn'(\cpmgn)=0$, we have 
\beqq \begin{split}
	& N( \Gmgn(\cpmgn+w) - \Gmgn(\cpmgn)) = N( \Gmgn(\cpmgn+w) -  \Gmgn(\cpmgn)- \Gmgn'(\cpmgn)w)  \\
	&\quad  
	= -  \sum_{i=1}^N \left[ \log \left( 1+ \frac{w}{\cpmgn-\eg_i} \right) -  \frac{w}{\cpmgn-\eg_i}  \right] + \left( \ef + \frac{\xi}{\sqrt{N}} \right)^2\beta \sum_{i=1}^N  \frac{n_i^2 w^2}{(\cpmgn+w-\eg_i)(\cpmgn-\eg_i)^2} . 
\end{split} \eeqq
Separating out $i=1$, using a Taylor approximation of the log function, and using \eqref{lem:specialsum}, 
\beqq\begin{split}
	&N \left( \Gmgn(\cpmgn+uN^{-1}) - \Gmgn(\cpmgn) \right)  
	=-\log\left(1+\frac{u}{\so_\mathfrak{M}}\right)+\frac{u}{\so_\mathfrak{M}}
	+\frac{(H+\xi)^2\beta n_1^2u^2}{(\so_\mathfrak{M}+u)\so_\mathfrak{M}^2} + \bhp{u^2N^{-2/3}}.
\end{split} \eeqq
We temporarily write the middle two terms with $x:= (H+\xi)^2\beta n_1^2$ and get
\beqq
	\frac{u}{\so_\mgn} + \frac{x u^2}{(\so_\mgn+u)\so_\mgn^2} 
	=u \left(\frac{1}{\so_\mgn}+\frac{x}{\so_\mgn^2} \right) - \frac{x}{\so_\mgn} + \frac{x}{\so_\mgn+x}.
\eeqq
Using \eqref{eq:smgnhN1/2} twice, the above formula can be written as
\beq\begin{split}
	&N \left( \Gmgn(\cpmgn+uN^{-1}) - \Gmgn(\cpmgn) \right)  
	\simeq  -\log\left(1+\frac{u}{\so_\mathfrak{M}}\right)+(\beta-1)(u-\so_\mathfrak{M})+1+\frac{(H+\xi)^2\beta n_1^2}{(\so_\mathfrak{M}+u)}. 
\end{split}\eeq
Thus, 
\beq\label{eq:ext1/2contourcalc}
\begin{aligned}
	\int_{\cpmgn-\ii\infty}^{\cpmgn+\ii\infty}   e^{\frac{N}{2}(\Gmgn(z) - \Gmgn(\cpmgn))} \dd z 
	\simeq & \frac1{N} \int_{-\ii \infty}^{\ii \infty} \sqrt{\frac{\so_\mgn}{\so_\mgn + u}} e^{\frac{(\beta - 1)(u - \so_\mgn)}{2} + \frac12 + \frac{(\efres + \xi)^2 \beta n_1^2}{2(\so_\mgn + u)}} \dd u  \\
	= &\frac{\so_\mgn ^{1/2}e^{-(\beta - 1)\so_\mgn + \frac12 }}{N} \int_{-\ii \infty}^{\ii\infty} \frac{e^{\frac{(\beta - 1)(\so_\mgn + u)}{2} + \frac{(\efres + \xi)^2 \beta n_1^2}{2(\so_\mgn + u)}}}{\sqrt{\so_\mgn + u}}  \dd u .
\end{aligned}
\eeq
The last integral is an integral formula of a modified Bessel function which can be evaluated explicitly (see e.g. \cite{Handbook}): 
\beq	\label{eq:ntif}
	\int_{0_++\ii\R} \frac{e^{aw+\frac{b}{w}}}{\sqrt{w}}  \dd w = 2 \pi \ii \left( \frac{b}{a} \right)^{1/4} I_{-\frac12}(2\sqrt{ab}) = \frac{2\ii\sqrt{\pi}}{\sqrt{a}} \cosh(2\sqrt{ab}).
\eeq
Hence, we obtain 
\beq	
\label{eq:int_const}
	\int_{\cpmgn-\ii\infty}^{\cpmgn+\ii\infty}  e^{\frac{N}2 (\Gmgn(z)- \Gmgn(\cpmgn))} \dd z 
	\simeq \frac{2\ii \sqrt{2 \pi \so_\mgn}  e^{- (\beta-1)\so_\mgn +\frac12 } }{N\sqrt{\beta-1}} 
	\cosh \left( (\efres+\xi) |n_1| \sqrt{\beta (\beta -1)} \right). 
\eeq
Note that the integral depends on $\xi$, unlike in the cases $h>0$ and $h\sim N^{-1/6}$. 
The denominator is the case when $\xi=0$. Thus, 
\beq
	\frac{\int e^{\frac{N}{2}(\Gmgn(z) - \Gmgn(\cpmgn))} \dd z}{\int e^{\frac{N}{2}(\G(z) - \G(\cp))} \dd z} 
	\simeq \sqrt{\frac{\so_\mgn}{\so}} e^{-(\beta - 1)(\so_\mgn - \so)} \frac{\cosh \left( (\efres+\xi) |n_1| \sqrt{\beta (\beta -1)} \right)}{\cosh \left(\efres |n_1| \sqrt{\beta (\beta -1)} \right) }.
\eeq

Combining all terms together, replacing $\beta \xi$ by $\xi$ and using $T=1/\beta$, we obtain the following.

\begin{result} \label{result:mgn12}
For $h=HN^{-1/2}$ and $0<T<1$, 
\beq	
	\langle e^{ \xi \sqrt{N} \mgn} \rangle 
	\simeq 
	e^{ \efres \xi + \frac{ T \xi^2}{2} }
	\frac{\cosh \left( (\efres+T\xi) |n_1| \frac{\sqrt{1 - T}}{T} \right)}{\cosh \left(\efres |n_1| \frac{\sqrt{1 - T}}{T} \right) }
\eeq
as $N\to\infty$ for asymptotically almost every disorder sample. 
\end{result}


The right-hand side is the product of two terms, implying that $\sqrt{N} \mgn$ is a sum two independent random variables. 
The exponential term on right-hand side is the moment generating function of a Gaussian distribution, while the ratio of the cosh functions is the moment generating function of a shifted Bernoulli distribution. 
Indeed, if $\P(X=1)=p$ and $\P(X=-1)=1-p$ with $p= \frac{e^{a}}{e^a+e^{-a}}$, then 
\beqq
	\E[e^{\xi X}]=p e^{\xi} + (1-p) e^{-\xi}= \frac{\cosh(a+\xi)}{\cosh(a)} .
\eeqq
Hence, 
we deduce the following result. 


\begin{result}\label{thm:ext1/2} For $h=HN^{-1/2}$ and $0<T<1$, 
\beq
\label{eq:mgnhN1/2}
	\mgn \simeqidsgibbs  \ef + \frac{ |n_1|\sqrt{1-\tmp} \gib B(\alpha) + \sqrt{\tmp} \gib N}{\sqrt{N}} 
\eeq
as $N\to\infty$ for asymptotically almost every disorder sample. Here, $\gib B(c)$ is a shifted Bernoulli thermal random variable with the probability mass function $P(\gib B(c) =1)=c$ and $P(\gib B(c) =-1)=1-c$ and $\alpha$ in \eqref{eq:mgnhN1/2} is given by
\beq
\label{eq:probY}
	\alpha:= \frac{ e^{\frac{H |n_1| \sqrt{1-T}}{T} } }{ e^{\frac{H |n_1| \sqrt{1-T}}{T} } + e^{- \frac{H |n_1| \sqrt{1-T}}{T} }}.
\eeq
The thermal random variable $\gib N$ has the standard Gaussian distribution and it is independent of $\gib B(\alpha)$.
\end{result}

\subsubsection{Matching with $h\sim N^{-1/6}$ and $h=0$}

As $\efres\to \infty$, the random variable $\gib B(\alpha) \to 1$.
The formal limit of \eqref{eq:mgnhN1/2} as $\efres\to \infty$ is 
\beq
	\mgn \simeqidsgibbs \ef + \frac{1}{\sqrt{N}} \left[|n_1| \sqrt{1 - T} + \sqrt{T} \gib N\right],
\eeq
which is the same as \eqref{eq:mgnhN1/6H0} from the $\ef=\efres N^{-1/6}$ regime. On the other hand, if we take $\efres\to 0$, then $\gib B(\alpha)\xrightarrow{D} \gib B(1/2)$.
Hence, the formal limit of \eqref{eq:mgnhN1/2} as $H\to 0$ is the same as the $h=0$ case \eqref{eq:h0folt}. 
Therefore, the result \eqref{eq:mgnhN1/2} matches with both the $h\sim N^{-1/6}$ and $h=0$ results.

\subsection{Susceptibility} \label{sec:suscep}

In this subsection, we discuss properties of the susceptibility, defined as the magnetization per external field strength. In the next subsection we discuss differential susceptibility 
\beq
	\scp= \frac{\mage}{\ef} = \frac{\langle \mgn \rangle}{\ef} = \frac1{h} \frac{\dd \fe_N}{\dd h}. 
\eeq
We denote by $\bar \scp$ or $\E_s [ \scp]$ the sample average of $\scp$. We denote by $\Var_s$ the sample variance. 
As described in Section \ref{sec:RMT}, we use the font $\simeqids$ to denote an asymptotic expansion in distribution with respect to the disorder sample.

\subsubsection{Macroscopic field $\ef=O(1)$}

From Subsection \ref{sec:mgnlowtmn12}, for fixed $h>0$ and $T>0$, 
\beq \label{eq:gascpehfae}
	\scp \simeq   \scpz+ \frac{\scpo}{\sqrt{N}} 
\eeq	 
for asymptotically almost every disorder sample, where 
\beq
	\scpz= s_1(\cp_0) \quad \text{and} \quad \scpo= \go_N(\cp_0;1)- \frac{\ef^2 s_2(\cp_0)  \go_N(\cp_0; 2)}{T s_2(\cp_0) + 2\ef^2  s_3(\cp_0)} 
\eeq	 
and $\cp_0$ is the solution of the equation \eqref{eq:cp0h>0} and $\go_N(z;k )$ is defined in  \eqref{eq:go}.

The leading term $\scpz$ is deterministic and satisfies:
\begin{itemize}
\item $\scpz$ is a decreasing function of $\ef$,
\item As $h\to \infty$, 
\beq
	\scpz(h, T)= \frac1{h}- \frac{T}{2h^2} + O(h^{-3}) \quad\text{for all $T>0$}
\eeq 
\item As $h\to 0$, 
\beq \label{eq:susehp}
	\scpz(h, T) \simeq 
	\begin{cases}
	\frac{1}{T} - \frac{\ef^2}{T(T^2-1)}   \qquad &\text{for $T>1$}\\ 
	1 - \frac{\ef^2}{2(1 - \tmp)} \qquad & \text{for $0<T<1$.} 
	\end{cases}
\eeq
\end{itemize}
See Figure \ref{fig:scpzvsh} for the graph of $\scpz$ as a function of $\ef$.

\medskip


The subleading term $\scpo$ depends on the disorder sample. We consider its sample-to-sample fluctuations. 
From \eqref{eq:go}, $\go_N(\cp_0;1)$ and $\go_N(\cp_0;2)$ converge to the centered bivariate Gaussian distribution with 
\beq
	\mathrm{Var}_s[\go_N(\cp_0; 1)] \to 2s_2(\cp_0) ,  \qquad   \mathrm{Var}_s[ \go_N (\cp_0; 2)] \to 2s_4(\cp_0),
\eeq
and 
\beq
	\mathrm{Cov}_s(\go_N(\cp_0; 1), \go_N(\cp_0; 2)) = \E_s \left[\frac{1}{N} \sum_{i = 1}^N \frac{(n_i^2 - 1)^2}{(\cp_0 - \widehat \eg_i)^3}\right] \to 2s_3(\cp_0). 
\eeq
as $N \rightarrow \infty$.
Hence, as $N\to \infty$, 
\beq
	\scpo \simeqids \NN(0, \sva^2) 
\eeq
where the sample variance is
\beq
	\sva^2 =   \frac{ 2 s_2(\cp_0)^2 \left( T^2 s_2(\cp_0)+ 2T \ef^2 s_3(\cp_0)+ h^4 s_4(\cp_0)\right) }{\left( T s_2(\cp_0) + 2\ef^2  s_3(\cp_0) \right)^2}.
\eeq 
See Figure \ref{fig:scpsigvsh} for the graph of $\sva^2$. The graph shows that $\sva^2$ is a monotonically decreasing function of $\ef$. 
It is easy to check that:
\begin{itemize}
\item As $\ef\to \infty$,  
\beq
	\sigma_s^2\simeq \frac{1}{2\ef^2} \qquad \text{for all $T>0$.}
\eeq
\item As $h\to 0$, 
\beq
	\sigma_s^2 \simeq \begin{dcases}
	\frac2{T^2-1} \qquad &\text{for $T>1$,}\\
	\frac{1-T}{\ef^2} \qquad &\text{for $T<1$.} 
	\end{dcases}
\eeq
\end{itemize} 

The above formula suggests that there is an interesting transition as $T$ approaches the critical temperature $T=1$ in the case where $h\to 0$. 
The behavior near the $(T,h)=(1,0)$ is worth studying, but we leave this subject for the future. 

\begin{figure}[H]
\centering
\begin{subfigure}{0.45\textwidth}
\includegraphics[width = 0.9\textwidth]{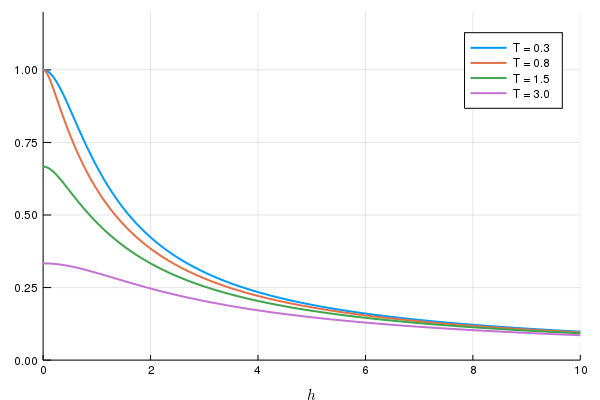}
\caption{Graph of $\scpz(\ef)$}
\label{fig:scpzvsh}
\end{subfigure}
\begin{subfigure}{0.45\textwidth}
\includegraphics[width = 0.9\textwidth]{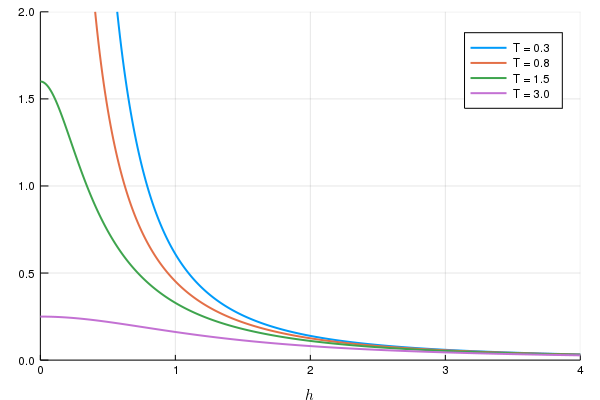}
\caption{Graph of $\sva^2(\ef)$}
\label{fig:scpsigvsh}
\end{subfigure}
\caption{Graph of $\scpz(\ef)$ and $\sva^2(\ef)$ as function of $\ef$ for various values of  $\tmp$}
\end{figure}

\subsubsection{Mesoscopic external field: $\ef\sim N^{-1/6}$ and $T<1$}

From Subsection \ref{sec:mgnlowtmn12}, for $h=HN^{-1/6}$ with fixed $H>0$ and $0<T<1$, 
\beq	
	\scp \simeq  1 + \frac{\crv(\stild)}{N^{1/3}} 
\eeq
for asymptotically almost every disorder sample, where $\crv(\stild)$ is given in \eqref{eq:mgfmgn}.

The behavior of $\crv(\stild)$ as $H\to \infty$ and $H\to 0$ is discussed in Subsubsections \ref{sec:16maef0} and \ref{sec:formlimh1/6mng}.
The sample-to-sample fluctuation of $\crv(\stild)$ is shown in Subsection \ref{sec:fetranltmp} and we see that
$\crv(\stild)  \simeqids  \crvlim(\slim)$ where  
\beq	
	\crvlim(\slim)= \lim_{n \rightarrow \infty} \left(\sum_{i = 1}^n \frac{\evg_i^2}{\varsigma +\airy_1 - \airy_i} - \frac1{\pi} \int_0^{\left(\frac{3\pi n}{2}\right)^{2/3}} \frac{\dd x}{ \sqrt{x}}\right)
\eeq
and $\slim > 0$ solves $1 - T = \efres^2 \sum_{i = 1}^\infty \frac{\evg_i^2}{(\varsigma + \airy_1 - \airy_i)^2}$. 
Here, $\airy_i$ is the GOE Airy point process and $\evg_i$ are i.i.d standard normal random variables independent of $\airy_i$.

\subsubsection{Microscopic external field: $h\sim N^{-1/2}$ and $T<1$} 

The thermal average of \eqref{eq:mgnhN1/2} implies that for $h=HN^{-1/2}$ with fixed $H>0$ and $0<T<1$, 
\beq \label{eq:scpcmicro}
	\scp \simeq  1 + \frac{|n_1| \sqrt{1 - \tmp}}{H} \tanh \left(\frac{H |n_1|\sqrt{1 - T}}{T}\right) =: \chimicro
\eeq
for asymptotically almost every disorder sample. 
The function $\chimicro$ is a decreasing function in both $H$ and $T$ (see Figures \ref{fig:scpHTn1} and \ref{fig:scpHTn122}). 
From the formula for $\chimicro$, we conclude that
\beq 
	\chimicro \simeq 	1+ \frac{|n_1| \sqrt{1 - \tmp}}{H} \qquad \text{as $H\to \infty$} 
\eeq
and
\beq \label{eq:sushohalf}
	\chimicro \simeq 	
	 1 +  \frac{n_1^2  (1 - \tmp)}{\tmp}  - \frac{H^2 n_1^4(1-T)^2}{3 T^3}  \qquad \text{as $H\to 0$.}
\eeq

\begin{figure}
\centering
\begin{subfigure}{0.45\textwidth}
\includegraphics[width = 0.9\textwidth]{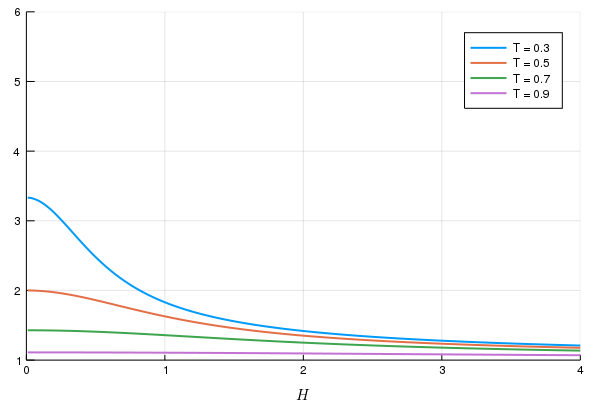}
\caption{$\chimicro$ as a function of $H$  when $|n_1| = 1$}
\end{subfigure}
\begin{subfigure}{0.45\textwidth}
\includegraphics[width = 0.9\textwidth]{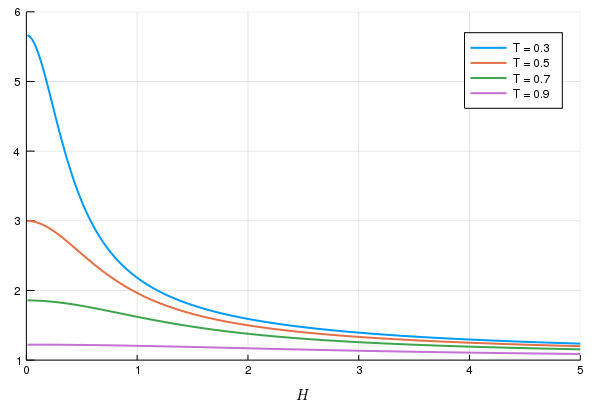}
\caption{$\chimicro$ as a function of $H$ when $|n_1| = \sqrt{2}$}
\end{subfigure}
\caption{Graph of $\chimicro(\efres, T)$ as a functions of $\efres$ for various values of $0<\tmp<1$.}
\label{fig:scpHTn1}
\end{figure}

\begin{figure}
\centering
\begin{subfigure}{0.45\textwidth}
\includegraphics[width = 0.9\textwidth]{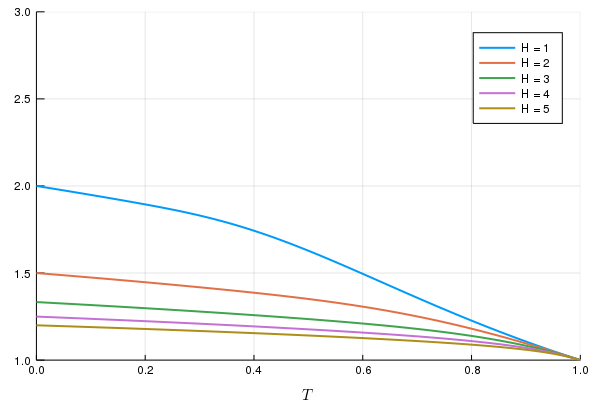}
\caption{$\chimicro$ as a function of $T$ when $|n_1| = 1$}
\end{subfigure}
\begin{subfigure}{0.45\textwidth}
\includegraphics[width = 0.9\textwidth]{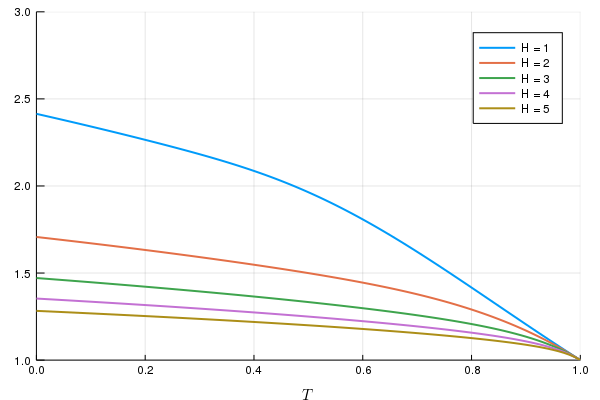}
\caption{$\chimicro$ as a function of $T$ when $|n_1| = \sqrt{2}$}
\end{subfigure}
\caption{Graph of $\chimicro(\efres, T)$ as a functions of $\tmp$ for various values of $\efres$.}
\label{fig:scpHTn122}
\end{figure}

\subsubsection{The zero external field limit of the susceptibility}

We consider two different limits of the susceptibility depending on how $\ef\to 0$ and $N\to \infty$ are taken.  
The first limit is obtained from \eqref{eq:susehp}:
\beq \label{eq:zerosunaive}
	\lim_{\ef\to 0} \lim_{\substack{N\to \infty \\ \ef>0}} \scp = 
	\begin{cases}
	\frac{1}{T}     \qquad &\text{for $T>1$}\\ 
	1 \qquad & \text{for $T<1$.} 
	\end{cases}
\eeq
See Figure \ref{fig:susvarvsh} (a). 
This result \eqref{eq:zerosunaive} was previously obtained in \cite{kosterlitz1976spherical}, and also in \cite{cugliandolo2007nonlinear}. The limit does not depend on the disorder sample.

The second limit is obtained from \eqref{eq:sushohalf} for $0<T<1$: 
\beq \label{eq:zerosucsoph}
	\lim_{H\to 0} \lim_{\substack{N\to \infty \\ h=HN^{-1/2}}} \scp = 1 +  \frac{ n_1^2  (1 - \tmp)}{\tmp}
\qquad \text{for $0<T<1$.}	
\eeq
See Figure \ref{fig:susvarvsh} (b). 
This limit depends on the disorder sample, but only on one variable, $n_1^2$. 
Observe that this limit blows up at $T=0$ while the limit \eqref{eq:zerosunaive} is finite at $T=0$.  
The sample-to-sample average of \eqref{eq:zerosucsoph} satisfies
\beq \label{eq:zerosucsampleavsoph}
	\lim_{H\to 0} \lim_{\substack{N\to \infty \\ h=HN^{-1/2}}} \bar{ \scp}  = \frac1{T}
\qquad \text{for $0<T<1$.}	
\eeq

\subsection{Differential susceptibility}\label{sec:diffsusc}


We also consider the differential susceptibility given by 
\beq
	\chid = \frac{\dd}{\dd h} \langle \mgn \rangle = \frac{\dd^2 \fe_N}{\dd h^2}= \frac{N}{T} \left( \langle\mgn^2\rangle - \langle \mgn \rangle^2 \right) .
\eeq

The results \eqref{eq:mgnphc}, \eqref{eq:mgn_H1/6}, and \eqref{eq:mgnhN1/2} imply the following formulas. All formulas hold for asymptotically almost every disorder sample. 
\begin{enumerate}[(a)]
\item For fixed $h>0$ and $T>0$, 
\beq
	\chid \simeq  s_1(\cp_0)  - \frac{2 \ef^2 s_2(\cp_0)^2} {\tmp s_2(\cp_0) + 2\ef^2 s_3(\cp_0)  } =:\chidma.
\eeq
\item For $h=HN^{-1/6}$ with fixed $H>0$ and $0<T<1$, 
\beq
	\chid\simeq 1.
\eeq
\item For $h=HN^{-1/2}$ with fixed $H>0$ and $0<T<1$, 
\beq \label{eq:chidmicro}
	\chid \simeq 1+ \frac{n_1^2(1-T)}{T \cosh^2 \left( \frac{H |n_1| \sqrt{1-T}}{T} \right) }=:\chidmic.
\eeq
\end{enumerate}

The limits for the macroscopic and mescopic regimes do not depend on the disorder samples, but the limit for the microscopic regime depends on the disorder variable $n_1^2$.
The macroscopic limit satisfies the following property as $h\to 0$: 
\beq
	\chidma\simeq \begin{dcases} \frac1{T} - \frac{3h^2}{T(T^2-1)} + O(h^4)\qquad &T>1, \\
	1- \frac{3h^2}{2(1-T)^2} + O(h^4) \qquad &0<T<1. 
	\end{dcases}
\eeq
On the other hand the microscopic limit satisfies, for $0<T<1$, 
\beq
	\chidmic \simeq \begin{dcases} 1+ O(e^{-\frac{2H |n_1| \sqrt{1-T}}{T} }) \qquad &\text{as $H\to \infty$.} \\
1+ \frac{n_1^2(1-T)}{T} - \frac{H^2n_1^4(1-T)^2}{T^3} \qquad &\text{as $H\to 0$.}
\end{dcases} \eeq

The zero external field limit is the same as the susceptibility of the last section even though the subleading terms differ by a factor of $3$. 
In both cases the limit is
\beq
	\lim_{H\to 0} \chidmic= \lim_{H\to 0} \chimicro= 1+ \frac{n_1^2(1-T)}{T} 
\eeq
and this value depends on the disorder variable $n_1^2$. Note that the sample-to-sample average of $n_1^2$ is $1$. This result shows that both susceptibilities satisfy Curie's law in the sample-to-sample average sense, but not in the quenched disorder sense. 

We note that if we take $T\to 0$ with $H>0$ fixed in \eqref{eq:scpcmicro} and \eqref{eq:chidmicro}, then
\beq \label{eq:sucatt0}
	 \chimicro \simeq 1+ \frac{|n_1|}{H} \quad \text{and} \quad \chidmic \simeq 1 \qquad \text{at $T=0$.}
\eeq
This shows that $\chidmic(T=0)$ does not diverge as $H\to 0$ but $\chimicro(T=0)$ does.

\section{Overlap with the ground state}\label{sec:ground}

Recall that $\pm \vu_1$ denote the unit eigenvectors corresponding to the largest eigenvalue of $M$. The overlap of the spin with the ground state and the squared overlap are defined as
\beq
	\OG= \frac{|\vu_1\cdot \sphv|}{\sqrt{N}}, \qquad 
	\OM =\OG^2 = \frac1{N} (\vu_1\cdot \sphv)^2,
\eeq
respectively. 
 The overlap $\OG=1$ when $T=h=0$ since the Hamiltonian is maximized when $\sphv$ is parallel to $\pm \vu_1$. 
The overlap measures how close the spin is to the ground state. 
Since it is more convenient to analyze, we consider $\OM$ in this section.

As with the overlap with the external field, there are no transitions when $T>1$ as $\ef\to 0$. 
However, when $T<1$, there are two interesting transitional regimes given by $\ef\sim N^{-1/6}$ and $\ef\sim N^{-1/3}$. The second regime did not appear for the overlap with the external field. On the other hand, the regime $\ef\sim N^{-1/2}$, which we studied for the free energy and the overlap with the external field, does not reveal any new features of $\OM$.  Instead, $\OM$ has the same properties for $h\sim N^{-1/2}$ as it does for $\ef=0$.


The moment generating function of $\OM$ has the integral formula given in Lemma \ref{lem:contour}, 
\beq\label{eq:mgf_omcopy}
	\langle e^{\beta \eta \OM} \rangle = e^{\frac N2(\Gom(\cpm) - \G(\cp))} \frac{\int^{\cpm + \ii\infty}_{\cpm - \ii\infty} e^{\frac N2 (\Gom (z) - \Gom(\cpm))} \dd z}{\int_{\cp - \ii \infty}^{\cp + \ii \infty} e^{\frac N2 (\G(z) - \G(\cp))} \dd z} 
\eeq
where
\beq
\begin{aligned}
	\Gom(z) & =  \beta z - \frac1N \log\left(z - \eg_1  - \bpm \right)- \frac{1}{N}\sum_{i = 2}^N \log(z - \eg_i) + \frac{\ef^2\beta n_1^2}{N( z - \eg_1 -\bpm) } + \frac{\ef^2\beta}{N} \sum_{i = 2}^N \frac{n_i^2}{z - \eg_i}.
\end{aligned}
\eeq
We take take $\cpm$ and $\cp$ to be the critical points of $\Gom$ and $\G$ respectively, and we use the notation
\beq \label{eq:bnotation}
	\bpm := \frac{2\eta}{N}. 
\eeq
The difference between $\Gom$ and $\G$ is that, in the case of $\Gom$, $\eg_1$ is changed to $\eg_1+ \bpm$. 

\medskip

The following two formulas will be used in the analysis below. 
First, we have 
\beq \label{eq:diffGomGher} 
	N( \Gom(\cpm) - \G(\cp) )  = N(\Gom(\cpm) - \G(\cp) - \G'(\cp)(\cpm - \cp))= D_1 + D_2 + D_3 + D_4
\eeq
where
\beq \label{eq:def_D}
\begin{aligned}
&D_1 = -\log \left(1 + \frac{\cpm - \cp - \bpm}{\cp - \eg_1}\right) + \frac{\cpm - \cp}{\cp - \eg_1}, \quad D_2 = -\sum_{i = 2}^N \left[\log \left(1 +  \frac{\cpm - \cp}{\cp - \eg_i}\right) - \frac{\cpm - \cp}{\cp - \eg_i}\right] , \\
	&D_3 = \ef^2 \beta n_1^2\left[\frac{1}{\cpm - \eg_1 - \bpm} - \frac{1}{\cp - \eg_1} + \frac{\cpm - \cp}{(\cp - \eg_1)^2}\right] , \quad D_4 = \ef^2\beta (\cpm - \cp)^2\sum_{i = 2}^N \frac{n_i^2 }{(\cpm - \eg_i)(\cp - \eg_i)^2}.
\end{aligned}
\eeq
Second, we can show from the equation $\Gom'(\cpm)-\G'(\cp)=0$ that 
\beq \label{eq:cpmcpdif} \begin{split}
	&(\cpm-\cp) \bigg[  \frac1{N(\cpm - \eg_1-\bpm ) (\cp - \eg_1 ) } + \frac{1}{N} \sum_{i=2}^N \frac1{(\cpm - \eg_i ) (\cp-\eg_i)}
	+ \frac{h^2\beta n_1^2}{N} \frac{\cp+\cpm-2\eg_1-\bpm}{(\cpm - \eg_1-\bpm )^2 (\cp - \eg_1 )^2 } \\
	&+  \frac{h^2\beta }{N}\sum_{i=2}^N  \frac{n_i^2( \cp+\cpm-2\eg_i)}{(\cpm - \eg_i)^2 (\cp - \eg_i )^2 }
	\bigg] = b\bigg[ \frac1{N(\cpm - \eg_1-\bpm ) (\cp - \eg_1 ) }+ \frac{h^2\beta n_1^2}{N} \frac{\cp+\cpm-2\eg_1-\bpm}{(\cpm - \eg_1-\bpm )^2 (\cp - \eg_1 )^2}
	\bigg].
\end{split}\eeq

\subsection{Macroscopic external field: $\ef=O(1)$} \label{sec:om_h>0}

\subsubsection{Analysis}

Fix $h>0$. 
The fluctuations of $\OM$ turn out to be of order $N^{-1}$.
Thus we set 
\beq
	\eta = \xi N \quad\text{so that} \quad 	\bpm= 2\xi.
\eeq
The critical point of $\G(z)$ is obtained in Subsection \ref{sec:freeenergypof} and is given by $\cp = \cp_0 + \bhp{N^{-1/2}}$ where $\cp_0$ solves the equation \eqref{eq:cp0h>0}.
We do not need an explicit formula for the term $\bhp{N^{-1/2}}$ in this section.  
Since $\Gom'(z) = \G'(z) + \bhp{N^{-1}}$ for $z>2$, a perturbation argument implies that the critical point of $\Gom(z)$ is given by 
\beq
	\cpm = \cp + \bhp{N^{-1}}.
\eeq

We use \eqref{eq:diffGomGher} to compute $N( \Gom(\cpm) - \G(\cp) )$. 
From the semi-circle law,  we have $D_2=\bhp{(\cpm-\cp)^2 N}= \bhp{N^{-2}}$ and $D_4=\bhp{N^{-1}}$. 
On the other hand, $D_1$ and $D_3$ are easy to compute and we find that 
\beq
	N(\Gom(\cpm) - \G(\cp)) = - \log\left(1 - \frac{2\xi}{\cp_0 - 2}\right) + \frac{2\ef^2\beta n_1^2 \xi}{(\cp_0 - 2)^2 (1 - \frac{2\xi}{\cp_0 - 2})} + \bhp{N^{-1/2}}.
\eeq

Since $\Gom^{(k)}(\cpm) = \bhp{1}$ for all $k\ge 2$, the ratio of the integrals \eqref{eq:mgf_omcopy} can be evaluated using the method of steepest descent.  
For $k=2$, 
\beqq
	\Gom''(\cpm) \simeq s_2(\cp_0) +  \ef^2\beta s_3(\cp_0) ,
\eeqq
which does not depend on $\xi$. 
Since $\G(\cp)$ is the special case of $\Gom(\cp)$ when $\xi=0$, we conclude that 
\beqq
	 \frac{\int^{\cpm + \ii\infty}_{\cpm - \ii\infty} e^{\frac N2 (\Gom (z) - \Gom(\cpm))} \dd z}{\int_{\cp - \ii \infty}^{\cp + \ii \infty} e^{\frac N2 (\G(z) - \G(\cp))} \dd z} 
	 \simeq \sqrt{\frac{\G''(\cp)}{\Gom''(\cpm)} } \simeq 1. 
\eeqq

Inserting these results into \eqref{eq:mgf_omcopy}, replacing $\xi$ with $(\cp_0 - 2)\xi$, and using $\beta = 1/\tmp$, we obtain the following. 

\begin{result}
For $h=O(1)$ and $T>0$, 
\beq \label{eq:hhhh}
	\langle e^{\frac{\cp_0 - 2}{\tmp}\xi N \OM}\rangle \simeq \left(1 - 2\xi\right)^{-1/2} e^{\frac{\ef^2 n_1^2 \xi}{\tmp(\cp_0 - 2) (1 - 2\xi)}}
\eeq
as $N\to \infty$ for asymptotically almost every disorder, where
$\cp_0>2$ is the solution of the equation \eqref{eq:cp0h>0}.
\end{result}


Note that if $X$ is a non-central Gaussian random variable  $\mu+ \gib N$, i.e. if $X^2$ is a non-centered chi-squared distribution with $1$ degree of freedom, then  
\beq \label{eq:noncechi}
	E[ e^{\xi X^2}]  = (1 - 2\xi)^{-1/2} e^{\frac{\mu^2 \xi}{1 - 2\xi }}.
\eeq
Therefore, we obtain the next result from the one above. 

\begin{result}\label{thm:groundh>0} For $h=O(1)$ and $T>0$, 
\beq
\label{eq:om_h>0}
	\OM\simeqidsgibbs \frac{\OMz}{N} \quad \text{where} \quad \OMz=  \frac{\tmp}{\cp_0-2}  \left \lvert \frac{\ef |n_1|}{\sqrt{\tmp(\cp_0 - 2)}}+  \gib N \right \rvert^2 
\eeq
as $N\to \infty$ for asymptotically almost every disorder,  where the thermal random variable $\gib N$ has the standard Gaussian distribution.
\end{result}


\subsubsection{Limits as $\ef\to \infty$ and $\ef\to 0$} 

Consider the formal limit of \eqref{eq:om_h>0} as $h\to \infty$. From \eqref{eq:cpvalueforhplh}, we find that  if we take $h>0$ and let $N\to \infty$ first and then $h\to \infty$, we get 
\beq
	\OM \simeqids \frac{1}{N} \left[ n_1^2 + \frac{2|n_1|\sqrt{T} }{\sqrt{\ef}} \gib N \right] 
\eeq
for all $T>0$. 
On the other hand, the equation \eqref{eq:cpohs} implies that if we take $h>0$ and let $N\to \infty$ first and then $h\to 0$, we obtain  
\beq \label{eq:915ama}
	\OM \simeqidsgibbs \frac{T^2}{N(T-1)^2} \left[ \gib N^2 + \frac{2\ef |n_1| }{T-1} \gib N \right]  \qquad \text{for $T>1$}
\eeq
and  
\beq \label{eq:sqrom_h1/6_h0} \begin{split}
	\OM \simeqidsgibbs \frac{16}{N} \left[ \frac{(1-T)^4 n_1^2}{h^6}  + \frac{\sqrt{T}(1-T)^3|n_1| }{h^5} \gib N  \right]   \qquad \text{for $0<T<1$. }
\end{split} \eeq
For $0<T<1$, the above result indicates that the overlap is of order $1$ when $h\sim N^{-1/6}$. 
We study this regime in the next subsection. 


\subsection{Mesoscopic external field: $\ef \sim N^{-1/6}$ and $T<1$}
\label{sec:gs1/6}

\subsubsection{Analysis}

We set 
\beq
	\ef = H N^{-1/6}
\eeq
for fixed $H>0$. 
If we insert $h=HN^{-1/6}$ in to the formula, the equation \eqref{eq:sqrom_h1/6_h0}  indicates that the fluctuations are of order $N^{-1/6}$.
Thus, we set 
\beq
	\eta = \xi N^{1/6} \quad \text{so that} \quad \bpm =  2\xi N^{-5/6} 
\eeq
in \eqref{eq:mgf_omcopy} and \eqref{eq:bnotation}.

The critical point $\cp$ of $\G(z)$ is obtained in Subsection \ref{sec:fetranltmp} and it is given by $\cp = \eg_1 + sN^{-2/3}$ where $s>0$ is the solution of the equation \eqref{eq:bsao13}.
We now consider the critical point of $\Gom(z)$. 
From the formula, we see that $\Gom'(z)$ is an increasing function of $z$ for $z>\eg_1+b$. Using $b>0$ and the explicit formula of the functions, we can easily check that 
$\Gom'(\gamma)< \G'(\cp)=0$ and $\Gom'(\gamma+b)>\G'(\cp)=0$.
Hence, we find that $\cp<\cpm<\cp+b$, and thus, $\cpm-\gamma=\cO(N^{-5/6})$. 
We now set 
\beq 
	\cpm = \cp + \Delta N^{-5/6}
\eeq
and determine $\Delta$ using \eqref{eq:cpmcpdif}. 
The right-hand side of the equation \eqref{eq:cpmcpdif} is equal to 
\beqq
	\frac{2\xi}{N^{5/6}} \left[ \frac{N^{1/3}}{s^2}+ \frac{2H^2\beta n_1^2 N^{2/3}}{s^3} \right] = \frac{4\xi H^2\beta n_1^2}{N^{1/6} s^3} \left( 1+ \bhp{N^{-1/3}}\right). 
\eeqq
For the left-hand side of the equation, the first two terms are of smaller order than the last two terms. Using $\cpm = \cp + \bhp{N^{-5/6}}$ and $b=\bhp{N^{-5/6}}$ for the other two sums, 
the left-hand side is equal to 
\beqq
	 \frac{\Delta}{N^{5/6}} \left[ \frac{2H^2\beta n_1^2 N^{2/3}}{s^3} + 2H^2\beta N^{2/3} \sum_{i=2}^N \frac{n_i^2}{(s+a_1-a_i)^3}  + \bhp{N^{1/3}} \right] .
\eeqq
Therefore, 
\beq \label{eq:deltaomh1/6}
	\Delta = \frac{2\xi n_1^2 s^{-3} }{\sum_{i = 1}^N n_i^2 (s + \egres_1 - \egres_i)^{-3}} + \bhp{N^{-1/6}}.
\eeq

\medskip

We now evaluate $N(\Gom(\cpm) - \G(\cp))$ using \eqref{eq:diffGomGher}.
It is easy to check that  $D_1= \bhp{N^{-1/6}}$ and $D_2=\bhp{N^{-1/3}}$.
Evaluating the first two leading terms, 
\beqq
	D_3 
	= \efres^2 \beta n_1^2 \left[ \frac{2 \xi}{s^2} N^{1/6} +  \frac{(\Delta -2\xi)^2}{s^3}+  \bhp{N^{-1/6}} \right] .
\eeqq
Finally,  
\beqq
	D_4 = \efres^2 \beta \Delta^2\sum_{i = 2}^N \frac{n_i^2}{(s + \egres_1 - \egres_i)^3} + \bhp{N^{-1/6}}.
\eeqq
Putting these together and also using the explicit formula of $\Delta$, we obtain 
\beq \label{eq:toqsposod} 
	N(\Gom(\cpm) - \G(\cp)) =   \frac{2\efres^2 \beta n_1^2 }{s^2} \xi N^{1/6}
	+  \frac{4 \efres ^2\beta n^2_1 \left[ \sum_{i = 2}^N n^2_i (s+\egres_1 - \egres_i)^{-3} \right] }{s^3 \left[ \sum_{i = 1}^N n_i^2(s + \egres_1 - \egres_i)^{-3} \right] }\xi^2 
	+ \bhp{N^{-1/6}}. 
\eeq

It remains to consider the integrals in \eqref{eq:mgf_omcopy}. 
The scale $h= H N^{-1/6}$ is the same as the one in Subsection \ref{sec:fetranltmp}. 
Since $\cpm=  \eg_1+s N^{-2/3}+\Delta N^{-5/6} = \eg_1+s N^{-2/3}+\bhp{ N^{-5/6} }$ and $b=O(N^{-5/6})$, 
the calculation from Subsection \ref{sec:fetranltmp} applies with only small changes. We find from the explicit formulas that 
$\Gom^{(k)}(\cpm) = \bhp{N^{\frac23k-\frac23}}$ for all $k\ge 2$ and 
\beqq 
	\Gom''(\cpm) =H^2\beta t^2\sum_{i=1}^N\frac{n_i^2}{(s+a_1-a_i)^3} + \bhp{N^{-1/6}}.
\eeqq
Thus, as in  Subsection \ref{sec:fetranltmp}, the main contribution to the integral comes from a neighborhood of radius $N^{-5/6}$ around the critical point, and the numerator can be evaluated using a Gaussian integral. 
Since the leading term of $\Gom''(\cpm)$ does not depend on $\xi$ and the denominator is the case of the numerator with $\xi=0$, we find that 
\beq
	 \frac{\int^{\cpm + \ii\infty}_{\cpm - \ii\infty} e^{\frac N2 (\Gom (z) - \Gom(\cpm))} \dd z}{\int_{\cp - \ii \infty}^{\cp + \ii \infty} e^{\frac N2 (\G(z) - \G(\cp))} \dd z}  \simeq 1. 
\eeq

From the above computations, we obtain an asymptotic formula for $\langle e^{\beta \xi N^{1/6} \OM} \rangle$. 
Moving a term of order $N^{1/6}$ to the left, changing $\beta \xi$ to $\xi$, replacing $\beta$ by $1/\tmp$,
and replacing $s$ by $\stild$, which solves the equation \eqref{eq:stildeq}, we arrive at the following result.

\begin{result}
For $\ef=\efres N^{-1/6}$ and $0<T<1$, 
\beq
	\langle e^{ \xi N^{1/6} (\OM -  \frac{\efres^2  n_1^2 }{\stild^2} )} \rangle \simeq e^{\frac{2 \efres ^2 \tmp n^2_1 \left[ \sum_{i = 2}^N n^2_i (\stild+\egres_1 - \egres_i)^{-3} \right] }{\stild^3 \left[ \sum_{i = 1}^N n_i^2(\stild + \egres_1 - \egres_i)^{-3} \right] }\xi^2}
\eeq
as $N\to \infty$ for asymptotically almost every disorder sample, 
where $\stild>0$ is the solution of the equation \eqref{eq:stildeq}. 
\end{result}

The right-hand side depends  on the disorder sample heavily, as the formula involves all of the $a_i$ and $n_i$.  
The above result implies the following.

\begin{result}\label{thm:ground1/6} 
For $h=HN^{-1/6}$ and $0<T<1$, 
\beq
\label{eq:om_h1/6}
	\OM \simeqidsgibbs \frac{\efres^2 n_1^2}{\stild^2} + \frac{\sigmaom \gib N}{N^{1/6}}
	= \left[ 1- T-  \efres^2 \sum_{i = 2}^N \frac{n_i^2}{(\stild + \egres_1 - \egres_i)^2} \right]
	+ \frac{\sigmaom \gib N}{N^{1/6}}
\eeq
as $N\to \infty$ for asymptotically almost every disorder sample, 
where the thermal random variable $\gib N$ has the standard normal distribution and $\sigmaom>0$ satisfies 
\beq
\label{eq:sigma_h1/6}
	\sigmaom^2 = \frac{4 \efres^2\tmp  n^2_1 \left[ \sum_{i = 2}^N n^2_i (\stild+a_1 - a_i)^{-3} \right] }{ \stild^3 \left[ \sum_{i = 1}^N n_i^2(\stild+\egres_1 - \egres_i)^{-3} \right] } .
\eeq
The equality of the leading terms in the two formulas of \eqref{eq:om_h1/6} follows from the equation \eqref{eq:stildeq} that $\stild$ satisfies. 
\end{result}

\subsubsection{Matching with $\ef=O(1)$}


We consider the $H\to \infty$ limit. 
From \eqref{eq:sqrs0whenehf}, we have $\stild \simeq \frac{H^4}{4(1-T)^2}$. Hence, the term\eqref{eq:om_h1/6} satisfies 
\beqq
	\sigmaom^2\simeq \frac{4 \efres^2\tmp  n^2_1}{\stild^3} \simeq \frac{4^4 Tn_1^2(1-T)^6}{H^{10}}.
\eeqq
Therefore, the first formula of \eqref{eq:om_h1/6} implies that if we take $h=HN^{-1/6}$ and let $N\to\infty$ first and then $H\to \infty$, we get 
\beq	\label{eq:omh1/6Hlarge}
	\OM \simeqidsgibbs   \frac{16}{N} \left[ \frac{(1-T)^4 n_1^2}{h^6}  + \frac{\sqrt{T}(1-T)^3|n_1| }{h^5} \gib N  \right] .
\eeq
This formula matches the formal limit given in \eqref{eq:sqrom_h1/6_h0}. Thus this regime matches with the $h=O(1)$ regime.

\subsubsection{Formal limit as $H\to 0$}

Using \eqref{eq:slimH0} for $\stild$, the denominator of \eqref{eq:sigma_h1/6} becomes $n_1^2+\bhp{H^3}$ as $H\to 0$. 
Thus, if we take $h=HN^{-1/6}$ and let $N\to\infty$ first and then take $H\to 0$, we get 
\beq \label{eq:heq16limHla}
	\OM \simeqidsgibbs 1-T  - \efres^2 \sum_{i=2}^N \frac{n_i^2}{(a_1-a_i)^2} + \frac{2\efres \sqrt{T}}{N^{1/6}} \left[   \sum_{i=2}^N \frac{n_i^2}{(a_1-a_i)^3} \right]^{1/2} \gib N . 
\eeq
The last two terms of \eqref{eq:heq16limHla} are of orders $H^2=h^2N^{1/3}$ and $HN^{-1/6}=h$, respectively. 
These two terms have the same order if $h\sim N^{-1/3}$. We study this regime in the next subsection. 
Note that, in this regime, the two terms are of order $N^{-1/3}$.


\subsection{Microscopic external field: $\ef \sim N^{-1/3}$ and $T<1$}
\label{sec:gs1/3}


\subsubsection{Analysis}\label{sec:gs1/3analysis}

Set
\beq
	\ef = H N^{-1/3}
\eeq
for fixed $H>0$. 
In the last part of the previous sub-subsection, a formal calculation indicated that the order of fluctuation in this regime is $N^{-1/3}$. 
We set 
\beq
	\eta = \xi N^{1/3} \quad \text{so that} \quad \bpm = 2\xi N^{-2/3}. 
\eeq

The regime $\ef\sim N^{-1/3}$ did not appear in previous sections. Hence, we first find the critical point $\cp$ of $\G(z)$.
Previously we saw that $\cp=\eg_1+ \bhp{N^{-2/3}}$ when $\ef\sim N^{-1/6}$ and  $\cp=\eg_1+ \bhp{N^{-1}}$ when $\ef\sim N^{-1/2}$. 
We expect that, in this regime, $\cp$ is between the above two cases, so we set $\cp=\eg_1+w$ for some $w$ and we assume $N^{-1}\ll w \ll N^{-2/3}$. 
The equation for the critical point is, using \eqref{eq:res_edge}, 
\beq \label{eq:ggmacpezq}
	\G'(\cp)= \beta- \frac1{N} \sum_{i=1}^N \frac1{\gamma-\eg_i} - \frac{H^2\beta}{N^{5/3}} \sum_{i=1}^N \frac{n_i^2}{(\gamma-\eg_i)^2}
	= \beta - \frac1{Nw} - 1 + \bhp{N^{-1/3}} - \frac{H^2\beta n_1^2}{N^{5/3}w^2}  =0. 
\eeq
Under the assumption for $w$, we see that $\frac1{NW} \ll \frac1{N^{5/3}w^2}$, and hence $w=\bhp{N^{-5/6}}$. 
Explicitly solving the equation $\beta  - 1 - \frac{H^2\beta n_1^2}{N^{5/3}w^2}   =0$, we find that 
\beq\label{eq:ground1/3sdef}
	\cp= \eg_1+ \sgr N^{-5/6} \quad \text{where} \quad \sgr = \sqrt{\frac{H^2\beta n_1^2}{\beta-1}}	 + \bhp{N^{-1/6}} . 
\eeq
For later use, we record that, upon inserting $\cp=\eg_1+ \sgr N^{-5/6}$ into the equation \eqref{eq:ggmacpezq}, $\sgr$ satisfies the following more detailed equation, using the notation $\Xi_N$ defined in  \eqref{eq:def_crvbp}:
\beq\label{eq:somhN1/3}
	\beta - \frac{1}{\sgr N^{1/6}} - 1 - \frac{\Xi_N}{N^{1/3}} 
	+ \bhp{N^{-1/2}} - \frac{H^2\beta n_1^2}{\sgr^2} 
	- \frac{H^2\beta }{N^{1/3}} \sum_{i=2}^N \frac{n_i^2}{(a_1-a_i)^2} =0.
\eeq

The critical point $\cpm$ of $\Gom(z)$ is easy to obtain since $\bpm= \frac{2\xi}{N^{2/3}}$ has the same order as the fluctuations of the eigenvalues $\eg_i$. 
The critical point equation is the same as in the case of $\G(z)$ except that $\eg_1$ is changed to $\eg_1+\bpm$. 
Thus we have 
\beq
	\cpm = \eg_1 + \bpm+ \sgr_\OM N^{-5/6} \quad \text{where} \quad \sgr_\OM= \sgr	 + \bhp{N^{-1/6}} . 
\eeq
For our computation, it turns out that we need an improved estimate for $\sgr_\OM -\sgr$. 
The equation $\Gom'(\cpm)=0$ is, in terms of $\sgr_\OM$, 
\beqq
	 \beta - \frac{1}{\sgr_\OM N^{1/6}}  - 1  
 - \frac{\efres^2 \beta n_1^2}{\sgr_{\OM}^2} + \bhp{N^{-1/3}} =0. 
\eeqq
This equation is the same as the equation \eqref{eq:somhN1/3} up to order $N^{-1/6}$. Therefore, we obtain an improved estimate
$\sgr _\OM = \sgr  + \bhp{N^{-1/3}}$. As a consequence, 
\beq	
	\cpm - \cp = \bpm+ \bhp{N^{-7/6}} = 2\xi N^{-2/3} + \bhp{N^{-7/6}}.
\eeq

\medskip

We now evaluate $N(\Gom(\cpm) - \G(\cp))$ using \eqref{eq:diffGomGher}. 
We have
\beqq
	D_1 = \frac{2\xi N^{1/6}}{\sgr }+ \bhp{N^{-1/3}} , 
	\qquad
	D_2 = -  \sum_{i= 2}^N \left[ \log\left(1 + \frac{2\xi}{\egres_1 - \egres_i} \right)  -  \frac{2\xi}{\egres_1-\egres_i} \right]  + \bhp{N^{-1/6}} , 
\eeqq
\beqq
	D_3  =  \frac{2\xi \efres^2 \beta n_1^2}{\sgr ^2}  N^{1/3} + \bhp{N^{-1/6}}, 
	\qquad 
	D_4  = 4\xi^2  \efres^2 \beta \sum_{i = 2}^N \frac{n_i^2}{(a_1 + 2\xi - \egres_i)(\egres_1 - \egres_i)^2} + \bhp{N^{-1/6}} . 
\eeqq
Note that $\sgr $ appears only in $D_1$ and $D_3$. 
Using the equation \eqref{eq:somhN1/3}, the sum $D_1+D_3$ can be expressed without using $\sgr $:
\beq
	D_1+D_3= 2\xi N^{1/3} \left[ \beta  - 1 - \frac{\Xi_N}{N^{1/3}}  -  \frac{\efres^2 \beta}{N^{1/3}} \sum_{i = 2}^N \frac{n_i^2}{(\egres_1 - \egres_i)^2}  \right] + \bhp{N^{-1/6}} . 
\eeq
On the other hand, using the notation $\Xi_N$ in  \eqref{eq:def_crvbp} again, we can write
\beq	\begin{split}
	D_2=  - \left[ \sum_{i= 2}^N \log\left(1 + \frac{2\xi}{\egres_1 - \egres_i} \right)  -  2\xi N^{1/3} \right] + 2\xi \Xi_N+ \bhp{N^{-1/6}}.
\end{split} \eeq
Adding $D_1, D_2, D_3$, and $D_4$, and combining two sums that are multiplied by $H^2\beta$, we find that 
\beq
\begin{aligned}
	N(\Gom(\cpm) - \G(\cp)) = 
	& \;2\xi (\beta  - 1) N^{1/3} + \left[ 2\xi N^{1/3} - \sum_{i = 2}^N \log \left(1 + \frac{2\xi}{\egres_1 - \egres_i}\right)\right]  \\
	& \quad - 2\xi \efres^2\beta   \sum_{i = 2}^N \frac{n_i^2}{(\egres_1 + 2\xi - \egres_i)(\egres_1 - \egres_i)} + \bhp{N^{-1/6}}
\end{aligned}
\eeq
We note that the term in brackets is $\bhp{1}$ due to \eqref{eq:def_crvbp}. 

\medskip

Finally, we consider the integrals in \eqref{eq:mgf_omcopy}, beginning with the numerator. 
Using $\cpm = \eg_1 + \bpm+ \sgr  N^{-5/6} + \bhp{N^{7/6}}$ and the explicit formula for $\Gom(z)$, we find that
\beqq
	\Gom^{(k)}(\cpm)= \bhp{N^{\frac56k - \frac56}}
\eeqq
for $k\ge 2$. Since $\Gom''(\cpm)=\bhp{N^{\frac56}}$, the main contribution to the integral comes from a neighborhood of radius $N^{-\frac{11}{12}}$ about the critical point. 
For $k=2$, we find explicitly that
\beqq
	\Gom''(\cpm)= \frac{2H^2 \beta n_1^2}{\sgr ^3} N^{-5/6} + \bhp{N^{-1}}. 
\eeqq
Hence, 
\beq\label{eq:mkxld}
	N(\Gom(\cpm+wN^{-\frac{11}{12}})-\Gom(\cpm)) 
	= \sum_{k=2}^\infty \frac{N^{1- \frac{11}{12}k} \Gom^{(k)}(\cpm) w^k }{k!} 
	= 	\frac{H^2 \beta n_1^2}{\sgr ^3} w^2 + \bhp{N^{-\frac{1}{12}}}
\eeq
for finite $w$, and the integral can be evaluated as a Gaussian integral. 
Since the leading term of \eqref{eq:mkxld} does not depend on $\xi$, we find that the ratio of the integrals in \eqref{eq:mgf_omcopy} is asymptotically equal to $1$. 

Combining the computations above, we obtain an asymptotic formula for $\left \langle e^{\beta \xi N^{1/3} \OM}\right \rangle$. 
Moving a term 
and using $\beta =1/T$, we arrive at the following result.

\begin{result}
For $\ef= \efres N^{-1/3}$ and $0<T<1$, 
\beq \label{eq:alalsppp}
	\left \langle e^{ \frac{\xi}{T} N^{1/3} \left( \OM - (1 - T) \right)}\right \rangle 
	\simeq e^{\xi N^{1/3}}
	\prod_{i=2}^N \frac{e^{  - \frac{\xi H^2 n_i^2}{T (\egres_1 + 2\xi - \egres_i)(\egres_1 - \egres_i)}}}{\sqrt{1 + \frac{2 \xi}{\egres_1 - \egres_i}}}
\eeq
as $N\to \infty$ for asymptotically almost every disorder sample. 
\end{result} 


We remark that the right-hand side is $\bhp{1}$ since
\beqq
	\xi N^{1/3} - \frac12 \sum_{i=2}^N \log \left( 1 + \frac{2 \xi}{\egres_1 - \egres_i } \right) = \bhp{1}.
\eeqq

The formula \eqref{eq:alalsppp} is a product of the moment generating functions of non-centered chi-squared distributions (see \eqref{eq:noncechi}). Hence, we obtain the following. 

\begin{result}\label{thm:ground1/3} For $h=HN^{-1/3}$ and $0< T<1$, 
\beq 
\label{eq:OMresulth13}
	\OM \simeqidsgibbs 1-T + \frac{T}{N^{1/3}} \gib W_N , \qquad \gib W_N= N^{1/3} - \sum_{i=2}^N\frac{ \big| \frac{H|n_i|}{\sqrt{T(a_1-a_i)}} + \gib n_i \big|^2} {a_1-a_i}  
\eeq
as $N\to\infty$  for asymptotically almost every disorder sample, where the thermal random variables $\gib n_i$ are independent standard normal random variables. 
\end{result}

Here, we emphasize that $n_i$ are sample random variables (given by the dot product of each eigenvector of $M$ with the external field) while $\gib n_i$ are thermal random variables. 
Note that $\gib W_N=\bhp{1}$ since $N^{1/3} - \sum_{i=2}^N\frac{  \gib n_i^2} {a_1-a_i} = \bhp{1}$.

\subsubsection{Matching with the mesocopic field, $\ef\sim N^{-1/6}$}

We take the formal limit $H\to \infty$ of \eqref{eq:OMresulth13} and compare with \eqref{eq:heq16limHla}. 
Then, using $N^{1/3}- \sum_{i=2}^N \frac1{a_1-a_i}=\bhp{1}$ from \eqref{eq:res_edge}, 
\beq \label{eq:13tempo}
	\gib W_N = -  \frac{H^2}{T} \sum_{i=2}^n \frac{n_i^2}{(a_1-a_i)^2} -\frac{2H}{\sqrt{T}} \sum_{i=2}^N \frac{|n_i| \gib n_i}{(a_1-a_i)^{3/2}}
	+ \bhp{1}.
\eeq
The second sum is a sum of independent (thermal) Gaussian random variables,and hence it has a Gaussian distribution. 
Therefore, if take $h=HN^{-1/3}$ and let $N\to \infty$ first and then take $H\to \infty$, we get
\beq \label{eq:OMlimat} \begin{split}
	& \OM \simeqidsgibbs 1-T - \frac{\efres^2 }{N^{1/3}}\sum_{i=2}^N \frac{n_i^2}{(a_1-a_i)^2} 
	+ \frac{2\efres \sqrt{T}}{N^{1/3}} \left[ \sum_{i=2}^N \frac{n_i^2}{(a_1-a_i)^3}\right]^{1/2} \gib N .
\end{split} \eeq
In order to compare this with the result \eqref{eq:heq16limHla}, 
we use the notation $h=H_\micro N^{-1/3}= H_\meso N^{-1/6}$. 
The equations \eqref{eq:OMlimat} and \eqref{eq:heq16limHla} are same once we set $H=H_\micro$ and $H=H_\meso$, respectively.



\subsection{No external field: $\ef=0$} \label{sec:groundstatehzerolow}


For $0<T<1$, the calculations of the previous subsection for $h=HN^{-1/3}$ go through; we obtain the result by setting  $H=0$ in \eqref{eq:OMresulth13}.  For $T>1$, the computations in Subsection \ref{sec:om_h>0} for $\ef=O(1)$ also apply to $\ef=0$; 
see \eqref{eq:915ama}.  

\begin{result}\label{result:grounds0}
For $h=0$, 
\beq \label{eq:mainresutforovolhs} 
	\OM \simeqidsgibbs \begin{dcases}
	\frac{\tmp^2 }{N (\tmp-1)^2 } \gib N^2 \quad &\text{for $\tmp>1$,} \\
		1 - \tmp + \frac{\tmp}{N^{1/3}} \left(  N^{1/3} - \sum_{i=2}^N \frac{  \gib n_i^2} {a_1-a_i} \right)\quad &\text{for $0<T<1$.} 
	\end{dcases}
\eeq
where the thermal random variable $\gib N$ has the standard normal distribution, and $\gib n_i$ are independent standard normal thermal random variables. 
\end{result}

\subsection{The thermal average} 

We use the notation
\beq
	\OP = \langle \OM \rangle= \langle \OG^2 \rangle
\eeq
to denote the thermal average of the squared overlap of a spin with the ground state. Previous subsections imply the following results. 

\begin{enumerate}[(i)]
\item For $\ef\ge 0$ and $T>1$, or for $\ef=O(1)$ with $h>0$ and $0<T<1$, 
\beq
	\OP \simeq \frac{\OP^0}{N}, \qquad \OP^0=\frac{\tmp}{\cp_0 - 2}\left[\frac{\ef^2 n_1^2}{T(\cp_0 - 2)} + 1\right]. 
\eeq
From the asymptotic formulas \eqref{eq:cpvalueforhplh} and \eqref{eq:cpohs} of $\cp_0$,  
\beq \label{eq:olwgssahl}
	\OP^0\simeq n_1^2 + \frac{T-(T-4)n_1^2}{\ef} \qquad \text{as $\ef\to \infty$ for all $T>0$} 
\eeq
and
\beq
	\OP^0 \simeq \begin{dcases}
	 \frac{T^2}{(T-1)^2} + \frac{\ef^2 (n_1^2-1)T^2 }{(T-1)^4}  \quad &\text{as $\ef\to 0$ for $T>1$}\\
	  \frac{16 n_1^2 (1-T)^4}{h^6} + \frac{4T(1-T)^2+32 n_1^2 (1-T)^4}{h^4}  \quad &\text{as $\ef\to 0$ for $0<T<1$.}
	  \end{dcases}
\eeq
See  Figure \ref{fig:om0_vs_hTn12} for graphs of $\OP^0$.

\item For $\ef= HN^{-1/6}$ with $0<T<1$, 
\beq	\label{eq:ophltmpN1/6}
	\OP \simeq \frac{\efres^2 n_1^2}{\stild^2} = 1 - \tmp - \efres^2 \sum_{i = 2}^N \frac{n_i^2}{(\stild + a_1 - a_i)^2}.
\eeq

\item For $\ef= HN^{-1/3}$ with $T<1$ (including the case when $H=0$), 
\beq
	\OP \simeq 1 - \tmp + \frac1{N^{1/3}} \left[ T \left( N^{1/3} - \sum_{i=2}^N \frac{1}{\egres_1-\egres_i} \right) - H^2 \sum_{i=2}^N \frac{n_i^2}{(\egres_1-\egres_i)^2} \right] . 
\eeq
\end{enumerate}

If we collect only the order $1$ terms, then as $N\to\infty$ with $T<1$, 
\beq
	\OP \to \begin{cases}
		0 & \text{for $h>0$}\\
	1 - \tmp - \efres^2 \sum_{i = 2}^N \frac{n_i^2}{(\stild + a_1 - a_i)^2} & \text{for $h=HN^{-1/6}$} \\
	1-T \qquad & \text{for $h=HN^{-1/3}$ (including $H=0$).} \\
	\end{cases}
\eeq
The sample-to-sample standard deviation of the thermal average of squared overlap satisfies for $0<T<1$, 
\beq
	\sqrt{ \overline{\OP^2} - (\overline{\OP})^2 }  = \begin{cases}
		O(N^{-1}) & \text{for $h=O(1)$}\\
	O(1) & \text{for $h\sim N^{-1/6}$} \\
	O(N^{-1/3}) \qquad & \text{for $h\sim N^{-1/3}$ (including $h=0$).} \\
	\end{cases}
\eeq 
The order is largest when $h\sim N^{-1/6}$.

\begin{figure}[H]
\centering
\includegraphics[width = 0.45\textwidth]{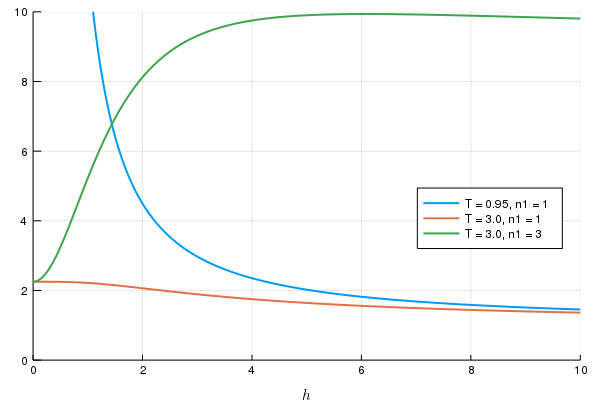}
\caption{Graphs of $\OP^0$ for $h=O(1)$ as function of $\ef$ for different combinations of $\tmp$ and $n_1$.}
\label{fig:om0_vs_hTn12}
\end{figure}

\subsection{Order of thermal fluctuations}

For $0<T< 1$, the standard deviation of the thermal fluctuations satisfies 
\beq
	\sqrt{ \langle \OM^2 \rangle - \langle \OM \rangle^2 }  = \begin{cases}
		\bhp{N^{-1}}& \text{for $h=O(1)$}\\
	\bhp{N^{-1/6}} & \text{for $h\sim N^{-1/6}$} \\
	\bhp{N^{-1/3}} \qquad & \text{for $h\sim N^{-1/3}$ (including $h=0$).}
	\end{cases}
\eeq 
for asymptotically almost every disorder sample. The thermal fluctuations are largest when $h\sim N^{-1/6}$.

\section{Overlap with a replica}\label{sec:2ol}

Let
\beq
	\ovl= \ovl^{1,2}= \frac{\spin^{(1)}\cdot \spin^{(2)}}{N} 
\eeq
be the overlap of a spin $\spin^{(1)}$ and its replica $\spin^{(2)}$, chosen independently from $S_{N-1}$ using the Gibbs measure with the same disorder sample. 
From Lemma \ref{lem:contour}, we have 
\beq \label{eq:ovlftsf}
	\langle e^{\eta \ovl} \rangle = e^{\frac{N}2 (\Govl(\cpovl, \cpovl;  a) - 2\G(\cp))} \frac{\iint e^{\frac{N}2 ( \Govl(z, w; a) - \Govl(\cpovl, \cpovl; a)) } \dd z \dd w}{\left( \int e^{\frac{N}2 (\G(z) - \G(\cp))} \dd z \right)^2 } 
\eeq
where
\beq
	\Govl(z, w;a) = \beta (z+w) - \frac1{N} \sum_{i=1}^N \log \left( (z-\eg_i)(w-\eg_i) - a^2 \right) 
	+ \frac{\ef^2 \beta}{N} \sum_{i=1}^N \frac{n_i^2 (z+w-2\eg_i + 2a) }{(z-\eg_i)(w-\eg_i) - a^2}
\eeq
and we set
\beq
	a= \frac{\eta}{\beta N}. 
\eeq
We take $\cp$ to be the critical point of $\G(z)$ and we chose $\cpovl>\lambda_1+|a|$ such that $(\cpovl,\cpovl)$ is a critical point of $\Govl(z,w;a)$.  We calculate $\cpovl$ below.

The partial derivative of $\Govl$ with respect to $z$ is 
\beq
	\frac{\partial \Govl}{\partial z} = \beta  - \frac1{N}  \sum_{i=1}^N \frac{w-\eg_i}{(z-\eg_i)(w-\eg_i) - a^2} 
	- \frac{\ef^2 \beta}{N} \sum_{i=1}^N \frac{n_i^2 (w-\eg_i+ a)^2 }{((z-\eg_i)(w-\eg_i) - a^2)^2} 
\eeq
and $\frac{\partial \Govl}{\partial w}$ is similar. Since $\frac{\partial \Govl}{\partial z}$ is an increasing function for real $z$ (and similarly with $\frac{\partial \Govl}{\partial w}$), there exists a critical point of the form $(z, w)=(\cpovl, \cpovl)$ where $\cpovl$ solves the equation 
\beq\label{eq:cpovldef}
	\beta  - \frac1{N}  \sum_{i=1}^N  \frac{\cpovl-\eg_i}{(\cpovl-\eg_i- a)(\cpovl-\eg_i+ a)} 
	- \frac{\ef^2 \beta}{N} \sum_{i=1}^N \frac{n_i^2  }{(\cpovl-\eg_i - a)^2} = 0  , \quad \cpovl>\eg_1+|a|.
\eeq
There may be other critical points, but $(\cpovl,\cpovl)$ is the one that we use for our steepest descent analysis.  For simplicity, we refer to this critical point as $\cpovl$ rather than $(\cpovl,\cpovl)$. 
For $a=0$, $\Govl(z,w;0)=\G(z)+\G(w)$, and in this case, the critical point is $(z,w)=(\cp, \cp)$.

\medskip

We use the following two formulas in this section.  
The first formula is 
\beq  \label{eq:Gdiffggaa} \begin{split}
	& N \left( \Govl(\cpovl, \cpovl;  a) - 2\G(\cp) \right)
	 = N (\Govl(\cpovl, \cpovl;  a) - 2\G(\cp) - 2 \G'(\cp) (\cpovl-\cp)) = B_1+B_2 
\end{split} \eeq
where
\beqq \begin{split}
	B_1= - \sum_{i=1}^N \left[ \log \left(  1+ \frac{2(\cpovl-\cp)}{\cp-\eg_i} + \frac{(\cpovl-\cp)^2 - a^2}{(\cp-\eg_i)^2} \right)  -\frac{2(\cpovl-\cp)}{\cp-\eg_i} \right] 
\end{split} \eeqq
and
\beqq
	B_2=  2\ef^2 \beta \sum_{i=1}^N n_i^2  \left[ \frac{1}{\cpovl-\eg_i - a}
	- \frac1{\cp-\eg_i} + \frac{\cpovl-\cp}{(\cp-\eg_i)^2} \right].
\eeqq
The second formula is 
\beq \label{eq:diffggga} \begin{split}
	& (\cpovl-\cp- a) \left[ \sum_{i=1}^N \frac{\cpovl-\eg_i}{(\cpovl-\eg_i- a)(\cpovl-\eg_i +a)(\cp-\eg_i)} 
	+ \ef^2 \beta \sum_{i=1}^N \frac{n_i^2 (\cp+\cpovl-2\eg_i- a) }{(\cpovl-\eg_i -  a)^2(\cp-\eg_i)^2}  \right] \\
	&\qquad =  - a \sum_{i=1}^N \frac{1}{(\cpovl-\eg_i + a)(\cp-\eg_i)}  ,
\end{split} \eeq
which follows from subtracting the critical point equations for $\cpovl$ and $\cp$.

\medskip

We also make use of the following lemma.

\begin{lemma}\label{lem:2olboundingcpovl} The point $\cpovl$ satisfies $\gamma<\cpovl<\gamma+a$.
\end{lemma}

\begin{proof} 
Let  
\beqq
	g(z)= \beta-\frac 1N\sum_{i=1}^N\frac{z-\lambda_i}{(z-\lambda_i-a)(z-\lambda_i+a)}-\frac{h^2\beta}{N}\sum_{i=1}^N\frac{n_i^2}{(z-\lambda_i-a)^2}.
\eeqq
Since $g(\cpovl)=0$, it is enough to show that $g(\cp)<0$ and $g(\cp+a)>0$.
Using $a>0$, we see that 
\beqq
	g(\cp)< \beta-\frac 1N\sum_{i=1}^N\frac{1}{\gamma-\lambda_i}-\frac{h^2\beta}{N}\sum_{i=1}^N\frac{n_i^2}{(\gamma-\lambda_i)^2}=\G'(\cp)=0. 
\eeqq 
On the other hand, 
\beqq\begin{split}
	g(\cp+a)= & \beta-\frac 1N\sum_{i=1}^N\frac{\gamma-\lambda_i+a}{(\gamma-\lambda_i)(\gamma-\lambda_i+2a)}-\frac{h^2\beta}{N}\sum_{i=1}^N\frac{n_i^2}{(\gamma-\lambda_i)^2}\\
>&\beta-\frac 1N\sum_{i=1}^N\frac{1}{\gamma-\lambda_i}-\frac{h^2\beta}{N}\sum_{i=1}^N\frac{n_i^2}{(\gamma-\lambda_i)^2}= \G'(\cp)= 0. 
\end{split}\eeqq
\end{proof}

\subsection{Macroscopic external field: $h=O(1)$} \label{sec:ovlpos}

\subsubsection{Analysis}

Fix $h>0$. 
It turns out that the fluctuations are of order $N^{-1/2}$. Hence, we set 
\beq
	\eta= \beta \xi \sqrt{N} \quad \text{so that} \quad a= \xi N^{-1/2}. 
\eeq

The critical point of $\G(z)$ is given in \eqref{eq:cp_h>0} by $\cp= \cp_0+ \cp_1 N^{-1/2}  + \bhp{N^{-1}}$. 
Consider the critical point $\cpovl$. By Lemma \ref{lem:2olboundingcpovl}, $\cpovl= \cp+ O(N^{-1/2})$. 
We now use the equation \eqref{eq:diffggga}. Using the semi-circle law approximation, we find that
\beq
	\cpovl-\cp-a = - \frac{a\left( s_2(\cp_0)+O(N^{-\frac12})\right)}{s_2(\cp_0)+ 2\ef^2 \beta s_3(\cp_0)+O(N^{-\frac12})} .
\eeq
Thus, 
\beq \label{eq:oclVaAf}
	\cpovl = \cp + \frac{\xi A}{\sqrt{N}} + O(N^{-1})
	\quad \text{where} \quad
	A= \frac{ 2\ef^2 \beta s_3(\cp_0)}{s_2(\cp_0)+ 2\ef^2 \beta s_3(\cp_0)}   .
\eeq

We evaluate $N (\Govl(\cpovl, \cpovl;  a) - 2\G(\cp))$ using \eqref{eq:Gdiffggaa}. From a Taylor approximation, 
\beq \begin{split}
	B_1 &= \sum_{i=1}^N \frac{(\cpovl-\cp)^2 + a^2}{(\cp-\eg_i)^2} + \bhp{N^{-1/2}} 
	= \left(  (\cpovl-\cp)^2 + a^2\right) N s_2(\cp)+\bhp{N^{-1/2}}. 
\end{split} \eeq
On the other hand, using the geometric series for $\frac{1}{\cpovl-\eg_i - a}= \frac{1}{(\cp-\eg_i)+(\cpovl -\cp - a)}$ and using \eqref{eq:weightls}, 
\beq \begin{split}
	B_2&= \sum_{i=1}^N n_i^2  \left[ \frac{a}{(\cp-\eg_i)^2} + \frac{(\cpovl-\cp-a)^2}{(\cp-\eg_i)^3} + O\left( \frac{(\cpovl-\cp-a)^3}{(\cp-\eg_i)^4} \right) \right] \\ 
	&=  a \left(  s_2(\cp)+ N^{-1/2}\go_N(\cp;2) \right) + (\cpovl-\cp-a)^2 s_3(\cp) + \bhp{N^{-1/2}}
\end{split} \eeq
where $\go_N(z;k)$ is defined in \eqref{eq:go00}. 
The leading term is $as_2(\cp)$ which is $O(N^{1/2})$ and the rest is $\bhp{1}$. Inserting $\cp=\cp_0+\cp_1 N^{-1/2}+ \bhp{N^{-1}}$ and using $s_2'(z)=-2s_3(z)$, we find that 
\beq \begin{split}
	& N (\Govl(\cpovl, \cpovl;  a) - 2\G(\cp)) = \xi^2 (1+A^2) s_2(\cp_0) \\
	&\quad + 2\ef^2 \beta \left( \xi^2 (A-1)^2 s_3(\cp_0)
	 + \xi \go_N(\cp_0; 2)  + \xi \sqrt{N} s_2(\cp_0) -2 \xi s_3(\cp_0) \cp_1\right) +\bhp{ N^{-\frac12} }. 
\end{split} \eeq

We now consider the integrals in  \eqref{eq:ovlftsf}. Since all partial derivatives of $\Govl(z,w)$ evaluated at the critical point $(z,w)=(\cpovl, \cpovl)$ are $\bhp{1}$, the two dimensional method of steepest descent applies.  Since the second derivatives evaluated at the critical point do not depend on $\xi$, we find that the ratio of the integrals in \eqref{eq:ovlftsf} is asymptotically equal to $1$.

Combining the computations above, we find that 
\beq \begin{split}
	\log \langle e^{\beta \xi \sqrt{N} \ovl} \rangle 
	\simeq 
	& \frac12 \xi^2 (1+A^2) s_2(\cp_0) \\
	& \qquad + \ef^2 \beta \left( \xi^2 (A-1)^2 s_3(\cp_0) 
	+ \xi \go_N(\cp_0; 2)  + \xi \sqrt{N} s_2(\cp_0) -2 \xi s_3(\cp_0) \cp_1\right) 
\end{split} \eeq
where $A$ is given by \eqref{eq:oclVaAf}. 
Using the formula \eqref{eq:cp1h>0} of $\cp_1$, we obtain 
\beq
	\go_N(\cp_0; 2)   -2 s_3(\cp_0) \cp_1
	= \frac{T s_2(\cp_0)  }{T s_2(\cp_0)+2h^2 s_3(\cp_0)} \go_N(\cp_0; 2) .
\eeq
Hence, we conclude the following.

\begin{result} 
For $h>0$ and $T>0$, 
\beq
	\log \langle e^{\xi \sqrt{N} (\ovl- \ef^2 s_2(\cp_0))} \rangle 
	\simeq 
	 \frac{ \ef^2 T s_2(\cp_0) \go_N(\cp_0; 2)  }{T s_2(\cp_0)+2h^2 s_3(\cp_0)}  \xi
	+ \frac{ T^2 s_2(\cp_0)( T s_2(\cp_0) + 4\ef^2 s_3(\cp_0))   }{2 (T s_2(\cp_0)+2h^2 s_3(\cp_0))} \xi^2
\eeq
as $N\to \infty$ for asymptotically almost every disorder sample, where $\cp_0>2$ is the solution of the equation 
$1 - T s_1(\cp_0) - \ef^2  s_2(\cp_0) = 0$, and $\go_N(z;k)$ is defined in \eqref{eq:go00}. 
\end{result}

As a consequence, we obtain the following.

\begin{result} \label{result:replica1}
For $h>0$ and $T>0$, 
\beq	\label{eq:ovlresh>0}
	\ovl \simeqidsgibbs h^2 s_2(\cp_0) + \frac1{\sqrt{N}} \left[  \frac{ \ef^2 T s_2(\cp_0) \go_N(\cp_0; 2)  }{T s_2(\cp_0)+2h^2 s_3(\cp_0)} 
	+ \sigmaovl \gib N \right] 
\eeq
as $N\to \infty$ for asymptotically almost every disorder sample, where the thermal random variable $\gib N$ has the standard normal  distribution and $\sigmaovl>0$ satisfies  
\beq	\label{eq:ovlvarh>0}
	\sigmaovl^2= \frac{T^2 s_2(\cp_0)( T s_2(\cp_0) + 4\ef^2 s_3(\cp_0))   }{T s_2(\cp_0)+2h^2 s_3(\cp_0)} . 
\eeq
\end{result}

\subsubsection{Discussion of the leading term}

The leading term 
\beq \label{eq:ovreplho1}
	\ovlz = \ovlz(T, h)= \ef^2 s_2(\cp_0)=1-Ts_1(\cp_0)
\eeq 
in \eqref{eq:ovlresh>0} depends on neither the choice of spin configuration nor the disorder sample. 
See Figure \ref{fig:r0_vs_h} for the graph of $\ovlz$ as a function of $\ef$. 


The value \eqref{eq:ovreplho1} for $\ovlz$ reproduces the prediction $q_0$ for the overlap
obtained in \cite{crisanti1992sphericalp, FyodorovleDoussal} from 
the replica saddle methods which predicts that $q_0$ is determined by \eqref{eq:q0overlap}, 
The equivalence is checked using that $s_2(z)=s_1(z)^2/(1-s_1(z)^2)$ and $q_0=1-T s_1(\cp_0)$. 

It is easy to check the following properties using a computation similar to the one in Subsection \ref{sec:mgnefpos}: 
\begin{itemize}

\item For every $T>0$, $\ovlz$ is an increasing function of $\ef>0$.
 
\item As $h\to\infty$, 
\beq
	\ovlz= 1- \frac{T}{\ef} + O(\ef^{-2})  \qquad \text{for all $T>0$.}
\eeq
\item 
As $\ef\to 0$, 
\beq \label{eq:oplhsoz}
	\ovlz= \begin{dcases}
	\frac{\ef^2}{T^2-1} - \frac{2T^2 \ef^4}{(T^2-1)^2} + O(\ef^6) \qquad &\text{for $T>1$,} \\
	1-T + \frac{T\ef^2}{2(1-T)}+ O(\ef^4) \qquad &\text{for $0<T<1$.}
	\end{dcases}
\eeq
\end{itemize}

\begin{figure}
\centering
\begin{subfigure}[t]{0.45\textwidth}
\includegraphics[width = 0.9\textwidth]{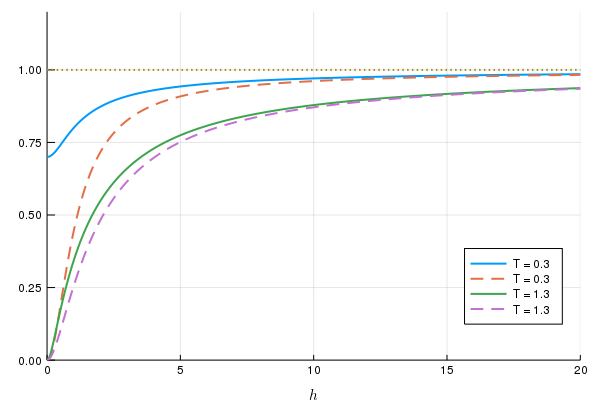}
\caption{Graph of $\ovlz$ (solid line) and $(\mgn^0)^2$ (dashed line) as a function of $\ef$ for $T=0.3$ and $T=1.3$}
\label{fig:r0_vs_h}
\end{subfigure}
\begin{subfigure}[t]{0.45\textwidth}
\includegraphics[width = 0.9\textwidth]{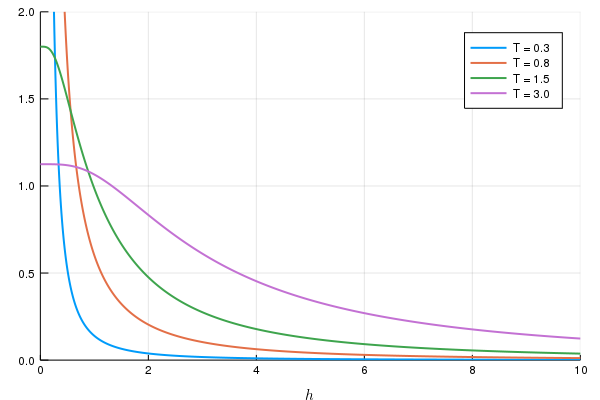}
\caption{Graph of $\sigmaovl^2$ as a function of $\ef$.}
\label{fig:ovlvar_vs_h}
\end{subfigure}
\caption{Graphs of $\ovlz$ and $\sigmaovl^2$.}
\end{figure}

\subsubsection{Discussion of the thermal variance}

The thermal variance of $\ovl$ satisfies 
\beq
	\langle \ovl^2 \rangle - \langle \ovl \rangle^2 \simeq \frac{\sigmaovl^2}{N}
\eeq
for $\sigmaovl^2$ given in \eqref{eq:ovlvarh>0} and it does not depend on the disorder sample. See Figure \ref{fig:ovlvar_vs_h} for the graph. It is a decreasing function of $\ef$, and satisfies
\beq 
	\sigmaovl^2 = \frac{2T^2}{h^2} - \frac{5T^3}{2\ef^3}+ O(h^{-4}) \quad\text{as $h\to \infty$ for all $T>0$}
\eeq 
and 
\beq \label{eq:ovvash}
	\sigmaovl^2(h, T) = 
	\begin{dcases}
	\frac{\tmp^2}{\tmp^2 - 1}    + O(\ef^4)\qquad &\text{as $\ef\to 0$ for $T>1$,}\\ 
	\frac{2\tmp^2(1 - \tmp)}{\ef^2} + O(1) \qquad & \text{as $\ef\to 0$ for $0<T<1$.} 
	\end{dcases}
\eeq


\subsubsection{Limit as $\ef \rightarrow \infty$}

As $\ef \to \infty$, using \eqref{eq:cpvalueforhplh} and $s_k(z)= z^{-k} + O(z^{-k-2})$ as $z\to \infty$, we find that
\beq
	 \frac{ \ef^2 T s_2(\cp_0) \go_N(\cp_0; 2)  }{T s_2(\cp_0)+2h^2 s_3(\cp_0)} 
	 \simeq \frac{\tmp  \sum_{i = 1}^N (n_i^2 - 1)}{2\ef\sqrt{N}}.
\eeq
Thus, we see that, for every $T>0$, if we take $N\to \infty$ with $h>0$ and then take $h\to \infty$,  
\beq 
	\ovl \simeqidsgibbs 1- \frac{T}{h}+ \frac{T}{\ef\sqrt{N} } \left[ \frac{  \sum_{i = 1}^N (n_i^2 - 1) }{2\sqrt{N}}
	+  \sqrt{2} \gib N \right] . 
\eeq

\subsubsection{Limit as $\ef \rightarrow 0$ when $\tmp > 1$}

Using \eqref{eq:cpohs}, if we take $N\to \infty$ with $h>0$ and then take $h\to 0$, we see that, for $T>1$, 
\beq
	\ovl \simeqidsgibbs \frac{\ef^2}{\tmp^2 - 1} - \frac{2\tmp^2 \ef^4}{(\tmp^2 - 1)^2} + \frac{1}{\sqrt{N}} \left[ \frac{\tmp}{\sqrt{\tmp^2 - 1}}\gib N + \ef^2 \go_N(\tmp + \frac1\tmp; 2) \right] .
\eeq

\subsubsection{Limit as $\ef \rightarrow 0$ when $\tmp < 1$}

Similarly, from \eqref{eq:cpohs}, if we take $N\to \infty$ with $h>0$ and then take $h\to 0$, we see that, for $0<T<1$, 
\beq	\label{eq:ovlh>0ltmp}
	\ovl \simeqidsgibbs (1- \tmp) + \frac{T  \ef^2}{2(1-T)}+ \frac{\tmp}{\ef \sqrt{N}} \left[\frac{ \ef^5 \go_N(\cp_0;2)}{2(1 - \tmp)^2} +  \sqrt{2(1 - \tmp)} \gib N \right].
\eeq
From the discussions around the equation \eqref{eq:goNh0ltmp}, we expect that $\ef^5 \go_N(\cp_0;2)=\bhp{1}$ as $\ef\to 0$ if $h\gg N^{-1/6}$.
This indicates that there may be a transition when $h\sim N^{-1/6}$. We study this regime in the next subsection. 
On the other hand, the thermal fluctuation term becomes of order $1$ if $h^{-1}N^{-1/2}=O(1)$. This indicates a new regime $h\sim N^{-1/2}$, which we study in a later section.

\subsection{Mescoscopic external field: $h\sim N^{-1/6}$ and $T<1$}
\label{sec:ovl1/6}

\subsubsection{Analysis}


Set
\beq
	h=HN^{-1/6}
\eeq
for fixed $H>0$. 
It turns out that the order of the fluctuations of $\ovl$ is $N^{-1/3}$. 
Hence, we set 
\beq
	\eta= \beta \xi N^{1/3} \quad \text{so that} \quad 	a= \xi N^{-2/3}. 
\eeq

The critical point of $\G(z)$ is given by $\cp= \eg_1 + sN^{-2/3}$ 
where $s>0$ solves the equation \eqref{eq:seqgeneral}. 
Inserting $h=HN^{-1/6}$, the equation takes the form
\beq \label{eq:Gpriaan10101s}
	 \beta- \frac1{N^{1/3}} \sum_{i=1}^N \frac{1}{s+a_1-a_i}  -  \efres^2 \beta \sum_{i=1}^N \frac{n_i^2}{(s + \egres_1 -\egres_i)^2} =0 . 
\eeq
The solution satisfies $s=\stild+ \bhp{N^{-1/3}}$ where $\stild$ solves the equation \eqref{eq:stildeq}.
 
For the critical point of $\Govl$, Lemma \ref{lem:2olboundingcpovl}) shows that $\cp< \cpovl< \cp+a$. Hence, $\cpovl-\cp-a= O(N^{-2/3})$. However, we can get a sharper bound on this difference.  The right-hand side of \eqref{eq:diffggga} is $\bhp{aN^{4/3}}$ and the bracket term of the left-hand side of the same equation is $\bhp{N^{5/3}}$, with the leading contribution coming from the second sum.  Hence, we find that 
\beq
	\cpovl = \gamma+a - \epsilon, \qquad \epsilon= \bhp{N^{-1}}. 
\eeq

We now evaluate \eqref{eq:Gdiffggaa}.
The first sum $B_1$ is 
\beqq \begin{split}
	 -  \sum_{i=1}^N \left[ \log \left( 1+ \frac{2(a-\epsilon)}{\cp-\eg_i} - \frac{(2a-\epsilon)\epsilon}{(\cp-\eg_i)^2} \right)  - \frac{2(a-\epsilon)}{\cp-\eg_i} \right] 
	\simeq  - \sum_{i=1}^N \left[ \log \left( 1+ \frac{2\xi}{s+a_1-a_i} \right)  - \frac{2\xi}{s+a_1-a_i}\right]  
\end{split} \eeqq
and this sum is $\bhp{1}$. 
For the second sum, we get 
\beqq
	B_2 
	= 2 \xi N^{1/3} H^2\beta  \sum_{i=1}^N \frac{n_i^2}{(s+a_1-a_i)^2}+ \bhp{N^{-1/3}}.
\eeqq
Therefore, $N (\Govl(\cpovl, \cpovl;  a) - 2\G(\cp))$ is equal to 
\beq \begin{split}
	- \sum_{i=1}^N \left[ \log \left( 1+ \frac{2\xi}{s+a_1-a_i} \right)  - \frac{2\xi}{s+a_1-a_i}\right] 
	 + 2 \xi N^{1/3}  H^2\beta \sum_{i=1}^N \frac{n_i^2}{(s+a_1-a_i)^2} +\cO(N^{-1/3}).
\end{split} \eeq
Using the equation \eqref{eq:Gpriaan10101s} for $s$, we can write  
\beq \label{eq:poiiuur}
	N (\Govl(\cpovl, \cpovl;  a) - 2\G(\cp))
	= 2\xi \beta N^{1/3}  - \sum_{i=1}^N  \log \left( 1+ \frac{2\xi}{s+a_1-a_i} \right) 
	+\cO(N^{-1/3}).
\eeq

Finally, we compute the integrals in \eqref{eq:ovlftsf}. A calculation similar to the one from Subsection \ref{sec:fetranltmp} shows that the $k$th partial derivatives of $\Govl$ evaluated at $(z,w)=(\cpovl, \cpovl)$ are $\bhp{N^{\frac23k-\frac23}}$. 
Since the second derivatives are $\bhp{N^{\frac23}}$, the main contribution to the integral comes from a neighborhood of radius $N^{-5/6}$ around the critical point. 
Moreover, from explicit computations, we find that
\beqq
	\frac{\partial^2 \Govl}{\partial z^2}(\cpovl, \cpovl)= \frac{\partial^2 \Govl}{\partial w^2}(\cpovl, \cpovl)\simeq x N^{2/3}, 
	\qquad
	\frac{\partial^2 \Govl}{\partial z \partial w}(\cpovl, \cpovl)\simeq y  N^{2/3}
\eeqq
where
\beqq
	x= 2H^2\beta \sum_{i = 1}^N  \frac{n_i^2 (s + \egres_1 - \egres_i + \xi)}{(s + \egres_1 - \egres_i)^3(s + \egres_1 - \egres_i + 2\xi)} , 
	\qquad
	y= 2H^2\beta \sum_{i = 1}^N  \frac{n_i^2 \xi }{(s + \egres_1 - \egres_i)^3(s + \egres_1 - \egres_i + 2\xi)} .
\eeqq
Using the method of steepest descent with the change of variables $z=\cpovl+ u N^{-5/6}$ and $w=\cpovl+ v N^{-5/6}$, the integral becomes 
\beq\begin{split}
	&\int_{\cpovl+\ii \R} \int_{\cpovl+\ii \R} e^{\frac{N}{2} (\Govl(z, w; a) - \Govl(\cpovl, \cpovl;  a))} \dd z \dd w 
	\simeq \frac1{N^{5/3}} \int_{\ii\R} \int_{\ii\R} e^{\frac14 (xu^2+xv^2+2yuv) } \dd u \dd v . 
\end{split}\eeq
Evaluating the Gaussian integral, inserting the formulas of $x$ and $y$, and noting that the denominator is the same as the numerator when $\xi=0$, the ratio of the integrals becomes
\beq
	\frac{\int  \int e^{ \frac{N}{2} (\Govl(z, w; a) - \Govl(\cpovl, \cpovl;  a)) } \dd z \dd w}{\left( \int e^{\frac{N}2 (\Govl(z) - \G(\cp))} \dd z \right)^2 } \simeq 	\sqrt{ \frac{\sum_{i=1}^N \frac{n_i^2}{(s+a_1-a_i)^3}}{\sum_{i=1}^N \frac{n_i^2}{(s+a_1-a_i)^2(s+a_1-a_i+2\xi)} } } . 
\eeq

Combining the above calculations and replacing $s$ by $\stild$, we obtain the following result after moving a term of order $N^{1/3}$. 

\begin{result}
For $h=H N^{-1/6}$ and $0<T<1$, 
\beq \label{eq:resmgnovlN1/6} \begin{split}
	\langle e^{\frac1{T}  \xi N^{1/3} \left( \ovl - (1 - \tmp)  \right) } \rangle 
	&\simeq e^{\xi  N^{1/3} - \frac12 \sum_{i=1}^N  \log \left( 1+ \frac{2\xi}{\stild+a_1-a_i} \right)  }
	\sqrt{ \frac{\sum_{i=1}^N \frac{n_i^2}{(\stild+a_1-a_i)^3}}{\sum_{i=1}^N \frac{n_i^2}{(\stild+a_1-a_i)^2(\stild+a_1-a_i+2\xi)} } } 
\end{split} \eeq
as $N\to \infty$ for asymptotically almost every disorder sample, where
$\stild>0$ is the solution of the equation \eqref{eq:stildeq}. 
\end{result}


The term in the exponent on the right-hand side is $\bhp{1}$. 

\begin{result} \label{result:replica16}
For $h=H N^{-1/6}$ and $0<T<1$, 
\beq \label{eq:ovconcwhnn1}
	\ovl \simeqidsgibbs 1- T + \frac{\tmp}{N^{1/3}} \sumws_N(\stild)
\eeq
as $N\to \infty$ for asymptotically almost every disorder sample, where $\stild>0$ is the solution of the equation \eqref{eq:stildeq}
and $\sumws_N(\stild)$ is a random variable defined by the generating function given by the right-hand side of \eqref{eq:resmgnovlN1/6}. 
\end{result}

\subsubsection{Matching with $\ef=O(1)$}

We take the formal limit of the result \eqref{eq:ovconcwhnn1} as $H\to \infty$. From \eqref{eq:sqrs0whenehf}, $\stild\to\infty$. 
The big square root term of the generating function on the right-hand side of \eqref{eq:resmgnovlN1/6} is approximately  is approximately $1$. 
On the other hand, 
\beqq
	\xi  N^{1/3} - \frac12 \sum_{i=1}^N  \log \left( 1+ \frac{2\xi}{\stild+a_1-a_i} \right) 
	\simeq  \xi \left( N^{1/3} - \sum_{i=1}^N \frac1{\stild + a_1-a_i} \right) + \xi^2 \sum_{i=1}^N \frac1{(\stild+a_1-a_i)^2}
\eeqq
Setting $x=\eg_1+\stild N^{-2/3}$, we have, using a formal application of the semi-circle law, 
\beqq
	N^{1/3} - \sum_{i=1}^N \frac1{\stild + a_1-a_i} = N^{1/3} \left(1- \frac1{N} \sum_{i=1}^N \frac1{x-\eg_i} \right) 
	\simeq N^{1/3} \left( 1- s_1(x) \right). 
\eeqq
Using \eqref{eq:stjaspt}, the above equation becomes
 \beqq
	N^{1/3} - \sum_{i=1}^N \frac1{\stild + a_1-a_i} \simeq N^{1/3} \sqrt{x-2} \simeq \sqrt{\stild}. 
\eeqq
For the other term, 
\beqq
	\sum_{i=1}^N \frac1{(\stild+a_1-a_i)^2} = \frac1{N^{4/3}} \sum_{i=1}^N \frac1{(x-\eg_i)^2} 
	\simeq \frac1{N^{1/3}} s_2(x) \simeq \frac1{N^{1/3}2\sqrt{x-2}}\simeq \frac1{2\sqrt{\stild}}.  
\eeqq
Hence, the generating function on the right-hand side of \eqref{eq:resmgnovlN1/6} is approximately $e^{\sqrt{\stild}\xi + \frac{\xi^2}{2\sqrt{\stild}}}$.
Therefore, 
\beqq
	\sumws_N(\stild) \simeqidsgibbs \sqrt{\stild} + \stild^{-1/4} \gib N 
\eeqq
for a thermal standard normal random variable $\gib N$. 
Inserting the large $H$ formula \eqref{eq:sqrs0whenehf} for $\stild$ and replacing $H=\ef N^{1/6}$, we find that if we take $h=HN^{-1/6}$ and let $N\to\infty$ first and then take $H\to \infty$, we get 
\beq	\begin{split}
	\ovl 
	&  \simeqidsgibbs 1- \tmp + \frac{\tmp\ef^2}{2(1 - T)} + \frac{\tmp}{\ef N^{1/2}}\left[\frac{ \ef^5 \go_N(\cp_0; 2)}{2(1 - \tmp)^2} + \sqrt{ 2(1 - \tmp)} \gib N \right].
\end{split} \eeq
This is the same as \eqref{eq:ovlh>0ltmp} which is obtained by first taking $N\to \infty$ with $h>0$ fixed and then taking $h\to 0$. 
Therefore, the result matches with the $h=O(1)$ case.


\subsubsection{Limit as $\efres \rightarrow 0$}

From \eqref{eq:slimH0}, $\stild=O(H)\to 0$ as $H\to 0$. 
The generating function on the right-hand side of \eqref{eq:resmgnovlN1/6} converges to 
\beqq
	e^{\xi  N^{1/3} - \frac12 \sum_{i=2}^N  \log \left( 1+ \frac{2\xi}{\stild+a_1-a_i} \right)  }
\eeqq
where the term $i=1$ cancels out with the limit of the big square root term. 
Using the moment generating function \eqref{eq:noncechi} for the chi-squared distribution, we find that if we take $h=HN^{-1/6}$ and $N\to \infty$ and then take $H\to 0$, then 
\beq \label{eq:hhhhw}
	\ovl \simeqidsgibbs 1- T + \frac{\tmp}{N^{1/3}}\left( N^{1/3} -   \sum_{i=2}^N \frac{\gib n_i^2}{ \egres_1-\egres_i} \right) . 
\eeq
for independent thermal standard Gaussian random variables $\gib n_i$.

\subsection{Microscopic external field: $h\sim HN^{-1/2}$ and $T<1$} \label{sec:ovl12}

\subsubsection{Analysis}

Set
\beq
	h= H N^{-1/2}
\eeq
for fixed $H>0$. 
It turns out that the fluctuations are of order $\bhp{1}$. In other words, the leading term of $\ovl$ converges to a random variable. 
We set 
\beq
	\eta= \beta \xi \quad \text{so that} \quad a= \xi N^{-1}. 
\eeq

The critical point of $\G(z)$ is $\cp= \eg_1 + \so N^{-1}$ from \eqref{eq:sohca}.
Consider the critical point of $\Govl$. Lemma \ref{lem:2olboundingcpovl} implies that $\cpovl= \eg_1 + \bhp{N^{-1}}$. 
We set 
\beq
	\cpovl= \eg_1 + \sovl N^{-1}, \qquad \sovl> |\xi|, 
\eeq
for some $\sovl$. Separating $i=1$ in the equation \eqref{eq:cpovldef}, we find that $\sovl$ is the solution of the equation 
\beq \label{eq:eqaforsov}
	\beta - 1 - \frac{\sovl}{\sovl^2-\xi^2} - \frac{H^2\beta n_1^2}{(\sovl - \xi)^2} + \bhp{N^{-1/3}}=0. 
\eeq
When $\beta= T^{-1}>1$, the equation $\beta - 1 - \frac{x}{x^2-\xi^2} - \frac{H^2\beta n_1^2}{(x - \xi)^2}=0$ has a unique solution and  $\sovl$ is approximated by this solution with error $\bhp{N^{-1/3}}$.

Using \eqref{eq:Gdiffggaa} and separating out the $i=1$ term, we find that $N (\Govl(\cpovl, \cpovl;  a) - 2\G(\cp))$ is equal to 
\beq \begin{split}
	 -  \log \left( \frac{\sovl^2 - \xi^2}{\so^2} \right)  + \frac{2(\sovl-\so)}{\so}  
	 + 2 H^2 \beta n_1^2  \left[ \frac{1}{\sovl - \xi}
	- \frac1{\so} + \frac{\sovl-\so}{\so^2} \right] + \bhp{N^{-1/3}} .
\end{split} \eeq
Using the equation for $\so$, this can be written as
\beq \begin{split}
	&N (\Govl(\cpovl, \cpovl;  a) - 2\G(\cp))  \\
	& = -  \log \left( \frac{\sovl^2 - \xi^2}{\so^2} \right)  + 2(\beta-1) (\sovl-\so)  
	 + 2 H^2 \beta n_1^2  \left[ \frac{1}{\sovl - \xi} 
	- \frac1{\so}\right] + \bhp{N^{-1/3}}.
\end{split} \eeq

We now consider the integrals in \eqref{eq:ovlftsf}.
As in Subsection \ref{sec:mgnlowtmn12} of the overlap with the external field when $h\sim N^{-1/2}$, the main contribution to the integral comes from a neighborhood of radius $N^{-1}$ around the critical point in both variables. 
Changing variables to $z=\cpovl+ uN^{-1}$ and $w=\cpovl+ v N^{-1}$, we find that all terms of the Taylor series are of the same order, so we see, as in Subsection \ref{sec:mgnlowtmn12}, that the 
integral is not approximated by a Gaussian integral. 
Therefore, we proceed by writing 
\beqq \begin{split}
	& N (\Govl(z, w; a) - \Govl(\cpovl, \cpovl;  a)) \\
	& = N (\Govl(z, w; a)- \Govl(\cpovl, \cpovl;  a) - (\Govl)_z (\cpovl, \cpovl; a) (z-\cpovl) - (\Govl)_w (\cpovl, \cpovl; a) (w-\cpovl) ) \\
	&= - \sum_{i=1}^N \left[ \log \left( \frac{(z-\eg_i)(w-\eg_i)-a^2}{(\cpovl-\eg_i)^2 - a^2} \right)  
	- \frac{(\cpovl-\eg_i)(z+w-2\cpovl)}{(\cpovl-\eg_i)^2-a^2} \right]  \\
	&\qquad  
	+ \ef^2 \beta \sum_{i=1}^N n_i^2  \left[ \frac{z+w-2\eg_i+ 2a}{(z-\eg_i)(w-\eg_i)-a^2} - \frac{2}{\cpovl-\eg_i - a}
	+ \frac{z+w-2\cpovl}{(\cpovl-\eg_i -a)^2}  \right]  .
\end{split} \eeqq
Inserting the change of variables and separating $i=1$ out,
\beqq \begin{split}
	N (\Govl(z, w; a) - \Govl(\cpovl, \cpovl;  a)) 
	\simeq& - \log \left( \frac{(u+\sovl) (v+\sovl)-\xi^2}{\sovl^2 - \xi^2} \right)   
	+ \frac{\sovl (u+v)}{\sovl^2-\xi^2}  \\
	&+ H^2 \beta n_1^2 \left[ \frac{u+v +2\sovl+ 2\xi}{ (u+\sovl) (v+\sovl)-\xi^2}  - \frac{2}{\sovl - \xi} 
	+ \frac{u+v}{(\sovl - \xi)^2} \right]  
\end{split} \eeqq
for finite $u$ and $v$. 
Using the equation \eqref{eq:eqaforsov}, this can be written as  
\beqq \begin{split}
	&N (\Govl(z, w; a) - \Govl(\cpovl, \cpovl;  a)) \\
	&\simeq - \log \left( \frac{(u+\sovl) (v+\sovl)-\xi^2}{\sovl^2 - \xi^2} \right)   
	+ (\beta-1) (u+v)   
	 + H^2 \beta n_1^2 \left[ \frac{u+v +2\sovl+ 2\xi}{ (u+\sovl) (v+\sovl)-\xi^2} - \frac{2}{\sovl- \xi}  \right]  .
\end{split} \eeqq
Thus, the numerator integral in \eqref{eq:ovlftsf} is asymptotically equal to 
\beqq\begin{split}
	\frac{\sqrt{\sovl^2-\xi^2}}{N^2}  \int \int \frac{e^{\frac12 (\beta-1)(u+v) + \frac{H^2 \beta n_1^2}{2} \left[ \frac{u+v +2\sovl+ 2\xi}{ (u+\sovl) (v+\sovl)-\xi^2} - \frac{2}{\sovl- \xi}  \right]} }{\sqrt{(u+\sovl) (v+\sovl)-\xi^2}}  \dd u \dd v
\end{split}\eeqq
where the contours are from $-\ii\infty$ to $\ii\infty$ such that all singularities lie on the left of the contours.
The denominator integral is the same with $\xi=0$. 

Combining the above calculations into \eqref{eq:ovlftsf} and making simple translations for the integral, we find that
\beq \label{eq:wsnithe}
	\langle e^{\beta \xi \ovl} \rangle 
	\simeq 
	\frac{\int \int  \frac1{\sqrt{uv-\xi^2}} e^{\frac12 (\beta-1)(u+v) + \frac{H^2\beta n_1^2(u+v + 2\xi) }{2(uv-\xi^2)} } \dd u \dd v}
	{\left( \int \frac1{\sqrt{u}} e^{\frac12 (\beta-1) u + \frac{H^2\beta n_1^2}{2u}} \dd u \right)^2 }  
\eeq
where the contours are vertical lines such that the points $\xi$ or $0$ lie on the left of the contours. 
We now evaluate the integrals using (recall \eqref{eq:ntif}) 
\beq \label{eq:icuJB}
	\int \frac{e^{au+ \frac{b}{u}}}{\sqrt{u}}  \dd u = \frac{2\ii \sqrt{\pi}}{\sqrt{a}} \cosh(2\sqrt{ab}) .
\eeq 
Consider the double integral in the numerator. 
For each $v$, we change the variable $u$ to $z$ by setting $uv-\xi^2=z$. We can define the branch cut appropriately such that the contour for $z$ does not cross the branch cut. 
The numerator becomes
\beqq
	\iint \frac1{v\sqrt{z}} e^{\frac{\beta-1}2( \frac{z+\xi^2}{v}+v) + \frac{H^2\beta n_1^2}{2z} (\frac{z+\xi^2}{v}+ v + 2\xi) } \dd z \dd v .
\eeqq
The $z$-integral can be evaluated using \eqref{eq:icuJB}. 
Writing the resulting $\cosh$ term as the sum of two exponentials, we can 
evaluate the $w$-integral again using \eqref{eq:icuJB}. 
The above double integral becomes
\beqq \begin{split}
	- \frac{2\pi}{\beta-1} \bigg[
	& e^{\sqrt{(\beta-1)\beta} H |n_1|} \cosh\left( \sqrt{(\beta-1)\beta} H |n_1| + (\beta-1)\xi \right)  \\
	&\quad + e^{-\sqrt{(\beta-1)\beta} H |n_1|} \cosh\left( \sqrt{(\beta-1)\beta} H |n_1| - (\beta-1)\xi \right) 
	\bigg].
\end{split} \eeqq
Writing $\cosh$ as the sum of two exponentials again, the expression above becomes a linear combination of $e^{(\beta-1)\xi}$ and $e^{-(\beta-1)\xi}$. 
The denominator in \eqref{eq:wsnithe} is the same as the numerator when $\xi=0$. 
Thus, using $\beta=1/T$ and re-scaling $\xi$, we obtain the following

\begin{result}
For $h=HN^{-1/2}$ and $0<T<1$, 
\beq
	\langle e^{ \xi \frac{\ovl}{1-T} } \rangle 
	\simeq
	\frac{ \cosh\left(   \frac{2\sqrt{1-T}H |n_1| }{T} \right) e^{\xi} + e^{-\xi} }
	{\cosh\left( \frac{2\sqrt{1-T}H |n_1| }{T} \right) +1} 
\eeq
as $N\to \infty$ for asymptotically almost every disorder sample. 
\end{result}

Recognizing that the right-hand side is the moment generating function of a shifted Bernoulli random variable, we obtain the following result. 

\begin{result}\label{result:ovl12mainr}
For $h=HN^{-1/2}$ and $0<T<1$,  
\beq \label{eq:ovl1/2main}
	\frac{\ovl}{1-T} \simeqidsgibbs \Bern (\theta) , 
	\qquad 
	\theta:=\frac{ \cosh\left(  \frac{2\sqrt{1-T}H |n_1| }{T} \right) }{\cosh\left(  \frac{2\sqrt{1-T}H |n_1| }{T}  \right) +1} 
\eeq
as $N\to \infty$ for asymptotically almost every disorder sample, where the thermal random variable $\Bern(c)$ is the (shifted) Bernoulli distribution taking values $1$ and $-1$ with probability $c$ and $1-c$, respectively. 
\end{result}

\subsubsection{Limits as  $H\to\infty$}

If we formally take the limit as $H\to \infty$ of the result \eqref{eq:ovl1/2main}, then  
\beq \label{eq:ovlHN1/2inf}
	\ovl \simeqidsgibbs 1-T.
\eeq
This is the same as the leading term of \eqref{eq:hhhhw} which is obtained by taking $h=HN^{-1/6}$ 
and letting $N\to \infty$ first and then taking $H\to 0$.


\subsection{No external field: $h=0$}\label{sec:2olh=0}

For $0<T<1$, the analysis in Subsection \ref{sec:ovl12} for $h=HN^{-1/2}$ extends to $H=0$ case as well. 
For $T>1$, the analysis in Subsection \eqref{sec:ovlpos} applies to all $h\ge 0$. We note that, for $h=0$ and $T>1$, $\cp_0=T+T^{-1}$ and $s_2(\cp_0)=\frac1{T^2-1}$. 
We have the following result.

\begin{result}  \label{result:replica0}
For $h=0$, 
\beq \label{eq:ovlh=0}
	\ovl\simeqidsgibbs \begin{cases}
	\frac{T}{\sqrt{N (T^2-1)}} \gib N  \qquad &\text{for $T>1$,}\\
	(1-T) \Bern (1/2) \qquad & \text{for $0<T<1$.} 
	\end{cases}
\eeq
\end{result} 

\section{Geometry of the spin configuration} \label{sec:Geometry}

The results on three types of overlaps tell us how the spin variables are distributed on the sphere. We discuss the geometry of the spin configuration vector $\spin=(\sigma_1, \cdots, \sigma_N)$ from the Gibbs measure in this section. 
Recall that $\vu_1$ is a unit vector which is parallel to the eigenvector corresponding to the largest eigenvalue of the disorder matrix. 
In this section, we choose $\vu_1$, among two opposite directions, as the one satisfying $\vu_1\cdot \efv\ge 0$. 
Recall the notation $n_1= \efv\cdot \vu_1$ and that the external field $\efv$ is a standard Gaussian vector. Note that $n_1=|n_1|$ because of the choice of $\vu_1$. 
The normalized spin vector can be decomposed as 
\beq
	\hsigma:= \frac{\spin}{\sqrt{N}}= a \vu_1 + b \frac{\efv-n_1 \vu_1}{\|\efv-n_1 \vu_1\|} + \mv, \qquad \mv\cdot \vu_1=\mv\cdot \efv=0,
\eeq
where $a$ and $b$ are components of the normalized spin vector in the $\vu_1$ and $\efv_1-n_1 \vu_1$ directions, respectively. 
The vector $\mv$ is perpendicular to both $\vu_1$ and $\efv$, and it satisfies 
\beq
	\|\mv\|^2=1 - a^2-b^2. 
\eeq
Note that $\|\efv-n_1 \vu_1\|^2=\|\efv\|^2-n_1^2 \simeq N + \bhp{N^{1/2}}$ and $n_1=\bhp{1}$. 
Thus, if we ignore subleading terms from each component, the above decomposition becomes
\beq
	\hsigma\simeq a\vu_1 + b \frac{\efv}{\sqrt{N}} + \mv =  a\vu_1 + b \hefv + \mv, \qquad \hefv:= \frac{\efv}{\sqrt{N}}. 
\eeq
The components $a$ and $b$ are related to the overlaps by the formulas
\beq \label{eq:abfrov}
	\OM = (\hsigma\cdot \vu_1)^2 = a^2  , \qquad \mgn = \hsigma\cdot \hefv = \frac{an_1}{\sqrt{N}} + b \frac{\|\efv-n_1 \vu_1\|}{\sqrt{N}}
	\simeq   \frac{an_1}{\sqrt{N}} + b 
\eeq
up to $\bhp{N^{-1}}$ terms. 
Furthermore, $\mv$ satisfies the equation 
\beq \label{eq:cfrov}
	\ovl= \hsigma^{(1)}\cdot \hsigma^{(2)} = a_1a_2+ b_1b_2 + \mv^{(1)} \cdot \mv^{(2)}.
\eeq

\subsection{The signed overlap with a replica for microscopic field, $h\sim N^{-1/2}$ and $T<1$}

Consider the decomposition for $h= HN^{-1/2}$ and $0<T<1$. 
The overlap with the ground state is given in Result \ref{thm:ground1/3} for $h\sim N^{-1/3}$ and Result \ref{result:grounds0} for $h=0$. 
Since the leading terms of the both results are same, given by $1-T$, the leading term holds also for $h\sim N^{-1/2}$. 
Thus, we find that $a^2 \simeq 1-T$ in this regime, and hence $|a|\simeq \sqrt{1-T}$. 
On the other hand, Result \ref{thm:ext1/2} on $\mgn$ implies that 
\beq \label{eq:b12c}
	\frac{an_1}{\sqrt{N}} + b \simeqidsgibbs h + \frac{|n_1| \sqrt{1-T} \gib B(\alpha)}{\sqrt{N}} + \frac{\sqrt{T} \gib N}{\sqrt{N}}. 
\eeq
Noting $h\sim N^{-1/2}$, we find that $b= \bhp{ N^{-1/2}}$. 
From the formulas of $a$ and $b$, we also find that $\|\mv\|^2=1-a^2-b^2 \simeq T$. Finally, Result \ref{result:ovl12mainr} implies that 
\beq \label{eq:vv12}
	a_1a_2+ b_1b_2 + \mathbf{v}^{(1)} \cdot \mathbf{v}^{(2)} \simeqidsgibbs (1-T) \gib B(\theta).
\eeq
Here, $\theta$ is given in \eqref{eq:ovl1/2main} and $\alpha$ in \eqref{eq:b12c} is given by \eqref{eq:probY}. They satisfy the relation $\theta= \alpha^2 +(1-\alpha)^2$. 
Now, we make the following ansatz on $a$. For $h=0$ and $0<T<1$, the spin configurations are equally likely to be on either of the double cones around $\vu_1$ with the cosine of the angle given by $\sqrt{1-T}$. This means that $a\simeqidsgibbs \sqrt{1-T} \Bern(1/2)$ for $h=0$ and $0<T<1$. 
For $h\sim N^{-1/2}$, 
we make the ansatz that 
\beq
	a= \hsigma \cdot \vu_1 \simeqidsgibbs \sqrt{1-T} \Bern(\varphi)
\eeq
for some $\varphi$ which we determine now. 
Note that if $X_1$ and $X_2$ are independent (thermal) random variables distributed as $\gib B(\varphi)$, then their product $X_1X_2$ is $\gib B(\varphi^2 +(1-\varphi)^2)$-distributed. 
Thus, the equations \eqref{eq:b12c} and \eqref{eq:vv12} become
\beqq
	\frac{|n_1|\sqrt{1-T} \Bern(\varphi)}{\sqrt{N}} + b \simeqidsgibbs h + \frac{|n_1| \sqrt{1-T} \gib B(\alpha)}{\sqrt{N}} + \frac{\sqrt{T} \gib N}{\sqrt{N}}
\eeqq
and
\beqq
	(1-T) \Bern(\varphi^2+(1-\varphi)^2) + \bhp{N^{-1}}+ \mathbf{v}^{(1)} \cdot \mathbf{v}^{(2)} \simeqidsgibbs (1-T) \gib B(\theta).
\eeqq
Since $\theta= \alpha^2 +(1-\alpha)^2$, it is reasonable to assume that the solutions are $\varphi=\alpha$,  and
\beqq
	a \simeqidsgibbs \sqrt{1-T} \Bern(\alpha), \qquad  b \simeqidsgibbs h + \frac{\sqrt{T} \gib N}{\sqrt{N}}.
\eeqq
This calculation leads us to the following conjecture on the signed overlap of the spin variable with a replica. 

\begin{conjecture} \label{conj:signedoverlap}
For a given disorder sample, let $\vu_1$ be the unit vector corresponding to the ground state such that $\vu_1\cdot \mathbf{g}\ge 0$. 
Then, for $h=H N^{-1/2}$ and $0<T<1$, the signed overlap with the ground state satisfies 
\beq \label{eq:signedoverlap}
	\frac{\spin\cdot \vu_1}{\sqrt{N}} \simeqidsgibbs \sqrt{1-T}  \Bern(\alpha), \qquad 	\alpha= \frac{e^{ \frac{ H |n_1| \sqrt{1-T} }{T} }}{e^{ \frac{ H |n_1|\sqrt{1-T} }{T} }+e^{-\frac{ H |n_1| \sqrt{1-T} }{T}}} ,
\eeq
as $N\to \infty$ for asymptotically almost every disorder sample. 
\end{conjecture}

The above conjecture implies that for $h=H N^{-1/2}$ the spin configuration vector concentrates on the intersection of the sphere and the double cone around $\vu_1$ 
where the cosine of the angle is $\sqrt{1-T}$, just like the $h=0$ case. However, while for $H=0$ the spin vector is equally likely to be on either of the cones, for $H>0$ the spin prefers the cone that is closer to $\efv$ than the other cone. As $H\to \infty$, the polarization parameter $\alpha\to 1$ and hence for $h\gg N^{-1/2}$, the spin vector is concentrated on one of the cones. 



\subsection{Spin decompositions in various regimes} 

The results of the overlaps give us information about the decomposition of the spin for other regimes of $h$ as well. 
From the first equation of \eqref{eq:abfrov}, we find $a^2$, and hence $|a|$. 
The discussion of the previous subsection implies that for $h\gg N^{-1/2}$, the spin vector concentrates on one of the cones. 
Thus, 
we expect that $a=|a|$ for such $h$. 
Using this formula of $a$, we then obtain $b$ from the second equation of \eqref{eq:abfrov}, from which we also find $\|\mv\|^2=1-a^2-b^2$. 
Finally, the equation \eqref{eq:cfrov} implies $\mv^{(1)}\cdot \mv^{(2)}$, and hence, the overlap $\hmv^{(1)} \cdot \hmv^{(2)}$ of the unit transversal vector $\hmv= \frac{\mv}{\|\mv\|}$ with its replica.  
We summarize the findings  in Table \ref{table:geometrysummary}. The result for the last row follows from the last subsection. 

\begin{table}
\renewcommand{\arraystretch}{1.9}
\centering
\begin{tabular}{l|l|l|l|l}
Case & $a=\hsigma \cdot \vu_1$ & $b\simeq \hsigma \cdot \hefv$& $\|\mv\|$ & $\hmv^{(1)} \cdot \hmv^{(2)}$    \\
\hline
$h\to\infty$  & $0$ & $1$ & 0 & 0 \\
$h=O(1)$  & 
$\frac{\sqrt{\OMz}}{\sqrt{N}}$ 
&  $hs_1 (\cp_0) $ & $\sqrt{1-h^2s_1(\cp_0)^2}$ & $\frac{h^2 s_1(\cp_0)^4}{(1-s_1(\cp_0)^2)(1-h^2 s_1(\cp_0)^2)}$ \\
$h\to 0$, $hN^{\frac16}\to\infty$ & $\frac{4(1-T)^2 |n_1|}{h^3\sqrt{N} }$ & $h$ &  $1$ 
 & $1-T$ \\
$h\sim N^{-\frac16}$ & $ \mathcal A(T, hN^{1/6}) $ & $h$ & $\sqrt{1-\mathcal A^2}$ & $ \frac{1-T - \mathcal A^2}{1-\mathcal A^2}$ \\
$hN^{\frac16}\to 0$, $hN^{\frac12}\to\infty$ & $\sqrt{1-T}$ & $h$& $ \sqrt{T} $ &  $o(1)$ \\ 
$h\sim N^{-\frac12}$ (and $h=0$) & $ \sqrt{1-T}  \Bern(\alpha)$ & $h+ \frac{\sqrt{T} \gib N}{\sqrt{N}}$ & $\sqrt{T}$ & $ o(1)$  \\
\end{tabular}
\caption{This table summarized the findings of the decomposition of the spin variable $\hsigma \simeq a\vu_1 + b \hefv+ \mv$ in different regimes for $0<T<1$. We indicate the leading order terms, except that we have $o(1)$ at two places. The $o(1)$ term in the fifth row is complicated to state and the $o(1)$ term in the last row is not determined from our analysis. The unit transversal vector is $\hmv=\frac{\mv}{\|\mv\|}$.}
\label{table:geometrysummary}
\end{table}

The result for the regime $h\sim N^{-1/6}$ (fourth row) follows from Results \ref{thm:ground1/6}, \ref{thm:ext1/6}, and \ref{result:replica16}. 
The term $\mathcal A=\mathcal A(T, hN^{1/6})$ is given by the leading term in Result \ref{thm:ground1/6}, 
\beq \label{eq:mathcaA161}
	\mathcal A= \sqrt{1- T-  h^2N^{1/3}  \sum_{i = 2}^N \frac{n_i^2}{(\stild + \egres_1 - \egres_i)^2} } = \frac{hN^{1/6}|n_1|}{\stild}, 
\eeq
where $\stild>0$ is the number that makes the two formulas of $\mathcal A$ equal. 
For every disorder sample, $\mathcal A$ is a decreasing function of $H=hN^{1/6}$, changing from $\sqrt{1-T}$ for $H=0$ to $0$ as $H\to \infty$. 

The result for the regime $h=O(1)$ (second row) follows from Result \ref{thm:groundh>0}, \ref{thm:exth>0}, and \ref{result:replica1}. 
The variable $\cp_0=\cp_0(T, h)>2$ is the solution of the equation \eqref{eq:apprxcpeqhp}. 
It satisfies $\cp_0\simeq  h + \frac{T}{2}$ as $h\to \infty$ and $\cp_0\simeq 2 + \frac{h^4}{4(1-T)^2}$ as $h\to 0$: See Lemma \ref{lem:cp0behro}.
The function $s_1(z)$ is the Stieltjes transform of the semicircle law. It satisfies $s_1(z) = z^{-1} + O(z^{-3})$ as $z\to \infty$ and $s_1(z)\simeq 1-\sqrt{z-2}$ as $z\to 2$: see \eqref{eq:stjaspt}. 
See Sub-subsection \ref{sec:mgnhplt} for properties of $\mgn^0= hs_1(\cp_0)$. 
The term $\frac{\sqrt{\OMz}}{\sqrt{N}}$ is from Result \ref{thm:groundh>0} and is given by 
\beq
	\frac{\sqrt{\OMz}}{\sqrt{N}} =  \frac1{\sqrt{N}}\left \lvert \frac{\ef |n_1|}{\cp_0 - 2}+  \frac{\sqrt{T} \gib N}{\sqrt{\cp_0-2}} \right \rvert. 
\eeq
For the last column, Result \ref{result:replica1} and the formula $b\simeq hs_1(\cp_0)$ imply that $\mv^{(1)}\cdot \mv^{(2)}\simeq h^2 s_2(\cp_0)-h^2s_1(\cp_0)^2$. We use the identity $s_2(z)= s_1(z)^2/(1-s_1(z)^2)$ for $z>0$ to simplify the formula.

The third row follows either from the fourth row or from the second row. Starting from the fourth row, we use \eqref{eq:omh1/6Hlarge}, which shows that
\beq \label{eq:mathAh16lar}
	\mathcal A^2 \simeq \frac{16(1-T)^4 n_1^2}{h^6N} 
\eeq
as $hN^{1/6} \to \infty$. We can also see this formula from \eqref{eq:mathcaA161} because $\stild \simeq \frac{h^4N^{2/3}}{4(1-T)^2}$  (see \eqref{eq:sqrs0whenehf}). Note that $\mathcal A=o(1)$ in this regime. 
On the other hand, if we start from the second row, we use \eqref{eq:sqrom_h1/6_h0} to find the same formula for $a$. Other columns can be found from $s_1(\cp_0)\simeq 1 - \frac{h^2}{2(1-T)}$ as $h\to 0$. 
Note that the two components $a$ and $b$ are comparable in size for $h\sim N^{-1/8}$.


The quantity $a$ in the fifth row follows either from the fourth row or from the last row. The formula \eqref{eq:heq16limHla} shows that 
$\mathcal A^2\simeq 1-T$ as $hN^{1/6}\to 0$. We also see this formula from \eqref{eq:mathcaA161} by dropping the $o(1)$ term. If we start from the last row, the polarization parameter $\alpha$  satisfies $\alpha\to 1$ as $hN^{1/2}\to \infty$, and hence $a\simeq \sqrt{1-T}$, giving the same formula for $a$. The other columns follow from this result. 
One can show using Result \ref{thm:ground1/3} and \eqref{eq:13tempo} that the subleading term in $a$ (not shown in Table \ref{table:geometrysummary}) is comparable to the leading term of $b$, which is $h$, when $h\sim N^{-1/3}$.

\subsection{Summary}\label{sec:GeometrySummary}

Three quantities contain thermal random variables: $a$ for  the regimes $h=O(1)$ and $h\sim N^{-1/2}$, and $b$ for the regime $h\sim N^{-1/2}$. 
Among those, $a$ for  the regime $h\sim N^{-1/2}$ is $\bhp{1}$ but the other two quantities are of smaller order $\bhp{N^{-1/2}}$.


The table shows that $a=O(1)$ for $h\le O(N^{-1/6})$ and $b=O(1)$ for $h\ge O(1)$. 
As $h$ increases, the $\vu_1$ component of a typical spin vector decreases while the $\hefv$ component increases. 
The above result shows that the crossover occurs in the regime $N^{-1/6}\ll h\ll O(1)$ in which both components are $o(1)$. 

 
The last column of the table is the overlap of the unit  transversal vector $\hmv$ with its replica. 
This overlap is $o(1)$ for $h\ll N^{-1/6}$. 
If the error were $O(N^{-1/2})$, it would give a strong indication that the thermal distribution of $\hmv$ is uniform on the transverse space (i.e. the set of unit vectors that are perpendicular to $\vu_1$ and $\efv$). 
The above result does not show the error, but we expect that the distribution on the transverse space is close to being uniform. 
On the other hand, for $h\ge O(N^{-1/6})$, the overlap of the unit transversal vector is non-zero and $\bhp{1}$. 
This implies that $\hmv$ is not uniformly distributed on the transverse space. 

Overall, for $0<T<1$, as we increase the external field, we expect the following geometry of the spin vector that is randomly chosen using the Gibbs (thermal) measure for a quenched disorder, i.e. for asymptotically almost every disorder sample. 
\begin{itemize}
\item For $h\ll N^{-1/6}$, the spin vector is on a double cone around $\vu_1$ (possibly preferring one cone to the other), and the thermal  distribution on the transverse space is close to being uniform.
\item For $h\sim N^{-1/6}$, the spin vector is polarized to a single cone around $\vu_1$, but the cone itself depends non-trivially on the disorder sample. The thermal distribution on the transverse space is not uniform and 
depends on the disorder sample.
\item For $N^{-1/6}\ll h\ll O(1)$, the spin vector entirely lies on the transverse space with only $o(1)$ components on the ground state and external field directions. Although the thermal distribution is not uniform, it does not depend on the disorder sample. 
\item For $h=O(1)$, the spin vector is on a cone around $\efv$ and the thermal distribution on the transverse space is not uniform. 
The cone and the distribution on the transverse space do not depend on the disorder sample. 
\item For $h\to \infty$, the spin vector is parallel to $\efv$. 
\end{itemize}

The result of this paper does not describe the distribution of $\hmv$ on the transverse space in detail. This can be achieved by studying the overlaps $\sigma\cdot \vu_i$ with other eigenvectors. This analysis can be done using the method of this paper and we leave this work as a future project. 

\medskip

The items in the table can be written in an uniform formula across all regimes as the following decomposition formula of the spin configuration vector: 
\beq
	\hsigma  \simeqidsgibbs\mathcal A \gib B(\alpha) \vu_1 + hs_1(\cp_0) \hefv + \sqrt{1- \mathcal A^2- h^2s_1(\cp_0)^2} \hmv
	+ \bhp{N^{-1/2}}
\eeq
where $\hmv$ is a unit vector in the transverse space, i.e. $\hmv\cdot \vu_1=\hmv\cdot \hefv=0$ and $\|\hmv\|=1$. 
All items in the middle three columns of the table other than two items, $a$ for the regime $h=O(1)$ and $b$ for the regime  $h\sim N^{-1/2}$,  are of order greater than $\bhp{N^{-1/2}}$. 
Hence, the above formula is meaningful for all items except those two. 

\begin{appendices} 

\section{Proof of Lemma \ref{lem:contour}}
\label{sec:integraloverlapproof}


We prove Lemma \ref{lem:contour}. 
First, 
\beqq
	\langle e^{\beta \eta \mgn}\rangle = \frac1{\pat_N(\ef)} \int_{S_{N - 1}} e^{\beta\frac{ \eta}{N} \efv\cdot \sphv}  e^{\beta \left( \frac 12 \sphv\cdot \sGOE \sphv  +h \efv\cdot \sphv \right)} \dd \omega_N(\sphv) = \frac{\pat_N(\ef + \eta N^{-1})}{\pat_N(\ef)}. 
\eeqq
Secondly, by definition, 
\beq \label{eq:appx2ovm}
	\langle e^{\beta \eta \OM} \rangle = \frac{1}{\pat_N} \int _{S_{N - 1}} e^{\beta \frac{\eta}{N} (\vu_1 \cdot \sigma)^2} e^{\beta \left( \sphv\cdot \sGOE \sphv  +h \efv\cdot \sphv  \right) } \dd \omega_N(\sphv).
\eeq
Since
\beqq
	\frac 12 \sphv\cdot \sGOE \sphv + \frac{\eta}{N} (\vu_1\cdot \sigma)^2=
	\frac12 \sum_{i=1}^N \eg_i (\vu_i\cdot \sigma)^2+ \frac{\eta}{N} (\vu_1\cdot \sigma)^2,
\eeqq
the integral in \eqref{eq:appx2ovm} is the same as that of $\pat_N$ with $\eg_1\mapsto \eg_1+ \frac{2\eta}{N}$. 
Finally, using the eigenvalue-eigenvector decomposition $M=O\Lambda O^T$ and changing variables $\frac1{\sqrt{N}} O^T\sigma =x$ and $\frac1{\sqrt{N}} O^T \tau =y$, we find that
\beq	\label{eq:mgnovlratio}
	\langle e^{\eta \ovl} \rangle = \frac{J(\frac{\beta N}{2}, \frac{\beta N}{2}; \frac{\eta}{N \beta}, \frac{\sqrt{\beta}\ef}{\sqrt{2}}, \frac{\sqrt{\beta}\ef}{\sqrt{2}})}{J(\frac{\beta N}{2}, \frac{\beta N}{2}; 0, \frac{\sqrt{\beta}\ef}{\sqrt{2}}, \frac{\sqrt{\beta}\ef}{\sqrt{2}})}.
\eeq
where we use the notation
\beqq
	J(u,v; a,b,c) = (uv)^{\frac{N}{2} - 1}\displaystyle{\int \int}  e^{ 2a\sqrt{uv}\sum\limits_{i = 1}^N x_i y_i  +u \sum\limits_{i=1}^N \eg_i x_i^2 + 2b\sqrt{u} \sum\limits_{i=1}^N n_i x_i  + v \sum\limits_{i=1}^N \eg_i y_i^2 + 2c\sqrt{v} \sum\limits_{i=1}^N n_i y_i } \dd \Omega^{\otimes 2}_{N-1}(x, y) . 
\eeqq
We evaluate the Laplace transform of $J(u,v,a,b,c)$. Changing of variable as $u = r^2$, $v = s^2$ and $rx \mapsto x$, $sy \mapsto y$, 
the Laplace transform 
\beqq
	Q(z,w)  = \int_0^\infty  \int_0^\infty e^{-zu - wv} J(u,v) \dd u \dd v 
\eeqq
becomes a 2-dimensional Gaussian integral which evaluates to   
\beqq
	Q(z,w)  = 4 \prod\limits_{i = 1}^N \frac{\pi}{\sqrt{(z - \eg_i)(w - \eg_i) - a^2}} 
	e^{\frac{n_i^2((w - \eg_i)b^2 + 2abc + (z - \eg_i) c^2}{(z - \eg_i)(w - \eg_i) - a^2}}.
\eeqq
The inverse Laplace transform gives a double integral formula for $J(u,v)$.

\section{A perturbation argument} \label{sec:pert}

The following perturbation lemma is used to obtain \eqref{eq:Gcp_h0}, \eqref{eq:Gcph>0} and \eqref{eq:GcphNhalf}.

\begin{lemma} \label{lemma:Gcp} 
Let $I$ be a closed interval of $\R$. 
Let  $G(z;N)$ be a sequence of random  $C^{4}$-functions for $z \in I$. 
Let $\epsilon = \epsilon(N) := N^{-\delta}$ for some $\delta > 0$ and assume that 
\beq
	G(z; N)= G_0(z; N) + G_1(z; N) \epsilon + G_2(z; N) \epsilon^2 + \bhp{\epsilon^3}
\eeq
and
\beq
	G'(z; N) =  G'_0(z; N) + G'_1(z; N) \epsilon + G'_2(z; N) \epsilon^2 + \bhp{\epsilon^3}
\eeq
for random $C^{4}$-functions $G_k(z; N)$.
Suppose that 
\beq
	G_k^{(\ell)}(z;N) = \bhp{1}
\eeq
uniformly for $z\in I$ for all $k=0,1,2$, $0\le \ell\le 4$ and also assume that there is a $\cp_0 \in I$ satisfying
\beq \label{eq:gamma0zand1}
	G_0'(\cp_0;N)=0, \qquad \lvert G_0''(\cp_0;N) \rvert \geq C > 0
\eeq
for a positive constant $C$.
Then there is a critical point $\cp=\cp(N)$ of $G(z;N)$ admitting the asymptotic expansion
\beq  \label{eq:cpexp}
	\cp= \gamma_0+ \gamma_1 \epsilon +  \cp_2 \epsilon^2 + \bhp{\epsilon^3}
\eeq
where 
\beq  \label{eq:cp1cp2}
	\cp_1 = -\frac{G_1'(\cp_0;N)}{G_0''(\cp_0;N)}, \qquad \cp_2 = -\frac{G_2'(\cp_0;N) + G_1''(\cp_0;N)\cp_1 + \frac12 G_0'''(\cp_0;N)\cp_1^2}{G_0''(\cp_0;N)}.
\eeq
Furthermore, 
\beq 
\label{eq:Gexpan}
\begin{split}
	G(\cp;N) = G_0(\cp_0;N) + G_1(\cp_0;N) \epsilon 
	+ \left( \frac12 G_1'(\cp_0;N)\cp_1 + G_2(\cp_0;N) \right) \epsilon^2 + \bhp{\epsilon^3}.
\end{split} 
\eeq
\end{lemma}

\begin{proof}
This lemma is standard when $G(z;N)$ is deterministic. The proof for the random $G(z;N)$ does not change. For simplicity, we suppress the dependence on $N$ in the notations; for example we write $G_0(z)$ instead of $G_0(z;N)$. 
In order to prove \eqref{eq:cpexp}, it is enough to show that for any $0 < t < \delta$, $G'(\cp_+)G'(\cp_-) < 0$ with $\cp_\pm = \cp_0 + \cp_1 \epsilon + \cp_2 \epsilon^2 \pm \epsilon^3N^t$.
From the Taylor expansion, 
\beq \begin{split}
	G'(\cp_\pm) & = G'_0(\cp_0) + (G_0''(\cp)\cp_1 + G_1'(\cp_0)) \epsilon \\
& \quad + \left(G_0''(\cp_0)\cp_2 + G_2'(\cp_0) + G_1''(\cp_0)\cp_1 + \frac12 G_0'''(\cp_0)\cp_1^2 \right)\epsilon^2 \pm G''_0(\cp_0) \epsilon^3 N^t + \bhp{\epsilon^3}.
\end{split} \eeq 
The definitions of $\cp_0, \cp_1$, and $\cp_2$ imply that 
\beq \begin{split}
	G'(\cp_\pm) & =  \pm G''_0(\cp_0) \epsilon^3 N^t + \bhp{\epsilon^3}
\end{split} \eeq 
Thus, $G'(\cp_+)G'(\cp_-) <0$ for all large enough $N$ and  we obtain \eqref{eq:cpexp}. 
The equation \eqref{eq:Gexpan} follows from 
\beq \begin{split}
	G(\cp) &= G_0(\cp) + G_1(\cp) \epsilon + G_2(\cp) \epsilon^2 + \bhp{\epsilon^3} = G_0(\cp_0) + (G_0'(\cp_0)\cp_1 + G_1(\cp_0)) \epsilon \\
	& + \left(G_0'(\cp_0)\cp_2 + \frac12 G_0''(\cp_0)\cp_1^2 + G_1'(\cp_0)\cp_1 + G_2(\cp_0) \right) \epsilon^2 + \bhp{\epsilon^3},
\end{split} \eeq
together with $G'_0(\cp_0) = 0$ and \eqref{eq:cp1cp2}.
\end{proof}

\begin{remark}
Here, we consider the asymptotic expansion of $G(z)$ up to the third order term. One can also consider the case where the expansion is up to the second order, then \eqref{eq:Gexpan} is still valid up to the second order. 
\end{remark}

\end{appendices}



\end{document}